%% file: main.tex
\documentclass[11pt]{ociamthesis}  

\usepackage{amsmath} 
\usepackage{amsthm} 
\usepackage{amssymb} 
\usepackage{bold-extra} 
\usepackage{hyperref} 
\usepackage{stmaryrd} 
\usepackage{doi} 

\include{my_macros}
\include{tikzit_header2} 

\title{Completeness and the \ZX-calculus}   

\author{Miriam K. Backens} 
\college{Merton College}   

\degree{Doctor of Philosophy}     
\degreedate{Trinity 2015}         

\begin{document}

\baselineskip=18pt plus1pt

\setcounter{secnumdepth}{3}
\setcounter{tocdepth}{3}

\maketitle                  
\include{acknowledgements}  
\include{abstract}          

\begin{romanpages}          
\tableofcontents            
\end{romanpages}            

\include{introduction}
\include{graphical}

\include{ZX-calculus}
\include{completeness}

\include{more_completeness}
\include{Spekkens}
\include{conclusions}


\addcontentsline{toc}{chapter}{Bibliography}
\bibliography{refs}        
\bibliographystyle{plain}  

\end{document}

%% file: my_macros.tex
\newenvironment{menum}
{\begin{enumerate}
  \setlength{\itemsep}{1pt}
  \setlength{\parskip}{0pt}
  \setlength{\parsep}{0pt}}
{\end{enumerate}}

\newenvironment{mitem}
{\begin{itemize}
  \setlength{\itemsep}{1pt}
  \setlength{\parskip}{0pt}
  \setlength{\parsep}{0pt}}
{\end{itemize}}

\newtheorem{thm}{Theorem}[section]
\newtheorem*{thm*}{Theorem}
\newtheorem{prop}[thm]{Proposition}
\newtheorem*{prop*}{Proposition}
\newtheorem{lem}[thm]{Lemma}
\newtheorem{cor}[thm]{Corollary}
\theoremstyle{definition}
\newtheorem{dfn}[thm]{Definition}
\newtheorem{ex}[thm]{Example}
\theoremstyle{remark}
\newtheorem*{remark}{Remark}

\newcommand{\ZX}{\textsc{zx}}
\newcommand{\NOT}{\textsc{not}}
\newcommand{\SWAP}{\textsc{swap}}

\newcommand{\abs}[1]{\left| #1 \right|}
\newcommand{\avg}[1]{\left\langle #1 \right\rangle}
\newcommand{\CC}{\mathbb{C}}

\newcommand{\ZZ}{\mathbb{Z}}
\newcommand{\HH}{\mathcal{H}} 

\newcommand{\ket}[1]{\left| #1 \right>} 
\newcommand{\bra}[1]{\left< #1 \right|} 

\newcommand{\intf}[1]{\left\llbracket #1 \right\rrbracket} 

\newcommand{\FHilb}{\mathbf{FHilb}}
\newcommand{\FRel}{\mathbf{FRel}}
\newcommand{\Hilb}{\mathbf{Hilb}}
\newcommand{\Mat}{\mathbf{Mat}}
\newcommand{\Rel}{\mathbf{Rel}}
\newcommand{\Set}{\mathbf{Set}}

\newcommand{\Spek}{\mathbf{Spek}}

\renewcommand{\t}[1]{\ensuremath{^{\otimes #1}}}

%% file: tikzit_header2.tex
\usepackage[svgnames]{xcolor} 

\usepackage{tikz}
\usetikzlibrary{shapes.geometric} 
\usetikzlibrary{shapes.misc} 
\usetikzlibrary{shapes.symbols} 
\usetikzlibrary{decorations.pathreplacing} 
\usetikzlibrary{decorations.markings} 
\usetikzlibrary{arrows} 

\pgfdeclarelayer{edgelayer}
\pgfdeclarelayer{nodelayer}
\pgfsetlayers{edgelayer,nodelayer,main}

\tikzstyle{none}=[inner sep=0pt]
\tikzstyle{rn}=[circle,fill=Red,draw=Black,line width=0.8 pt,minimum size=5pt,inner sep=0pt]
\tikzstyle{gn}=[circle,fill=Lime!40,draw=Black,line width=0.8 pt,minimum size=5pt,inner sep=0pt]
\tikzstyle{Hadamard}=[rectangle,fill=Yellow,draw=Black,minimum size=8pt,inner sep=0pt,label={center:$\scriptstyle\mathrm{H}$}]
\tikzstyle{gphase}=[rounded rectangle,rounded rectangle arc length=90,fill=Lime!15,inner sep=2pt]
\tikzstyle{rphase}=[rounded rectangle,rounded rectangle arc length=90,fill=Red!25,inner sep=2pt]
\tikzstyle{HadSpek}=[rectangle,fill=Yellow,draw=Black,minimum size=6pt]

\tikzstyle{normalrect}=[rectangle,fill=white,draw=black,minimum height=10pt,minimum width=12pt,inner sep=0pt]
\tikzstyle{unitary}=[rectangle,fill=white,draw=black,minimum height=14pt,minimum width=14pt,inner sep=0pt]
\tikzstyle{wideunitary}=[rectangle,fill=white,draw=black,minimum height=14pt,minimum width=24pt,inner sep=0pt]
\tikzstyle{graph}=[ellipse,fill=White,draw=Black,minimum height=0.5cm,minimum width=3cm,inner sep=0pt]
\tikzstyle{biggraph}=[ellipse,fill=White,draw=Black,minimum height=0.75cm,minimum width=6.5cm,inner sep=0pt]
\tikzstyle{scalardiamond}=[diamond,draw=black,fill=white,inner sep=1pt,minimum size=0.4cm]
\tikzstyle{dtr}=[regular polygon,regular polygon sides=3,shape border rotate=180,fill=white,draw=black,minimum size=22pt,inner sep=0pt]
\tikzstyle{utr}=[regular polygon,regular polygon sides=3,fill=white,draw=black,minimum size=22pt,inner sep=0pt]

\tikzstyle{bigcloud}=[cloud,fill=white,draw=black]
\tikzstyle{bscalar}=[star,fill=Black,draw=Black,minimum size=8pt,inner sep=0pt]
\tikzstyle{diagram}=[rectangle,fill=white,draw=black,minimum height=0.5cm,minimum width=1cm,inner sep=0pt]

\tikzstyle{squ}=[rectangle,fill=white,draw=black,minimum height=0.5cm,minimum width=0.6cm,inner sep=0pt]
\tikzstyle{twqu}=[rectangle,fill=white,draw=black,minimum height=1.25cm,minimum width=0.6cm,inner sep=0pt]
\tikzstyle{thrqu}=[rectangle,fill=white,draw=black,minimum height=2cm,minimum width=0.6cm,inner sep=0pt]
\tikzstyle{bn}=[circle,fill=black,draw=black,inner sep=0pt,minimum size=4pt]
\tikzstyle{whn}=[circle,fill=none,draw=black,minimum size=8pt]
\tikzstyle{xn}=[circle,fill=none,draw=none,minimum size=8pt]

\tikzstyle{solidn}=[circle,fill=Black,draw=Black,minimum size=7pt,inner sep=0pt]
\tikzstyle{hollown}=[circle,fill=White,draw=Black,line width=0.8 pt,minimum size=7pt,inner sep=0pt]

\tikzstyle{map}=[trapezium,trapezium left angle=90,trapezium right angle=120,fill=White,draw=Black,inner sep=0pt, minimum height=0.5cm]
\tikzstyle{daggermap}=[trapezium,trapezium left angle=90,trapezium right angle=60,fill=White,draw=Black,inner sep=0pt, minimum height=0.5cm]

\tikzstyle{roundedLeft}=[rounded rectangle,fill=white,draw=black,minimum width=1cm,minimum height=1.5cm,rounded rectangle east arc=none]

\tikzstyle{every picture}=[baseline=-0.1cm,scale=0.5]

\newcommand{\effect}[1]{\begin{tikzpicture}
	\begin{pgfonlayer}{nodelayer}
		\node [style=none] (0) at (0, -0.25) {};
		\node [style={#1}] (1) at (0, 0.25) {};
	\end{pgfonlayer}
	\begin{pgfonlayer}{edgelayer}
		\draw [] (0.center) to (1);
	\end{pgfonlayer}
\end{tikzpicture}}

\newcommand{\Hadamard}[0]{\begin{tikzpicture}
	\begin{pgfonlayer}{nodelayer}
		\node [style=Hadamard] (0) at (0, 0) {};
		\node [style=none] (1) at (0, 0.5) {};
		\node [style=none] (2) at (0, -0.5) {};
	\end{pgfonlayer}
	\begin{pgfonlayer}{edgelayer}
		\draw (2.center) to (1.center);
	\end{pgfonlayer}
\end{tikzpicture}}

\newcommand{\HadSpek}[0]{\begin{tikzpicture}
	\begin{pgfonlayer}{nodelayer}
		\node [style=HadSpek] (0) at (0, 0) {};
		\node [style=none] (1) at (0, 0.5) {};
		\node [style=none] (2) at (0, -0.5) {};
	\end{pgfonlayer}
	\begin{pgfonlayer}{edgelayer}
		\draw (2.center) to (1.center);
	\end{pgfonlayer}
\end{tikzpicture}}

\newcommand{\halfscalar}[0]{\begin{tikzpicture}
	\begin{pgfonlayer}{nodelayer}
		\node [style=bscalar] (1) at (0, 0) {};
	\end{pgfonlayer}
\end{tikzpicture}}

\newcommand{\innerprod}[2]{\begin{tikzpicture}
	\begin{pgfonlayer}{nodelayer}
		\node [style={#1}] (0) at (0, 0.6) {};
		\node [style={#2}] (1) at (0, -0.6) {};
	\end{pgfonlayer}
	\begin{pgfonlayer}{edgelayer}
		\draw (0) to (1);
	\end{pgfonlayer}
\end{tikzpicture}}

\newcommand{\innerprodgr}[0]{\begin{tikzpicture}
	\begin{pgfonlayer}{nodelayer}
		\node [style=gn] (0) at (0, 0.4) {};
		\node [style=rn] (1) at (0, -0.15) {};
	\end{pgfonlayer}
	\begin{pgfonlayer}{edgelayer}
		\draw (0) to (1);
	\end{pgfonlayer}
\end{tikzpicture}}

\newcommand{\joinnode}[0]{\begin{tikzpicture}
	\begin{pgfonlayer}{nodelayer}
		\node [style=gn] (0) at (0, -0) {};
		\node [style=none] (1) at (0, 0.5) {};
		\node [style=none] (2) at (0.5, -0.5) {};
		\node [style=none] (3) at (-0.5, -0.5) {};
	\end{pgfonlayer}
	\begin{pgfonlayer}{edgelayer}
		\draw [bend right=15] (2.center) to (0);
		\draw [bend left=15] (3.center) to (0);
		\draw (0) to (1.center);
	\end{pgfonlayer}
\end{tikzpicture}}

\newcommand{\phase}[1]{\begin{tikzpicture}
	\begin{pgfonlayer}{nodelayer}
		\node [style=none] (0) at (0, 0.4) {};
		\node [style={#1}] (1) at (0, -0) {};
		\node [style=none] (2) at (0, -0.4) {};
	\end{pgfonlayer}
	\begin{pgfonlayer}{edgelayer}
		\draw [] (0.center) to (1);
		\draw [] (1) to (2.center);
	\end{pgfonlayer}
\end{tikzpicture}}

\newcommand{\scalar}[1]{\begin{tikzpicture}
	\begin{pgfonlayer}{nodelayer}
		\node [style={#1}] (1) at (0, 0) {};
	\end{pgfonlayer}
\end{tikzpicture}}

\newcommand{\splitnode}[0]{\begin{tikzpicture}
	\begin{pgfonlayer}{nodelayer}
		\node [style=gn] (0) at (0, -0) {};
		\node [style=none] (1) at (0, -0.5) {};
		\node [style=none] (2) at (-0.5, 0.5) {};
		\node [style=none] (3) at (0.5, 0.5) {};
	\end{pgfonlayer}
	\begin{pgfonlayer}{edgelayer}
		\draw [bend right=15] (2.center) to (0);
		\draw [bend left=15] (3.center) to (0);
		\draw (0) to (1.center);
	\end{pgfonlayer}
\end{tikzpicture}}

\newcommand{\state}[1]{\begin{tikzpicture}
	\begin{pgfonlayer}{nodelayer}
		\node [style={#1}] (0) at (0, 0) {};
		\node [style=none] (1) at (0, 0.5) {};
	\end{pgfonlayer}
	\begin{pgfonlayer}{edgelayer}
		\draw (0) to (1.center);
	\end{pgfonlayer}
\end{tikzpicture}}

\newcommand{\gendiagram}[1]{\begin{tikzpicture}
	\begin{pgfonlayer}{nodelayer}
		\node [style=diagram] (0) at (0, -0) {#1};
		\node [style=none] (1) at (-0.75, 1) {};
		\node [style=none] (2) at (0.75, 1) {};
		\node [style=none] (3) at (-0.75, -1) {};
		\node [style=none] (4) at (0.75, -1) {};
		\node [style=none] (5) at (0, 0.75) {$\ldots$};
		\node [style=none] (6) at (0, -0.75) {$\ldots$};
	\end{pgfonlayer}
	\begin{pgfonlayer}{edgelayer}
		\draw (1.center) to (3.center);
		\draw (2.center) to (4.center);
	\end{pgfonlayer}
\end{tikzpicture}}

\newcommand{\geneffect}[1]{\begin{tikzpicture}
	\begin{pgfonlayer}{nodelayer}
		\node [style=utr,label={center:#1}] (0) at (0, 0.4) {};
		\node [style=none] (1) at (0, -0.35) {};
	\end{pgfonlayer}
	\begin{pgfonlayer}{edgelayer}
		\draw (0) to (1.center);
	\end{pgfonlayer}
\end{tikzpicture}}

\newcommand{\gengraph}[1]{\begin{tikzpicture}
	\begin{pgfonlayer}{nodelayer}
		\node [style=none] (0) at (-1.5, 0.75) {};
		\node [style=none] (1) at (1.5, 0.75) {};
		\node [style=none] (2) at (-1.5, 0) {};
		\node [ellipse,fill=White,draw=Black,minimum width=2cm,minimum height=0.5cm,inner sep=0pt] (3) at (0, -0) {#1};
		\node [style=none] (4) at (1.5, 0) {};
		\node [style=none] (5) at (0, 0.65) {$\ldots$};
	\end{pgfonlayer}
	\begin{pgfonlayer}{edgelayer}
		\draw (0.center) to (2.center);
		\draw (1.center) to (4.center);
	\end{pgfonlayer}
\end{tikzpicture}}

\newcommand{\genscalar}[1]{\begin{tikzpicture}
	\begin{pgfonlayer}{nodelayer}
		\node [style=scalardiamond] (0) at (0, -0) {#1};
	\end{pgfonlayer}
\end{tikzpicture}}

\newcommand{\genstate}[1]{\begin{tikzpicture}
	\begin{pgfonlayer}{nodelayer}
		\node [style=dtr,label={center:#1}] (0) at (0, 0) {};
		\node [style=none] (1) at (0, 0.75) {};
	\end{pgfonlayer}
	\begin{pgfonlayer}{edgelayer}
		\draw (0) to (1.center);
	\end{pgfonlayer}
\end{tikzpicture}}

\newcommand{\genunitary}[1]{\begin{tikzpicture}
	\begin{pgfonlayer}{nodelayer}
		\node [style=unitary] (0) at (0, 0) {#1};
		\node [style=none] (1) at (0, 0.75) {};
		\node [style=none] (2) at (0, -0.75) {};
	\end{pgfonlayer}
	\begin{pgfonlayer}{edgelayer}
		\draw (2.center) to (1.center);
	\end{pgfonlayer}
\end{tikzpicture}}

\newlength{\Speklength}
\newcommand{\Spekcolour}[0]{teal}

\newcommand{\SpekOneSquare}[4]{\begin{tikzpicture}[fill=\Spekcolour, every path/.style={draw}]
  #1 (0,-\Speklength) rectangle +(\Speklength,\Speklength);
  #2 (\Speklength,-\Speklength) rectangle +(\Speklength,\Speklength);
  #3 (0,0) rectangle +(\Speklength,\Speklength);
  #4 (\Speklength,0) rectangle +(\Speklength,\Speklength);
\end{tikzpicture}}

\newcommand{\SpekTwoSquare}[1]{\begin{tikzpicture}[fill=\Spekcolour]
  #1%
  \foreach \x in {-2\Speklength,-\Speklength,0,\Speklength}
    \foreach \y in {-2\Speklength,-\Speklength,0,\Speklength}
      \draw (\x,\y) rectangle +(\Speklength,\Speklength);
\end{tikzpicture}}

\newcommand{\Spekfill}[2]{\fill (#1\Speklength,#2\Speklength) rectangle +(\Speklength,\Speklength);
}

\newcommand{\sqGate}[1]{\begin{tikzpicture}
	\begin{pgfonlayer}{nodelayer}
		\node [style=squ] (0) at (0, -0) {#1};
		\node [style=none] (1) at (-1, -0) {};
		\node [style=none] (2) at (1, -0) {};
	\end{pgfonlayer}
	\begin{pgfonlayer}{edgelayer}
		\draw (1.center) to (2.center);
	\end{pgfonlayer}
\end{tikzpicture}}

%% file: acknowledgements.tex
\begin{acknowledgements}
 Firstly, I would like to thank my supervisors, Samson Abramsky and Bob Coecke, for giving me the opportunity to do this research.
 I owe much gratitude to Dominic Horsman, whose feedback, advice, and encouragement have been invaluable.

 Thank you to Ross Duncan for bringing the question of \ZX-calculus completeness to my attention.
 I also wish to thank all the other people working on this topic and on related questions for many interesting discussions.

 Thanks to the administrative staff at the Department for Computer Science for always being helpful, whether with university bureaucracy or with the organisation of student conferences and other academic or social events.
 Thank you also to everyone involved with CoGS and OxWoCS -- my time at this department would not have been the same without you.

 Many thanks to my family for their support and encouragement.

 Finally, a big thank you to my friends for being there in good times as well as in hard ones.
 The last four years would have been a lot less fun without OUSFG.
 Special thanks to Lyndsey and John for letting me stay at their houses while writing up, to bridge the time until my move to Bristol.
\end{acknowledgements}

%% file: abstract.tex
\begin{abstract}
 Graphical languages offer intuitive and rigorous formalisms for quantum physics.
 They can be used to simplify expressions, derive equalities, and do computations.
 Yet in order to replace conventional formalisms, rigour alone is not sufficient: the new formalisms also need to have equivalent deductive power.
 This requirement is captured by the property of completeness, which means that any equality that can be derived using some standard formalism can also be derived graphically.

 In this thesis, I consider the \ZX-calculus, a graphical language for pure state qubit quantum mechanics.
 I show that it is complete for pure state stabilizer quantum mechanics, so any problem within this fragment of quantum theory can be fully analysed using graphical methods.
 This includes questions of central importance in areas such as error-correcting codes or measurement-based quantum computation.
 Furthermore, I show that the \ZX-calculus is complete for the single-qubit Clifford+T group, which is approximately universal: any single-qubit unitary can be approximated to arbitrary accuracy using only Clifford gates and the T-gate.
 In experimental realisations of quantum computers, operations have to be approximated using some such finite gate set.
 Therefore this result implies that a wide range of realistic scenarios in quantum computation can be analysed graphically without loss of deductive power.
 
 Lastly, I extend the use of rigorous graphical languages outside quantum theory to Spekkens' toy theory, a local hidden variable model that nevertheless exhibits some features commonly associated with quantum mechanics.
 The toy theory for the simplest possible underlying system closely resembles stabilizer quantum mechanics, which is non-local; it thus offers insights into the similarities and differences between classical and quantum theories.
 I develop a graphical calculus similar to the \ZX-calculus that fully describes Spekkens' toy theory, and show that it is complete.
 Hence, stabilizer quantum mechanics and Spekkens' toy theory can be fully analysed and compared using graphical formalisms.

 Intuitive graphical languages can replace conventional formalisms for the analysis of many questions in quantum computation and foundations without loss of mathematical rigour or deductive power.
\end{abstract}

%% file: introduction.tex
\chapter{Introduction}

The problems being investigated in quantum-theoretical research are getting increasingly complex.
As the experimental realisation of usefully-sized general-purpose quantum computers approaches, the focus of much theoretical research in quantum computation is on fault-tolerant computation schemes \cite{gottesman_introduction_2010,devitt_quantum_2013}, which need to deal with a large number of physical qubits to encode a reasonably-sized logical computation: a fault-tolerant implementation of Shor's algorithm \cite{shor_algorithms_1994} for factorising a 2048 bit number -- a typical size for an RSA key -- is likely to require billions of underlying physical qubits \cite{van_meter_blueprint_2013}.

While much progress has been made in the understanding of quantum information theory, computation, and foundations, the mathematical formalisms have not changed very much.
In classical computer science, the increasing complexity of problems and algorithms has led to the invention of increasingly abstract formalisms: for example, programming languages that are designed for ease of use by human programmers have almost completely replaced the old languages that closely followed the physical workings of the computing device.
This abstraction has two advantages: on the one hand, it makes writing code easier and less error-prone, and on the other hand, it makes code in modern programming languages more widely portable because the details of the implementation that vary from processor to processor are handled automatically.
Computer scientists have also invented a wide range of new formalisms for describing algorithms and problems, from the notion of abstract games \cite{apt_lectures_2011} to flow charts (originally introduced as process charts \cite{gilbreth_process_1921}).

A similar change is needed in quantum information theory.
If quantum computing is indeed more powerful than classical computation, the details of general operations on quantum systems will never be efficiently tractable using classical means.
Nevertheless, more intuitive and abstract formalisms can simplify the analysis of problems significantly.
From this perspective, matrix mechanics is like assembly language, one of the earliest programming languages: it is good for controlling all the details of a problem, but for complicated tasks, those details make the formalism error-prone and drown out the conceptual properties and high-level features that may be more relevant to a solution.

An important class of high-level formalisms are graphical languages, which consist of two-dimensional diagrams -- as opposed to algebraic notations, which are written as one-dimensional strings of symbols.
A range of such languages have been developed in the quantum computation, information, and foundations community.
The most widely known graphical language for quantum computation is quantum circuit notation, where qubits are drawn as wires and operations as boxes \cite{deutsch_quantum_1989,nielsen_quantum_2010}.
Many graphical languages are introduced informally, nevertheless it is possible to make them rigorous using category theory \cite{joyal_geometry_1991}.
This process involves defining a translation from diagrams to an algebraic language, and proving that different translations of the same diagram, or translations of different diagrams that nevertheless seem ``intuitively equal'', produce equivalent algebraic terms.
For example, in a quantum circuit diagram, gate symbols can ``slide along'' wires and the length of wires does not matter.
E.g., this diagram:
\begin{center}
 \input{tikz_files/circuit_wire_length1.tikz}
\end{center}
seems intuitively equal to this one:
\begin{center}
 \input{tikz_files/circuit_wire_length2.tikz}
\end{center}
for any single-qubit gates $U$ and $V$, and it is possible to make this equality rigorous.
Furthermore, the equivalence of the two diagrams is much more intuitively obvious than the corresponding algebraic equality:
\begin{equation}\label{eq:time_ordering}
 (U\otimes I)(I\otimes V) = (I\otimes V)(U\otimes I),
\end{equation}
where $I$ is the single-qubit identity operator.

There are also more specifically quantum-mechanical phenomena that can be represented particularly intuitively in graphical languages.
These require a move away from quantum circuit notation to richer graphical languages which represent states, measurement outcomes, and in particular entanglement in a coherent way.
Such graphical languages for quantum theory were introduced by Abramsky and Coecke \cite{coecke_logic_2003,coecke_logic_2004,abramsky_categorical_2004}.
They keep the notation of qubits as wires and unitary operators as boxes, but rotate it by $90^\circ$ so diagrams are read from bottom to top rather than left-to-right.
Inspired by the Dirac kets, a single-qubit state is denoted by a triangle with one wire coming out:
\begin{equation}
 \ket{\psi} \qquad \mapsto \qquad \genstate{$\psi$}.
\end{equation}
There is no wire going in because it is irrelevant what the qubit was doing before: this is the essence of state preparation.
Similarly, the outcome of a destructive single-qubit measurement is a triangle pointing the other way, with one wire going in:
\begin{equation}
 \bra{\phi} \qquad \mapsto \qquad \geneffect{$\phi$}.
\end{equation}
There is no wire going out because what comes after the measurement does not matter.
Two-qubit states can be represented by triangles with two wires coming out, and two-qubit measurement outcomes by triangles with two wires going in, and so on.
So the Bell state $\frac{1}{\sqrt{2}}(\ket{00}+\ket{11})$ could be denoted by a triangle with two wires coming out, but as this state is maximally entangled, it actually makes more sense to draw it as a curved wire, a ``cup'' \cite{abramsky_categorical_2004} (we drop the normalisation factor for consistency with later definitions):
\begin{equation}
 \ket{00}+\ket{11} \qquad \mapsto \qquad \input{tikz_files/cup_diagram.tikz}
\end{equation}
Then, ignoring normalisation, the quantum teleportation protocol \cite{bennett_teleporting_1993} is represented by the following diagram, where we have assumed for simplicity that no Pauli correction is necessary:
\begin{equation}
 \input{tikz_files/teleportation1.tikz}
\end{equation}
Alice holds the unknown state $\ket{\psi}$ and half of a Bell pair.
Bob holds the other half.
Alice performs a Bell-basis measurement on her two qubits with outcome $\frac{1}{\sqrt{2}}(\bra{00}+\bra{11})$, which is denoted by the ``cap'', or upside-down curved wire.
Now the proof that Bob ends up with the state $\ket{\psi}$ consists of straightening and then shortening the wire:
\begin{equation}
 \input{tikz_files/teleportation1.tikz} \qquad \mapsto \qquad \input{tikz_files/teleportation2.tikz} \qquad \mapsto \qquad \input{tikz_files/teleportation3.tikz}
\end{equation}
This is not just a way of informally illustrating the quantum teleportation protocol: the process of straightening and shortening wires is mathematically well-defined and rigorous \cite{abramsky_categorical_2004}.
Algebraic notations can therefore be replaced with more intuitive graphical languages without losing mathematical rigour.
It is also possible to derive complicated equalities entirely graphically.

When working with algebraic equations to solve a mathematical problem, these equations are transformed according to certain rules.
For example, adding the same thing to both sides of a true equality yields another true equality.
Another example is the rule that a part of an algebraic formula can be replaced with something equal to yield a new formula equal to the original one.
For example, consider the equation:
\begin{equation}\label{eq:rewriting_example}
 HZH = X,
\end{equation}
where $H$ is the Hadamard gate and $Z$ and $X$ are the respective Pauli gates.
As a consequence of this equality, whenever the term $HZH$ appears in an expression, it can be replaced with $X$, or conversely.
These ``cut and paste'' algebraic transformations are called \emph{rewriting}, and equalities like \eqref{eq:rewriting_example} are \emph{rewrite rules}.
Systems of such rewrite rules are analysed in the area of computer science called \emph{term rewriting} \cite{baader_term_1998}.
A similar approach can be taken with diagrams: specifying a set of basic diagram equalities as rewrite rules allows the derivation of more complicated diagram equations by cutting and pasting parts of diagrams.
That is \emph{graphical rewriting} \cite{heckel_graph_2006}.
For example, \eqref{eq:rewriting_example} can easily be turned into an equality between two quantum circuits:
\begin{equation}
 \sqGate{$H$}\!\sqGate{$Z$}\!\sqGate{$H$} \;  = \; \sqGate{$X$} \, ,
\end{equation}
which can then be used as a graphical rewrite rule.

The \ZX-calculus is a graphical language for pure state qubit quantum mechanics that allows the representation of states, measurement outcomes, and entanglement.
It was first introduced by Coecke and Duncan \cite{coecke_interacting_2008} and extended by Duncan and Perdrix \cite{duncan_graph_2009}.
In this graphical language, qubits are represented by wires and maps by labelled nodes.
The \ZX-calculus comes with a set of rewrite rules.
It has already been used to analyse a range of questions in quantum computation and quantum foundations, from quantum circuits \cite{coecke_interacting_2008}, via measurement-based quantum computations \cite{coecke_interacting_2008,duncan_rewriting_2010}, topological cluster-state computation \cite{horsman_quantum_2011}, quantum key distribution \cite{coecke_graphical_2011,hillebrand_superdense_2012}, and quantum secret sharing \cite{hillebrand_quantum_2011,zamdzhiev_abstract_2012}, to non-locality \cite{coecke_phase_2011}.

In order to replace other standard formalisms, a graphical language like the \ZX-calculus needs to have several important properties.
Firstly, it should be \emph{universal}, meaning that any process in the underlying theory can be represented graphically.
Secondly, the graphical language should be \emph{sound}, meaning that the rewrite rules allow only the derivation of true equalities.
This property is crucial: a new formalism is no good if it conflicts with the old one.
Thirdly, it should be \emph{complete}, meaning that the rewrite rules allow the derivation of \emph{all} true equalities.

Universality and soundness are straightforward to ensure, and indeed the \ZX-calculus is both universal and sound by construction \cite{coecke_interacting_2008,coecke_interacting_2011}.

In this thesis, I prove that the \ZX-calculus is complete for several important fragments of quantum theory, i.e.\ within these fragments, any true equality between \ZX-calculus diagrams can be derived using the rewrite rules.
Therefore, standard formalisms for those fragments of quantum theory can be replaced with the \ZX-calculus without any loss of deductive power.

The first \ZX-calculus completeness result in this thesis is for stabilizer quantum mechanics \cite{gottesman_stabilizer_1997}, a fragment of quantum theory that can be operationally described by restricting the allowed operations to preparations of computational basis states, computational basis measurements, and the Clifford group of unitaries.
Stabilizer quantum mechanics is of central importance in areas such as error-correcting codes \cite{nielsen_quantum_2010} or measurement-based quantum computation \cite{raussendorf_one-way_2001}.

I show that, using the \ZX-calculus rewrite rules, any stabilizer \ZX-calculus diagram can be brought into a normal form.
This normal form is not unique, but all equalities between normal form diagrams can be derived graphically.
As the rewrite rules of the \ZX-calculus are invertible, being able to bring any diagram into a normal form and being able to derive all equalities between normal form diagrams implies that all equalities between arbitrary diagrams can be derived.
Thus any question within pure state qubit stabilizer quantum mechanics can be analysed entirely using the intuitive graphical formalism.
This includes the derivation of equalities between operators as well as the computation of probabilities.

Furthermore, I show that the \ZX-calculus is also complete for the single-qubit Clifford+T group.
This group of operations is approximately universal, i.e.\ any single-qubit unitary can be approximated to arbitrary accuracy using just operations from the Clifford+T gate set \cite{boykin_universal_1999}.
The completeness proof for the single-qubit Clifford+T group is built around the definition of a normal form for such diagrams and the proof that it is unique.
As all the rewrite rules are invertible, the existence of a unique normal form immediately implies that all equalities between single-qubit Clifford+T operators can be derived from the rewrite rules of the \ZX-calculus.

In realistic implementations of quantum computers, particularly fault-tolerant ones, not all operations can be implemented directly \cite{nielsen_quantum_2010}.
Instead, general operations are approximated using gates from a finite set such as Clifford+T, e.g.\ using the Solovay-Kitaev algorithm \cite{dawson_solovay-kitaev_2006}.
Thus being able to derive all equalities within such an approximately universal group means that a wide range of realistic questions can be analysed graphically without loss of deductive power.
Work is ongoing to combine single-qubit Clifford+T completeness with stabilizer completeness into a completeness result for multi-qubit Clifford+T operators.

The final completeness result in this thesis extends the use of rigorous graphical languages outside quantum theory.
Toy models for quantum foundations are models that are described entirely using classical physics but which nevertheless exhibit many phenomena usually considered quantum.
They therefore offer insights into the similarities and differences between quantum and classical behaviour.
To gain these insights, it is useful to have similar formalisms for describing a toy model and its quantum-physical equivalent.
Here, I focus on Spekkens' toy bit theory \cite{spekkens_evidence_2007,spekkens_quasi-quantization_2014}, a toy model that is very similar to stabilizer quantum mechanics while being described in terms of local hidden variables.
Stabilizer quantum mechanics on the other hand is non-local: it is possible to violate Bell inequalities \cite{bell_einstein-podolsky-rosen_1964} using only stabilizer operations.

I construct a graphical language similar to the \ZX-calculus for the toy theory and give a set of sound rewrite rules for it.
Furthermore, I prove that this graphical language allows the derivation of all true equalities about the theory.
Therefore stabilizer quantum mechanics and Spekkens' toy bit theory can be fully analysed and compared using intuitive graphical methods.

The remainder of this thesis is structured as follows.

Chapter \ref{ch:graphical} contains an introduction to graphical languages for quantum theory, and how to make them rigorous.
Furthermore, the properties of soundness and completeness are rigorously defined.

The \ZX-calculus with its rewrite rules is introduced in detail in Chapter \ref{s:ZX-calculus}.
That chapter also contains an introduction to stabilizer quantum mechanics and the single-qubit Clifford+T group, together with standard formalisms for describing them, as well as their representations in the \ZX-calculus.

Chapter \ref{ch:completeness} starts with a recap of the proof that the full \ZX-calculus is incomplete.
Original work is contained in Section \ref{s:possible_completeness}, where it is shown that completeness results for restricted fragments of pure state qubit quantum mechanics are possible despite the incompleteness proof, and from Section \ref{s:ZX_graph_states} onwards.
A normal form for stabilizer \ZX-calculus diagrams is introduced and then used to prove that the \ZX-calculus is complete for scalar-free stabilizer quantum mechanics, i.e.\ where two operators $U$ and $V$ are taken to be equal if there exists some non-zero complex number $c$ such that $U=cV$.

That completeness result is expanded in Chapter \ref{ch:more_completeness}.
First, it is shown that the \ZX-calculus is complete for stabilizer quantum mechanics with scalars, i.e.\ any true equality between stabilizer \ZX-calculus diagrams (now with the usual notion of equality) can be derived from the rewrite rules.
This includes the definition of a unique normal form for stabilizer zero diagrams: diagrams representing a zero matrix.
Finally, the completeness proof is extended to the single-qubit Clifford+T group.

Spekkens' toy theory is introduced in the first section of Chapter \ref{s:spekkens}.
A \ZX-like graphical calculus for this toy model is developed and then shown to be complete.
Most of the original work in Sections \ref{s:spekkens_graphical} and \ref{s:spekkens_completeness} was done jointly with Ali Nabi Duman, with the exception of results involving scalar diagrams, which are solely my own work.

Chapter \ref{ch:conclusions} contains the conclusions and some ideas for further work.

%% file: tikz_files/cup_diagram.tikz
\begin{tikzpicture}
	\begin{pgfonlayer}{nodelayer}
		\node [style=none] (0) at (1, 0.25) {};
		\node [style=none] (1) at (2, 0.25) {};
		\node [style=none] (2) at (1.5, -0.25) {};
	\end{pgfonlayer}
	\begin{pgfonlayer}{edgelayer}
		\draw [bend right=45, looseness=1.00] (2.center) to (1.center);
		\draw [bend right=45, looseness=1.00] (0.center) to (2.center);
	\end{pgfonlayer}
\end{tikzpicture}

%% file: graphical.tex
\chapter{Graphical languages and completeness}
\label{ch:graphical}

New discoveries in theoretical physics are formulated and derived using mathematics.
The specific mathematical formalism used thus plays an important role in determining how easy it is to find new results, or to understand them.
Some results are much more intuitive in certain formalisms than in others.

For example, the rule for chain rule for differentiation of a function $f(y)$ with respect to some variable $x$ that $y$ depends on is very intuitive when expressed in Leibniz's notation:
\begin{equation}
 \frac{df}{dx} = \frac{df}{dy} \frac{dy}{dx},
\end{equation}
but much less so in operator notation:
\begin{equation}
 D_x f = (D_y f)(D_x y).
\end{equation}
On the other hand, the operator notation clearly separates the differential operator from the operand, whereas Leibniz's notation does not.

Similarly, while the Old Babylonians would have been able to do many quantum mechanical calculations in cuneiform, that would not have been easy -- and not just because they did not know quantum theory.
Cuneiform uses a base-60 position-value system that allows the representation of large numbers as well as fractions, as long as they terminate in base-60.
Other numbers were approximated, for example $\sqrt{2}$, which is given as $1\;24\;51\;10$ in base-60 -- in cuneiform, there is no symbol for separating the whole part of a number from the fractional part -- in the clay tablet shown in Figure \ref{fig:clay_tablet}; i.e.:
\begin{equation}
 \sqrt{2} \approx 1 + 24(60)^{-1} + 51(60)^{-2} + 10(60)^{-3} \approx 1.4142130,
\end{equation}
which is the closest approximation to three sexagesimal places, and correct up to six decimal places.

For arithmetic operations such as reciprocals, squares, and square roots, the Babylonians relied heavily on pre-computed tables.
They did not have vectors, or complex numbers, so computations involving a three-component complex vector would have to be split into six interlinked computations for the real and imaginary parts of each component.
The Babylonians did not use equations either (those would not be introduced until the 16th century AD), instead relying on ``recipes'' for solving specific classes of problems \cite{merzbach_history_2011}.
Thus, even if they had known about quantum theory, they probably would not have been able to explore the conceptual consequences in much depth as they would have been too busy shutting up and calculating.

\begin{figure}
 \centering
 \includegraphics{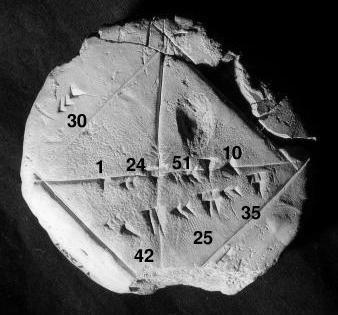}
 \caption{A Babylonian clay tablet showing an approximation of $\sqrt{2}$ as $1\;24\;51\;10$ in base-60.
 This number is used to compute the diagonal of a square of side $30$ with result $42\;25\;35$ in base-60.
 (Photo by Bill Casselman under Creative Commons Attribution 2.5 Generic license \cite{casselman_ybc_2007}.)}
 \label{fig:clay_tablet}
\end{figure}

In this chapter, we consider different formalisms for quantum theory with a particular focus on graphical languages.
By graphical languages we mean two-dimensional mathematical notations or formalisms.
A variety of such languages are currently in use in the quantum computing and quantum foundations community.
We give an overview over some of these languages.
Often, graphical languages are introduced informally; nevertheless, they can be made rigorous using the mathematical framework of category theory.
We introduce the category theory needed to formalise graphical languages for quantum theory.
Next, we give a short introduction to graphical rewriting as a method for deriving equalities between diagrams in graphical languages, and introduce completeness and related concepts.
Finally, we explain how derivations in graphical languages can be automated.

\section{Formalisms for quantum theory}

Many new mathematical formalisms were developed alongside quantum theory; the physicist's interests driving mathematical innovation and the mathematical progress enabling new understanding of physical theory.
We explain why we focus on formalisms used in quantum computation and quantum foundations, and why graphical languages are particularly useful in those areas of research.

\subsection{Quantum computation and quantum foundations}

Quantum theory encompasses the study of many different types of physical systems, from photons to atoms and larger structures.
Quantum foundations is particularly concerned with investigating the differences between quantum physical behaviour and classical physics.
To do this, it helps to focus on idealised systems and ignore aspects of real systems that complicate their analysis but are not considered to be relevant to foundational questions.

For example, while many real physical systems have an infinite-dimensional state space, it is a lot more straightforward to deal with finite-dimensional systems.
Moreover, research often focuses on the smallest non-trivial quantum system: the qubit, whose state space is $\CC^2$.

Qubits also play a central role in the study of quantum computation as the analogues of classical bits.
In classical computing, any finite amount of data can be encoded in a finite string of bits; similarly, in quantum computing, any quantum state of a finite-dimensional system can be encoded in the joint state of some finite number of qubits.
Therefore the study of qubit-based quantum computers yields insights about more general systems as well.

The field of quantum computation is closely related to quantum foundations in that both are concerned with finding similarities and differences between quantum behaviour and classical behaviour.
Quantum computation is more restricted in that it focuses on the efficiency of solving various mathematical problems by encoding them in quantum systems,
whereas quantum foundations involves more general aspects of quantum physics.
Furthermore, most approaches to quantum computation consider the evolution of quantum systems to happen in discrete controlled time steps, e.g.\ the gates in the quantum circuit model or the measurements in measurement-based quantum computing, whereas generally quantum systems evolve continuously.
Many approaches to quantum computing focus on unitary evolution with measurements; this is justified as any quantum process can be considered to be a unitary process on some larger system, parts of which are then discarded.
Quantum computation makes use of a wide range of tools developed in classical computer science, many of which are not used in other areas of quantum foundations.

\subsection{Why graphical languages}

Matrix mechanics has been one of the dominant formalisms for quantum theory since its inception, and it is the main formalism for quantum computing.
It has been used to derive many important and interesting results.
Yet, matrix mechanics is a very low-level formalism, which makes it unwieldy and hard to parse when computations get more complicated.
For example, the size of a matrix is exponential in the dimension of the state space of the underlying system.
Furthermore, it is not straightforward to determine conceptual properties of quantum operations expressed as matrices, e.g.\ whether a matrix acting on multiple systems represents a local transformation or not.
Similarly, some important quantum mechanical phenomena are not at all obvious in matrices, e.g.\ quantum teleportation, which was only discovered more than 60 years after the introduction of matrix mechanics \cite{bennett_teleporting_1993}.

For some fragments of quantum theory there exist more efficient descriptions, e.g. the stabilizer formalism for stabilizer quantum mechanics \cite{gottesman_stabilizer_1997}.
Yet the stabilizer formalism is efficient only for the stabilizer fragment of quantum theory.
Thus, general high-level languages are needed for the study of quantum computation and quantum foundations.
By this we mean languages that hide some of the intricacies of the matrix formalism and instead focus more on conceptual properties of the quantum processes.
The terminology is taken from computer science, where low-level programming languages -- those that closely mimic the actual workings of a computer -- have been superseded by higher-level ones, which are much easier for humans to write and understand, at the cost of requiring a more complicated translation before code can be executed \cite{goldschlager_computer_1982}.
Example of low-level programming languages include machine code and assembly language.
Almost all programming languages commonly used today are high-level, e.g.\ Python, Java, or C++.

The distinction between low-level and high-level can be used more widely, where generally low-level descriptions or formalisms are more detailed and specific, whereas high-level descriptions are more abstract and general.

Specifically, we consider high-level \emph{graphical} languages, i.e.\ languages that use two-dimensional diagrams.
This is in contrast to algebraic terms, which are written on a line and thus are one-dimensional.
Graphical languages can be much more intuitive and easier to understand than algebraic ones.
They are also better at showing symmetries of the underlying structures.

Two-dimensional languages allow \emph{parallel composition} -- applying transformations to two different systems at the same time -- to be separated from \emph{sequential composition} -- the application of transformations to the same system at different times -- by designating one dimension to roughly correspond to ``space'' and the other to ``time''.
This makes graphical languages particularly useful for the study of networked and multiply-connected processes.

\begin{ex}
 The parallel composition of matrices is the Kronecker product.
 E.g., for two 2 by 2 matrices $A$ and $B$ defined as:
 \begin{equation}
  A = \begin{pmatrix} a_{11} & a_{12} \\ a_{21} & a_{22} \end{pmatrix} \qquad \text{and} \qquad B = \begin{pmatrix} b_{11} & b_{12} \\ b_{21} & b_{22} \end{pmatrix},
 \end{equation}
 their parallel composite is the 4 by 4 matrix:
 \begin{equation}
  A\otimes B = \begin{pmatrix} a_{11} b_{11} & a_{11} b_{12} & a_{12}b_{11} & a_{12} b_{12} \\ a_{11} b_{21} & a_{11} b_{22} & a_{12}b_{21} & a_{12} b_{22} \\  a_{21} b_{11} & a_{21} b_{12} & a_{22}b_{11} & a_{22} b_{12} \\ a_{21}b_{21} & a_{21} b_{22} & a_{22}b_{21} & a_{22} b_{22} \end{pmatrix}.
 \end{equation}
\end{ex}

\begin{ex}
 The sequential composition of matrices is matrix multiplication.
 E.g., for two 2 by 2 matrices $A$ and $B$ as defined in the previous example, their sequential composite is the 2 by 2 matrix:
 \begin{equation}
  AB = \begin{pmatrix} a_{11} b_{11} + a_{12} b_{21} & a_{11} b_{12} + a_{12} b_{22} \\ a_{21} b_{11} + a_{22} b_{21} & a_{21} b_{12} + a_{22} b_{22} \end{pmatrix}.
 \end{equation}
\end{ex}

A major difference between classical physics and quantum physics is the way the state spaces of systems compose in parallel, i.e.\ when the systems are put ``side by side'' \cite{abramsky_categorical_2008}: classically, the resulting state space is the Cartesian product of the original spaces, meaning each state of the joint system can be described by specifying separate states for each of the component systems.
For quantum systems, on the other hand, the joint state space is the tensor product of the original state spaces and joint states may not correspond to well-defined states of the separate systems: they can be entangled.
Thus in the study of quantum foundations, the study of composite systems is centrally important, and graphical languages offer an intuitive way of analysing these systems.

\section{Graphical languages for quantum theory}

A variety of high-level graphical languages are already in use in the quantum computing, quantum information, and quantum foundations community.
We introduce several of these languages and discuss their applications and limits.

\subsection{Quantum circuit notation}
\label{s:quantum_circuits}

Quantum circuit notation is a well known graphical language for quantum computation, associated with the quantum circuit model of quantum computation \cite{deutsch_quantum_1989}, which is derived from classical logic gate circuits.
In both quantum and classical circuits, computations are broken down into basic steps called \emph{gates}, which are taken from a fixed \emph{gate set}.
This enables the complexity of computations to be analysed: if every gate is assumed to take a fixed amount of time or some other resource, then the number of gates in the circuit or the number of sequential layers of gates is a measure of the complexity.

The quantum circuit model implicitly assumes that the evolution of the underlying quantum systems happens in discrete steps, as represented by the discrete gates.
Furthermore, it assumes that systems remain in the same state unless acted upon by a gate.

In quantum circuit notation, gates are (usually) denoted by labelled boxes with $n$ input wires on the left and $n$ output wires on the right, where $n$ is some positive integer \cite{nielsen_quantum_2010}.
According to the number of their inputs (and outputs), gates are referred to as ``single-qubit gates'', ``two-qubit gates'', and so forth.
A piece of wire without any gates denotes the identity transformation on a single qubit, thus the length of wires is irrelevant for the interpretation of a quantum circuit diagram.
Gates can be combined by stacking them horizontally, which denotes the tensor product of the corresponding matrices, i.e.\ the gates are applied to different systems at the same time.
Alternatively, the inputs of one gate can be plugged into the outputs of another, which denotes matrix multiplication: the gates are applied to the same system at different times.
An example circuit with one-, two-, and three-qubit gates is shown in Figure \ref{fig:circuit_example}.

\begin{figure}
 \centering
 \input{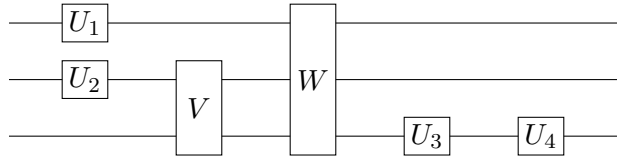}
 \caption{A quantum circuit diagram on three qubits. The operators $U_1,U_2,U_3,$ and $U_4$ are single-qubit unitaries, $V$ is a two-qubit unitary, and $W$ a three-qubit unitary.}
 \label{fig:circuit_example}
\end{figure}

Wires in a quantum circuit diagram always connect inputs of one gate to outputs of another: never inputs to inputs or outputs to outputs, and never inputs of one gate to outputs of the same gate.
Thus quantum circuit diagrams do not contain any cycles; a path following wires from outputs to inputs and traversing gates from inputs to outputs can never return to a gate previously visited.

The most commonly used gate set for quantum circuits consists of arbitrary single-qubit gates together with the two-qubit controlled-\NOT{} gate, which represents the matrix:
\begin{equation}
 C_X = \begin{pmatrix} 1 & 0 & 0 & 0 \\ 0 & 1 & 0 & 0 \\ 0 & 0 & 0 & 1 \\ 0 & 0 & 1 & 0 \end{pmatrix}.
\end{equation}
The controlled-\NOT{} gate is usually denoted by the symbol shown in Figure \ref{fig:gate_examples} a, rather than by a box.

Any unitary operation on a finite number of qubits can be expressed as a quantum circuit consisting of controlled-\NOT{} and single-qubit gates.
In classical computing, a finite set of gates suffices to construct a logic circuit computing any Boolean function, e.g.\ the \textsc{nand} gate.
For quantum computing, there are many finite gate sets that allow any unitary operator to be approximated to arbitrary accuracy \cite{nielsen_quantum_2010}.

One such set is the so-called Clifford+T set \cite{boykin_universal_1999}, which consists of the Clifford gates:
\begin{equation}
 \left\{ \sqGate{$S$}, \quad \sqGate{$H$}, \quad \input{tikz_files/circuit_cNOT.tikz} \right\},
\end{equation}
where:
\begin{align}
 S &= \begin{pmatrix} 1 & 0 \\ 0 & i \end{pmatrix}, \quad\text{and} \\
 H &= \frac{1}{\sqrt{2}} \begin{pmatrix} 1 & 1 \\ 1 & -1 \end{pmatrix},
\intertext{together with the $T$-gate:}
 T &= \begin{pmatrix} 1 & 0 \\ 0 & e^{i\pi/4} \end{pmatrix}.
\end{align}
The $S$-gate is often called \emph{phase gate}, though that name is sometimes used for any or all gates of the form:
\begin{equation}\label{eq:generalised_phase_gate}
 R_\phi = \begin{pmatrix} 1 & 0 \\ 0 & e^{i\phi} \end{pmatrix}
\end{equation}
with some real number $\phi$.
To avoid confusion, we shall call the latter \emph{generalised phase gates}.
The $H$-gate is called \emph{Hadamard gate}.

Strictly speaking, the phase gate is redundant in the Clifford+T set as $T^2=S$.
It makes sense to include $S$ as a separate gate nevertheless, since in many quantum error correcting codes, Clifford gates (including the phase gate) are easy to implement in a fault-tolerant fashion, while $T$ is much harder to implement fault-tolerantly \cite{nielsen_quantum_2010}.
Thus, for the analysis of the complexity of fault-tolerant computations, it makes sense to distinguish between $S$- and $T$-gates.

In both of the above gate sets, the controlled-Z gate (see Figure \ref{fig:gate_examples} c) is sometimes used instead of controlled-\NOT{} because it is symmetric under interchange of the two qubits it acts upon.
The fundamental properties of the gate set remain unchanged under this substitution because the two gates can be transformed into each other using single-qubit Clifford gates:
\begin{equation}
 C_Z = (I\otimes H) C_X (I\otimes H),
\end{equation}
where $I$ denotes the single-qubit identity transformation:
\begin{equation}
 I = \begin{pmatrix} 1 & 0 \\ 0 & 1 \end{pmatrix}.
\end{equation}

\begin{figure}
 \centering
 \input{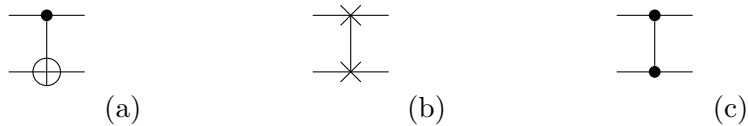}
 \caption{Examples of special gate symbols in quantum circuit notation: (a) controlled-\NOT{} gate, (b) \SWAP{} gate, (c) controlled-Z gate \cite{nielsen_quantum_2010}.
 These symbols show clearly that \SWAP{} and controlled-Z are symmetric under interchange of the two qubits they act upon, whereas controlled-\NOT{} is not symmetric.}
 \label{fig:gate_examples}
\end{figure}

Where complicated quantum processes are built up from more basic transformations, a quantum circuit diagram can be much easier to understand than a corresponding algebraic representation.
For example, the quantum circuit in Figure \ref{fig:circuit_example} can be written algebraically as:
\begin{equation}\label{eq:circuit_example}
 (I\otimes I\otimes (U_4\circ U_3)) \circ W \circ (I\otimes V) \circ (U_1\otimes U_2\otimes I),
\end{equation}
where $I$ denotes the single-qubit identity transformation.
It is clearly much easier to see how the different transformations compose in the diagram.

Yet there are also some issues with quantum circuit notation.
Quantum circuits are not rigorously defined and there are no widely accepted rules for determining whether two circuits are equal: to test equality, circuits are usually translated back into matrices.
This problem could be resolved as shown in Section \ref{s:graphical_rigorous}; cf.\ also the set of generators and relations for quantum circuits representing Clifford unitaries given by Selinger \cite{selinger_generators_2013}.

Quantum circuit notation is also not as intuitive as it could be: for example, instead of the \SWAP{} gate symbol in Figure \ref{fig:gate_examples} b, it would be better to use a wire crossing.
That way, equalities such as:
\begin{equation}
 \SWAP{} \circ (U\otimes V) \circ \SWAP{} = V\otimes U
\end{equation}
for any single-qubit unitaries $U,V$ become intuitively obvious, cf.\ Figure \ref{fig:swap}.
This, again, is a problem that can be remedied.

Lastly, quantum circuits always distinguish strictly between the inputs and outputs of a gate and do not allow curved wires or cycles.
Due to map-state duality, this distinction is not a very natural one for quantum processes.
As demonstrated e.g.\ by atemporal diagrams (see Section \ref{s:atemporal}), it can be quite useful to drop, or at least loosen, the strict time ordering in diagrams.
Furthermore, cycles have a very natural interpretation in diagrams for quantum processes as representing the operation of tracing out subsystems.
Nevertheless, these generalisations are not allowed by quantum circuit diagrams.

\subsection{Stabilizer graphs}
\label{s:stabilizer_graphs}

The stabilizer graph notation represents pure qubit stabilizer states as decorated graphs \cite{elliott_graphical_2008}.
Stabilizer states are those quantum states that are simultaneous eigenstates of a group of Pauli products: tensor products of the Pauli matrices and the identity matrix (cf.\ Section \ref{s:stabilizer_QM}).
A special class of stabilizer states are the \emph{graph states}, whose entanglement structure is that of a finite simple graph, where the qubits represent the vertices and entanglement represents the edges.
Graph states thus have a straightforward diagrammatic representation by simply drawing the associated graph.

Any stabilizer state is related to some graph state via a local Clifford operation, i.e.\ an operation that decomposes into a tensor product of single-qubit Clifford unitaries \cite{van_den_nest_graphical_2004}.
Stabilizer graph notation extends the graph state notation to general stabilizer states by using decorations on the graph vertices to denote the unitary applied to the corresponding qubit.
Thus, vertices in stabilizer graphs can be empty or filled, have a minus sign or not, and they can have a self-loops or not.
An example stabilizer graph is shown in Figure \ref{fig:stabilizer_graph_example} a.

\begin{figure}
 \centering
 \input{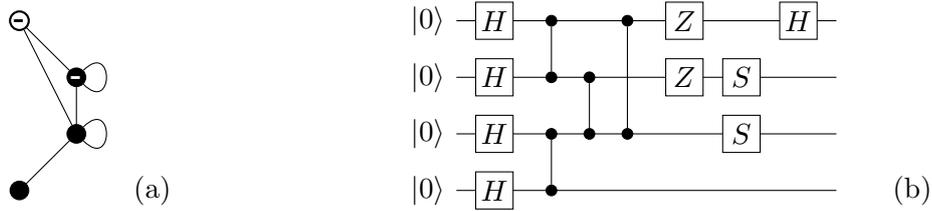}
 \caption{(a) A stabilizer graph, and (b) a quantum circuit for preparing the corresponding stabilizer state. The initial layer of Hadamard gates and the following controlled-Z gates prepare the graph state, thus the correspondence between controlled-Z gates in the circuit and edges in the graph. The final single-qubit gates correspond to the decorations: $Z$-gates to minus signs, $S$-gates to self-loops, and Hadamard gates to empty nodes.}
 \label{fig:stabilizer_graph_example}
\end{figure}

Stabilizer graphs are not unique, i.e.\ there may be multiple ways of representing the same state.
Nevertheless, the formalism includes a decision procedure for diagram equality, as well as algorithms for the transformation of stabilizer graphs under Clifford operations \cite{elliott_graphical_2008} and Pauli measurements \cite{elliott_graphical_2010}.

Stabilizer graphs are a more efficient notation for stabilizer states than the standard notation in terms of computational basis states.
Some symmetries of stabilizer states are easier to see in stabilizer graphs than in other formalisms.

Yet, unlike the other graphical languages introduced here, stabilizer graph notation represents quantum states rather than more general transformations.
Thus, most of the discussion in this chapter about making graphical languages rigorous does not apply to stabilizer graphs.
Furthermore, this is the least general notation introduced here, as it can only represent a fragment of pure qubit quantum theory.
Still, the stabilizer graph formalism provides some useful ideas for later work in this thesis, cf.\ Section \ref{s:ZX_graph_states}.

\subsection{Atemporal diagrams}
\label{s:atemporal}

Atemporal diagrams generalise circuit diagrams by dropping any notion of time ordering in order to explore map-state duality \cite{griffiths_atemporal_2006}.
They also allow arbitrary state spaces, rather than just qubits.

Diagrams consist of large labelled circles called \emph{centres}, which represent quantum states, transformations, or measurements, as well as smaller labelled nodes.
The latter, which are connected to the centres by directed edges, denote the Hilbert spaces involved in a process.
Edges can connect to centres anywhere, and centres can have any number of edges.
Some examples of atemporal diagrams are shown in Figure \ref{fig:atemporal_diagram_example}.
Two atemporal diagrams are equal whenever the same components are connected in the same way, irrespective of the actual layout of the diagram.
The direction of the edges, together with the decoration of the nodes -- ``open'' (i.e.\ empty) or ``closed'' (i.e.\ filled) -- indicates whether a process involves a Hilbert space or its dual space.
There is some redundancy in the notation, as edges are always directed towards open nodes and away from closed ones.

Disconnected diagrams can be put next to each other to denote the tensor product of the corresponding transformations.
An open node and a closed node with the same label, representing some Hilbert space and its dual, can be plugged together; this corresponds to an inner product or a (possibly partial) trace.
When diagram components are plugged together in this way, the two connected nodes are usually left out and the wire is labelled instead, cf.\ Figure \ref{fig:atemporal_diagram_example} b.
The adjoint of an atemporal diagram can be constructed by changing empty nodes to filled and filled ones to empty, flipping the direction of arrows, and adding ``$\dagger$'' symbols to the labels of boxes with the rule that two daggers on the same label cancel.
An example is given by Figure \ref{fig:atemporal_diagram_example} b and c.

\begin{figure}
 \centering
 \input{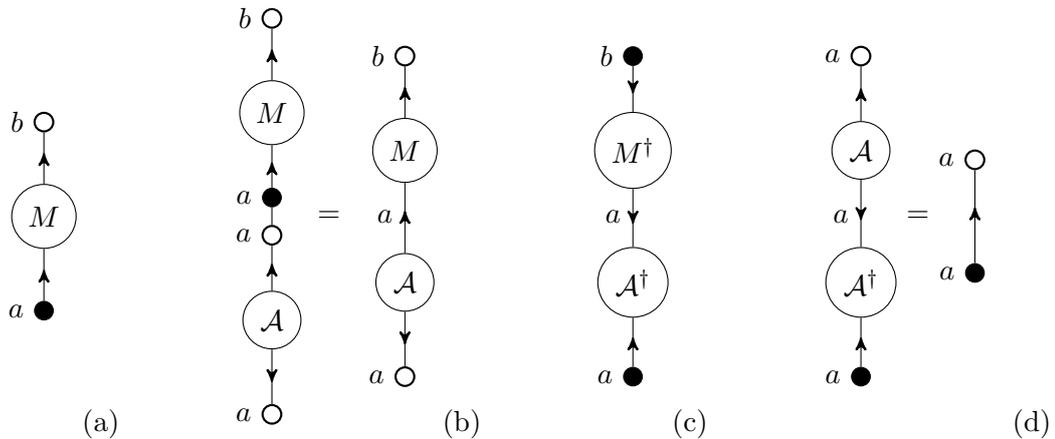}
 \caption{Examples of atemporal diagrams: (a) A transformation $M\in\HH_a^\dagger\otimes\HH_b$, i.e.\ a map from $\HH_a$ to $\HH_b$.
 (b) Applying a transposer to the wire labelled ``$a$'' yields a state $\mathcal{A}M$ in the space $\HH_a\otimes\HH_b$.
 (c) The adjoint of the state $\mathcal{A}M$, which is an element of the space $\HH_a^\dagger\otimes\HH_b^\dagger$.
 (d) The adjoint of the transposer $\mathcal{A}$ is its inverse \cite{griffiths_atemporal_2006}: plugging $\mathcal{A}$ and $\mathcal{A}^\dagger$ together yields the identity map on $\HH_a$, which is written as a line with no centres.
 Note that while all the diagrams shown here are line graphs, atemporal diagrams can have any structure and centres are allowed to have more than two connecting edges.}
 \label{fig:atemporal_diagram_example}
\end{figure}

The only distinction between inputs and outputs in atemporal diagrams is the direction of the arrows and the colour of the nodes.
An invertible ``transposer'' $\mathcal{A}\in\HH_a\otimes\HH_a$ maps closed nodes to open ones, cf.\ Figure \ref{fig:atemporal_diagram_example} b, c, and d.
The transposer is nominally a state, so it might seem strange that it should have an inverse.
Yet in the notation of atemporal diagrams, a state in $\HH_a\otimes\HH_a$ can also be thought of as a map from $\HH_a^\dagger\to\HH_a$, which can be invertible in the more usual sense.
This notion of invertibility also appears in the \emph{snake equations} in compact closed categories, see Section \ref{s:graphical_rigorous}.

Atemporal diagrams are useful for showing analogies between maps and states, but as a notation they are too general to be very useful for computations.

\subsection{Other graphical languages}

There are various other graphical languages for quantum information or computation, many of them inspired by Penrose's graphical notation for tensors \cite{penrose_applications_1971}.
In Penrose's notation, tensors are denoted by simple geometric symbols like circles or squares, and tensor indices by wires going into or out of the tensor symbol.
Outer products correspond to juxtaposition of tensor symbols, contractions to wire connections.
Further decorations on sets of wires are used to denote symmetrisation or anti-symmetrisation over the corresponding indices.
A simple example of this notation is shown in Figure \ref{fig:other_graphical_languages} a.

\begin{figure}
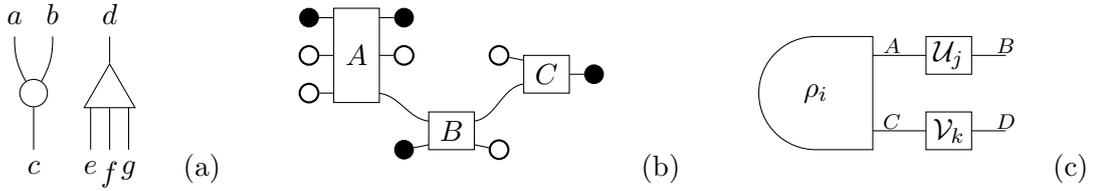

 \centering
 \input{tikz_files/penrose_example.tikz} $\;$ (a) $\qquad$
 \input{tikz_files/duotensor_example.tikz} $\;$ (b) $\qquad$
 \input{tikz_files/Pavia_example.tikz} $\;$ (c)
 \caption{(a) Penrose's graphical notation for the outer product $\theta^{ab}_c\chi^d_{efg}$, where $\theta^{ab}_c$ is denoted by a circle and $\chi^d_{efg}$ by a triangle. (b) The duotensor for a network consisting of three operations. (c) Diagram for the preparation of a joint state of two systems $A$ and $C$, followed by local transformations on the two subsystems, in the Pavia notation.}
 \label{fig:other_graphical_languages}
\end{figure}

Hardy's duotensor notation provides a unified graphical language for generalised probabilistic theories including quantum theory \cite{hardy_formalism-local_2013}.
Operations in those theories are denoted by boxes, which can be wired together to form networks.
Duotensors are the mathematical objects corresponding to certain operations, these are represented graphically as boxes with black or white nodes on all of the outputs.
Plugging duotensors together corresponds to summations.
In this way, duotensors can be used to derive probabilities for the corresponding operations.
An example duotensor network is shown in Figure \ref{fig:other_graphical_languages} b.

Chiribella et al.\ describe and analyse general probabilistic theories using a graphical language \cite{chiribella_probabilistic_2010}, which we call the Pavia notation.
In this language, which is modelled after quantum circuit notation, labelled boxes denote processes, called ``tests''.
Wires correspond to systems and are labelled with the system type.
Tests can be composed in parallel or in sequence, and there are tests without inputs, corresponding to preparations of systems, and tests with no outputs, corresponding to destructive measurements.
An example diagram in this graphical language is shown in Figure \ref{fig:other_graphical_languages} c.
Tests are probabilistic, i.e.\ they may have one of a set of different effects.
In addition to the effect on systems, any test also has a heralding output: this can be thought of as a display or light on a laboratory device, which signals which of the set of operations has actually happened.
These outputs are indicated by subscripts on the test labels.

Penrose's graphical notation is not a notation for quantum theory but for tensors.
The other two notations encompass quantum theory, but they are very general.
This makes them less useful for specific applications in quantum theory.

\subsection{The \ZX-calculus}

There are various graphical languages directly based on categorical quantum mechanics, including the \ZX-calculus \cite{coecke_interacting_2008}.
The \ZX-calculus is a formalism for pure state qubit quantum mechanics with post-selected measurements.
Diagrams consists of green and red nodes called \emph{spiders} with arbitrarily many inputs and outputs and attached \emph{phase labels}, plus yellow nodes with one input and output each.
The green and red nodes represent maps in the computational and Hadamard basis, respectively, and the yellow nodes represent Hadamard operators.
Nodes can be connected in any way and edges are allowed to cross and curve.

Ignoring normalisation, quantum circuits consisting of controlled-\NOT{}, Hadamard, and generalised phase gates -- cf.\ \eqref{eq:generalised_phase_gate} -- can straightforwardly be translated into the \ZX-calculus, see Figure \ref{fig:circuit-ZX}.
The \ZX-calculus is more versatile than quantum circuit notation though: most \ZX-calculus diagrams do not arise from quantum circuits in this way.

\begin{figure}
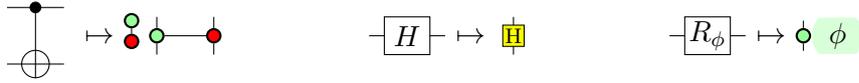

 \centering
 \input{tikz_files/circuit_cNOT.tikz} $\mapsto$ \innerprodgr{} \input{tikz_files/controlled-NOT.tikz} $\qquad\qquad$ \sqGate{$H$} $\mapsto$ \Hadamard{} $\qquad\qquad$ \sqGate{$R_\phi$} $\mapsto$ \phase{gn, label={[gphase]right:$\phi$}}
 \caption{The translations of controlled-\NOT{}, Hadamard, and generalised phase gates $R_\phi$ into the \ZX-calculus \cite{coecke_interacting_2008}. Note the change of orientation from left-to-right to bottom-to-top. In the \ZX-calculus representation of controlled-\NOT{}, the vertical wire through the green node (here: on the left) corresponds to the control qubit, the vertical wire through the red node (here: on the right) to the target qubit.}
 \label{fig:circuit-ZX}
\end{figure}

Furthermore, with the ability to represent states and post-selected measurements as well as unitary transformations, measurement-based quantum computations (MBQC) \cite{raussendorf_one-way_2001} can be translated into \ZX-calculus diagrams in a much more natural way than their translation into circuits \cite{duncan_rewriting_2010}.
Each qubit in a graph or cluster states -- the underlying resource for MBQC -- can be represented in the \ZX-calculus as a green node with an output.
Edges in the graph state are represented by yellow nodes connected to two qubits.
The measurements in the basis $\{\ket{0}+e^{i\phi}\ket{1}, \ket{0}-e^{i\phi}\ket{1}\}$ for some real $\phi$ required by MBQC algorithms are represented as green nodes with one input and phase labels $\phi$ or $(\phi+\pi)$.
Extra notation can be introduced to keep track of the propagating Pauli corrections resulting from the different measurement outcomes \cite{duncan_rewriting_2010}.

\begin{figure}
 \centering
 \input{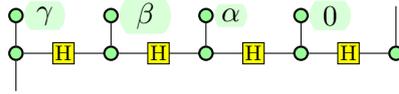}
 \caption{The \ZX-calculus representation of a MBQC pattern for a general single-qubit unitary $R_\alpha H  R_\beta H R_\gamma$, where $R_\phi$ are the generalised phase gates defined in \eqref{eq:generalised_phase_gate} and $H$ is the Hadamard gate.
 In the MBQC pattern, the input qubit on the left is entangled with the first qubit of a 4-qubit line graph.
 The first four qubits are then projected onto the states $(\ket{0}+e^{i\phi}\ket{1})$ with $\phi$ taking values $\gamma,\beta,\alpha,$ and $0$, respectively.
 For simplicity, it has been assumed here that all measurements give the desired outcomes so that corrections are not necessary.
 The \ZX-calculus notation can also be extended to keep track of the propagation of error corrections \cite{duncan_rewriting_2010}.}
 \label{fig:MBQC-ZX}
\end{figure}

A more detailed and rigorous introduction to the \ZX-calculus is given in Chapter \ref{s:ZX-calculus}.

\section{Making graphical languages rigorous}
\label{s:graphical_rigorous}

Graphical notations are often introduced as informal personal short-hands and used to develop an intuitive understanding of a problem that can then be confirmed using a more rigorous but less intuitive language.
This means doing the same work twice: once graphically, then again in the alternative formalism.
An alternative approach is to make the graphical languages themselves rigorous, so reliable results can be derived entirely graphically.

The graphical languages for quantum theory introduced in the previous section, with the exception of stabilizer graphs, have many properties in common: in all of these languages, processes are denoted by some kind of node or box, and systems are denoted by wires.
These languages can be made rigorous using category theory.
That approach was pioneered by Joyal and Street, who analysed a range of graphical notations from Feynman diagrams to Petri Nets and gave them rigorous underpinnings \cite{joyal_geometry_1991}.
Category theory is the natural formalism for making graphical languages rigorous, as \emph{monoidal categories} are the most general mathematical structures incorporating both parallel and sequential composition of transformations.
We introduce the concepts from category theory needed to make graphical languages rigorous.
We then explain how to apply the category theory to graphical languages like the ones considered in the previous section.
Further information can be found in \cite{coecke_categories_2010}, which is aimed at physicists.
The standard textbook is \cite{mac_lane_categories_1998}.
A soon-to-appear textbook will introduce category theory, graphical languages, and quantum theory side-by-side \cite{coecke_picturing_2015}.

\subsection{Basic category theory for graphical languages}
\label{s:category_theory}

A category is an abstract mathematical structure describing -- informally speaking -- a collection of processes and the way they compose.
Unlike in a group, where any two elements can be composed, generally not all processes in a category are composable: each process in a category has a specified ``input system'' and ``output system'', and two processes compose only if the input system of the one is the same as the output system of the other.
Formally, the ``systems'' are called \emph{objects} and the processes \emph{arrows}.

\begin{dfn}
 A \emph{category} $\mathcal{C}$ consists of:
 \begin{itemize}
  \item a collection of \emph{objects} $Ob(\mathcal{C})$,
  \item for any two objects $A,B\in Ob(\mathcal{C})$, a set of \emph{arrows} $\mathcal{C}(A,B)$,
  \item for each object $A\in Ob(\mathcal{C})$, an \emph{identity arrow} $1_A\in\mathcal{C}(A,A)$, and
  \item a \emph{sequential composition operation} for arrows:
   \begin{equation}
    (-\circ -): \mathcal{C}(B,C) \times \mathcal{C}(A,B) \to \mathcal{C}(A,C),
   \end{equation}
   where $A,B,C\in Ob(\mathcal{C})$,
 \end{itemize}
 satisfying the following axioms:
 \begin{itemize}
  \item Composition is associative, i.e.\ for any $f\in\mathcal{C}(A,B)$, $g\in\mathcal{C}(B,C)$, and $h\in\mathcal{C}(C,D)$:
   \begin{equation}
    h\circ(g\circ f) = (h\circ g) \circ f.
   \end{equation}
  \item The identity arrows are units for composition, i.e.\ for all $f\in\mathcal{C}(A,B)$:
   \begin{equation}
    1_B\circ f = f = f\circ 1_A.
   \end{equation}
 \end{itemize}
\end{dfn}

As the source and target objects are important, an arrow $f\in\mathcal{C}(A,B)$ is usually written as $f:A\to B$ or even $A\xrightarrow{f} B$.
Arrows are sometimes also called \emph{morphisms}; an arrow that has an inverse is called \emph{isomorphism}.

\begin{ex}
 Any physical theory could be formalised as a category, where the physical systems are the objects and their transformations are the arrows.
 In this case, the sequential composition operation corresponds to simply applying one transformation after the other, and the identity arrow corresponds to the transformation that leaves a system invariant.
\end{ex}

\begin{ex}
 A more mathematical example is $\Set$, the category whose objects are sets and whose arrows are functions. 
 Composition is sequential application of functions, i.e.\ given functions $f:A\to B$ and $g:B\to C$ for sets $A,B,C$, their composite is:
 \begin{equation}
  g\circ f: A\to C :: a\mapsto g(f(a)).
 \end{equation}
 The identity arrows are the identity functions, i.e.\ for $A\in Ob(\mathcal{C})$:
 \begin{equation}
  1_A: A\to A :: a\mapsto a.
 \end{equation}
\end{ex}

\begin{ex}
 The category $\Rel$ again has sets as objects but the arrows are relations.
 A relation $A\xrightarrow{R} B$ can be thought of as a subset of the Cartesian product $A\times B$.
 The sequential composition operation in $\Rel$ is that of relational composition, i.e.\ the composite of $R$ and a relation $B\xrightarrow{S} C$ is:
 \begin{equation}
  S\circ R = \{ (a,c) \mid \exists b\in B \text{ s.t. } (a,b)\in R \wedge (b,c) \in S \} \subseteq A\times C.
 \end{equation}
 Identity arrows are the identity functions, considered as relations:
 \begin{equation}
  1_A = \{(a,a) \mid a\in A \} \subseteq A\times A.
 \end{equation}
\end{ex}

\begin{ex}
 The category $\Hilb$ has complex Hilbert spaces as objects and bounded linear maps as arrows.
 The sequential composition operation is the composition of linear maps as functions.
 Identity arrows are the usual identity linear maps.
\end{ex}

The categories $\FRel$ and $\FHilb$ are defined by restricting the objects of $\Rel$ to finite sets and of $\Hilb$ to finite-dimensional Hilbert spaces, respectively.

Classical deterministic physics is modelled in the category $\Set$.
$\Hilb$ and $\FHilb$ are the settings for categorical quantum mechanics.
While $\Rel$ at first glance seems very similar to $\Set$ -- after all, they have the same objects -- it is actually more similar to $\Hilb$.
We see in Chapter \ref{s:spekkens} that a subcategory of $\FRel$ describes Spekkens' toy bit theory.

\begin{dfn}
 Let $\mathcal{C}$ and $\mathcal{D}$ be categories. $\mathcal{C}$ is a \emph{subcategory} of $\mathcal{D}$ if all objects and arrows of $\mathcal{C}$ are also objects and arrows of $\mathcal{D}$, with identities and composition of arrows being the same in both categories.
\end{dfn}

It can be interesting to relate categories to each other via transformations that act on categories.

\begin{dfn}\label{dfn:functor}
 Let $\mathcal{C}$ and $\mathcal{D}$ be categories. A map $F:\mathcal{C}\to\mathcal{D}$ is a \emph{functor} if it satisfies the following:
 \begin{itemize}
  \item $F$ assigns an object $FA\in Ob(\mathcal{D})$ to each object $A\in Ob(\mathcal{C})$,
  \item $F$ assigns an arrow $Ff\in\mathcal{D}(FA,FB)$ to each arrow $f\in\mathcal{C}(A,B)$,
  \item $F$ preserves composition of arrows:
   \begin{equation}
    F(f\circ g) = Ff \circ Fg
   \end{equation}
   for any composable $f,g$ in $\mathcal{C}$, and
  \item $F$ preserves identity arrows:
   \begin{equation}
    F1_A = 1_{FA}.
   \end{equation}
 \end{itemize}
\end{dfn}

\begin{ex}
 There exists a functor from the category of physical systems that evolve according to classical deterministic physics into the category $\Set$, sending each physical system to its set of states, and each transformation to a corresponding function between state sets.
\end{ex}

\begin{ex}
 The map from $\Set$ to $\Rel$ that does the obvious thing on objects and sends each function $f:A\to B$ to a relation $R_f\subseteq A\times B$ given by:
 \begin{equation}
  R_f = \{(a,f(a)) \mid a\in A\},
 \end{equation}
 is a functor.
\end{ex}

A basic category allows sequential composition of transformations, i.e.\ applying transformations one after the other.
There is no way of expressing the idea of putting two systems side by side and applying a transformation to the first ``at the same time'' as applying a transformation to the second.
To study this new type of composition, called \emph{parallel composition}, we add new structure to categories.

\begin{dfn}
 A \emph{strict monoidal category} is a category $\mathcal{C}$ together with a parallel composition operation for objects, denoted by $A\otimes B$, a \emph{unit object} $I$, and a parallel composition operation for arrows:
 \begin{equation}
  (-\otimes -) : \mathcal{C}(A,B)\times\mathcal{C}(C,D) \to \mathcal{C}(A\otimes C,B\otimes D),
 \end{equation}
 such that for any $A,B,C\in Ob(\mathcal{C})$ and any arrows $f,g,h,j$ that are composable in the required ways, the following hold.
 \begin{itemize}
  \item The parallel composition is associative on objects, and $I$ is a unit for it:
   \begin{gather}
    (A\otimes B)\otimes C = A\otimes (B\otimes C) \\
    A\otimes I = A = I\otimes A.
   \end{gather}
  \item The parallel composition is associative on arrows, and $1_I$ is a unit for it:
   \begin{gather}
    h\otimes(g\otimes f) = (h\otimes g) \otimes f \\
    f\otimes 1_I = f = 1_I \otimes f.
   \end{gather}
  \item Parallel and serial composition satisfy the \emph{interchange law}:
   \begin{equation}\label{eq:interchange}
    (g\circ f)\otimes(j\circ h) = (g\otimes j) \circ (f\otimes h).
   \end{equation}
 \end{itemize}
\end{dfn}

The parallel composition operation in a monoidal category is also called \emph{monoidal product}, hence the name.
The term ``strict'' in the above definition refers to the fact that the associative and unit laws for parallel composition are equalities.
In general monoidal categories, these only hold up to so called \emph{structural isomorphisms}, which satisfy a number of \emph{coherence equations}.
The coherence equations then imply the interchange law.
We ignore these intricacies here, which is justified as any monoidal category is equivalent -- in a rigorously-defined way, via functors that preserve the monoidal structure -- to some strict monoidal category \cite{mac_lane_categories_1998}, and graphical languages always yield strict monoidal categories.

\begin{ex}
 The category $\Set$ can be made into a monoidal category by using the Cartesian product of sets as the parallel composition on objects.
 The unit object is the one-element set.
 Parallel composition of functions corresponds to element-wise application: given functions $f:A\to B$ and $g:C\to D$, the parallel composite of $f$ and $g$ is:
 \begin{equation}
  f\otimes g: (A\otimes C)\to (B\otimes D) :: (a,c)\mapsto(f(a),g(c)).
 \end{equation}
 $\Rel$ can be made into a monoidal category in a similar way.
\end{ex}

\begin{ex}
 The category $\Hilb$ can be made into a monoidal category with the usual tensor product as the parallel composition operation.
 The unit object is the one-dimensional Hilbert space.
 The same holds for $\FHilb$.
\end{ex}

\begin{ex}\label{ex:Mat_CC}
 The strict monoidal category equivalent to $\FHilb$, denoted $\Mat_\CC$, has natural numbers as objects (which can be thought of as the dimension of the Hilbert space).
 Arrows in $\Mat_\CC(n,m)$ are complex matrices of size $m$ by $n$, with matrix multiplication as sequential composition and identity matrices as identity arrows.
 The parallel composition of objects is given by multiplication of numbers: $n\otimes m=nm$, with 1 as the unit object.
 Arrows compose in parallel by Kronecker product of matrices.

 An object in $\Mat_\CC$ can be thought of as a Hilbert space with a chosen basis, which then allows linear maps to be uniquely expressed as matrices in terms of those chosen bases.
\end{ex}

Strict monoidal categories are already almost sufficient for describing circuit diagrams with their rigid structure: all that is missing is a \SWAP-map that interacts with the other arrows in the intuitively expected way.

\begin{dfn}
 A \emph{strict symmetric monoidal category} is a strict monoidal category $\mathcal{C}$ with a \emph{swap arrow} $\sigma_{A,B}$ for any pair of objects $A,B\in Ob(\mathcal{C})$, which satisfies the following axioms.
 \begin{itemize}
  \item Swapping two systems and then swapping them again is equivalent to not doing anything:
   \begin{equation}
    \sigma_{B,A} \circ \sigma_{A,B} = 1_A\otimes 1_B.
   \end{equation}
  \item Swapping two objects and then applying two arrows in parallel is the same as interchanging the arrows and then swapping, i.e.\ for any $f:A\to A'$ and $g:B\to B'$:
   \begin{equation}
     (f\otimes g)\circ\sigma_{A,B} = \sigma_{B',A'}\circ (g\otimes f).
   \end{equation}
  \item Swapping an object with the unit object $I$ is the same as not doing anything:
   \begin{equation}
    \sigma_{A,I} = 1_A.
   \end{equation}
  \item Swapping an object with a composite object is the same as component-wise swapping:
   \begin{equation}
    (1_B\otimes\sigma_{A,C})\circ(\sigma_{A,B}\otimes 1_C) = \sigma_{A,B\otimes C}.
   \end{equation}
 \end{itemize}
\end{dfn}

Again, the strict symmetric monoidal category is actually a special case of a symmetric monoidal category, in which several of the axioms involve isomorphisms rather than being exact equalities.

\begin{ex}
 The category $\Set$ is symmetric with swap arrow:
 \begin{equation}
  \sigma_{A,B}: (A\otimes B)\to (B\otimes A) :: (a,b)\mapsto (b,a).
 \end{equation}
 The corresponding relation is a swap arrow for $\Rel$, and similarly for $\FRel$.
\end{ex}

\begin{ex}
 The category $\Hilb$ is symmetric with the swap arrow $\sigma_{\mathcal{H},\mathcal{H'}}$ for two Hilbert spaces $\mathcal{H},\mathcal{H}'$ being the unique linear map satisfying:
 \begin{equation}
  \ket{\phi}\otimes\ket{\psi}\mapsto\ket{\psi}\otimes\ket{\phi}
 \end{equation}
 for all $\ket{\phi}\in\mathcal{H}$, $\ket{\psi}\in\mathcal{H'}$.
 The swap arrow for $\FHilb$ is defined similarly.
\end{ex}

Circuit diagrams, including quantum circuits, can be modelled in strict symmetric monoidal categories.
\ZX-calculus diagrams on the other hand have components that cannot be expressed in general strict symmetric monoidal categories: curved wires with either two inputs or two outputs, and cycles.
These can be described category-theoretically as a \emph{compact structure} \cite{abramsky_categorical_2004,abramsky_categorical_2008}.

\begin{dfn}\label{dfn:compact_closed_category}
 A strict symmetric monoidal category $\mathcal{C}$ is called a \emph{compact closed category} if for every object $A\in Ob(\mathcal{C})$ there exists an object $A^*\in Ob(\mathcal{C})$ called the \emph{dual} of $A$ with arrows $\eta_A:I\to A^*\otimes A$ and $\epsilon_A:A\otimes A^*\to I$ such that:
 \begin{align}
  (\epsilon_A\otimes 1_A) \circ (1_A\otimes\eta_A) &= 1_A, \quad \text{and} \\
  (1_{A^*}\otimes\epsilon_A) \circ (\eta_A\otimes 1_{A^*}) &= 1_{A^*}.
 \end{align}
\end{dfn}

As before, if the compact structure is put on a general symmetric monoidal category, the equalities in the definition involve various isomorphisms, which are identities in the strict case.
The coherence theorems for the structural isomorphisms for compact closed categories were originally proved in \cite{kelly_coherence_1980}.

\begin{ex}
 The category $\Set$ is not compact closed because there is only a single function from any object $A$ to the one-element set: the function that maps every element of $A$ to the single element of the one-element set.
 It is therefore impossible to find arrows satisfying the equalities in Definition \ref{dfn:compact_closed_category}.
\end{ex}

\begin{ex}
 The category $\Rel$ on the other hand can be given a compact structure: take each set to be self-dual, i.e.\ $A^*=A$ for all $A\in Ob(\Rel)$.
 Denote the one-element set by $\{\bullet\}$.
 Consider the relations $\eta_A$ and $\epsilon_A$ as subsets of $\{\bullet\}\times(A\times A)$ and $(A\times A)\times\{\bullet\}$, respectively.
 Then: 
 \begin{align}
  \eta_A &= \{ (\bullet,(a,a)) \mid a\in A \}, \quad \text{and} \\
  \epsilon_A &= \{ ((a,a),\bullet) \mid a\in A \}.
 \end{align}
\end{ex}

\begin{ex}
 The category $\FHilb$ is compact closed.
 The dual of a Hilbert space $\mathcal{H}$ is taken to be the usual dual, i.e.\ the space of functions $\CC\to\mathcal{H}$.
 For finite-dimensional Hilbert spaces, this makes $\mathcal{H}^*$ isomorphic to $\mathcal{H}$.
 To define the arrows $\eta_\mathcal{H}$ and $\epsilon_\mathcal{H}$, pick an orthonormal basis $\{\ket{i}\}$ for $\mathcal{H}$ and, using the same notation for the corresponding basis of $\mathcal{H}^*$, let:
 \begin{align}
  \eta_\mathcal{H} &= \sum_i \ket{i}\otimes\ket{i}, \quad \text{and} \\
  \epsilon_\mathcal{H} &= \sum_i \bra{i}\otimes\bra{i}.
 \end{align}

 The category $\Hilb$ cannot be given a compact structure because the maps $\eta_\mathcal{H}$ and $\epsilon_\mathcal{H}$ as defined above are not bounded when $\mathcal{H}$ is infinite-dimensional.
\end{ex}

There is a final piece of category-theoretical structure that is useful for describing quantum theory: dagger functors, which are generalisations of the Hermitian adjoint of linear maps.

\begin{dfn}
 A \emph{dagger functor} on a category $\mathcal{C}$ is a functor $(-)^\dagger:\mathcal{C}\to\mathcal{C}$ which acts as the identity on objects, i.e.\ $A^\dagger = A$ for all $A\in Ob(\mathcal{C})$, and satisfies the following conditions on arrows.
 \begin{itemize}
  \item The dagger functor inverts the directions of arrows and of sequential composition:
   \begin{align}
    (f:A\to B)^\dagger &= (f^\dagger:B\to A), \quad \text{and} \\
    (f\circ g)^\dagger &= g^\dagger\circ f^\dagger.
   \end{align}
  \item The dagger functor is involutive, i.e.\ for all arrows $f$ in $\mathcal{C}$:
   \begin{equation}
    (f^\dagger)^\dagger = f.
   \end{equation}
 \end{itemize}
\end{dfn}

The first property, that of inverting the direction of arrows, is also referred to as \emph{contravariance} of the functor.
A compact closed category with a dagger functor that interacts nicely with the parallel composition, the swap arrow, and the compact structure, is called \emph{dagger compact closed}.
This type of category was first introduced in \cite{abramsky_categorical_2004} under the name ``strongly compact closed category''.

\begin{dfn}
 A \emph{dagger compact closed category} is a compact closed category $\mathcal{C}$ with a dagger functor $(-)^\dagger$ satisfying the following conditions.
 \begin{itemize}
  \item The dagger of the parallel composite of two arrows is the same as the parallel composite of the daggers of the two arrows:
   \begin{equation}
    (f\otimes g)^\dagger = f^\dagger \otimes g^\dagger.
   \end{equation}
  \item The dagger of the swap arrow is its inverse:
   \begin{equation}
    \sigma_{A,B}^\dagger = \sigma_{B,A}.
   \end{equation}
  \item The maps associated with the compact structure for an object and its dual object are related to each other via the dagger functor:
   \begin{equation}
    \epsilon_A^\dagger = \eta_{A^*}.
   \end{equation}
 \end{itemize}
\end{dfn}

\begin{ex}
 The category $\FHilb$ is dagger compact closed with Hermitian adjoint of linear maps as the dagger.
 Similarly $\Mat_\CC$, where the dagger is the usual Hermitian adjoint of complex matrices.
\end{ex}

\begin{ex}
 The category $\Rel$ is dagger compact closed with \emph{relational converse} as the dagger.
 For a relation $R\subseteq A\times B$, the relational converse is defined as:
 \begin{equation}
  R^\dagger = \{(b,a) \mid (a,b)\in R\} \subseteq B\times A.
 \end{equation}
\end{ex}

Dagger compact closed categories are the setting for categorical quantum mechanics, i.e.\ the analysis of quantum theory and similar theories as categories.
As laid out in the following sections, dagger compact closed categories are also the appropriate setting for formalising graphical languages based on string diagrams.

\subsection{String diagrams, algebraic equalities, and graph isomorphisms}

In formalising graphical languages, we focus here on languages that represent processes as \emph{string diagrams}: diagrams consisting of boxes, which denote maps, and wires, which denote systems -- or, equivalently, the identity maps on those systems.
Unlike mathematical graphs, wires in string diagrams do not need to be connected to boxes at both ends, or at all.
Quantum circuit notation, Penrose's notation, the Pavia notation, and the \ZX-calculus are examples of this type of graphical language.

There are two steps to the process of making a graphical language rigorous: firstly, one needs to give an explicit translation between diagrams and algebraic terms. Secondly, one needs to prove that two diagrams that seem intuitively equal translate to algebraic terms that are equal.

String diagrams can be translated into algebraic terms by cutting the diagrams up into segments containing exactly one map, using only horizontal and vertical cuts.
We assume this is always possible.
Each map is then represented by a symbol in the term language, and symbols are combined using -- in the quantum case -- tensor product and matrix multiplication.
The correspondence between cut direction and mode of composition depends on the direction of the diagrams; for quantum circuits, horizontal cuts correspond to tensor products and vertical ones to matrix multiplication.
For example, to produce \eqref{eq:circuit_example}, the circuit diagram has been cut up as shown in Figure \ref{fig:cut_example}.

\begin{figure}
 \centering
 \input{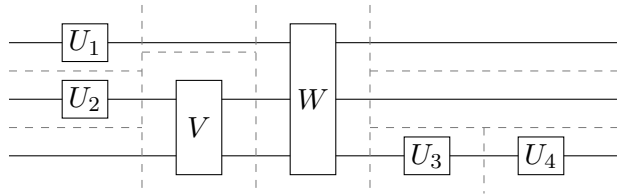}
 \caption{The cuts -- indicated by dashed lines -- that produce \eqref{eq:circuit_example} from Figure \ref{fig:circuit_example}.}
 \label{fig:cut_example}
\end{figure}

To make this correspondence rigorous, one needs to show that the terms resulting from different decompositions of the same diagram are all equal.
This is ensured, among other axioms, by the requirement that maps $f,g,h,j$ in any monoidal category, which are composable in the required way, satisfy the \emph{interchange law} \eqref{eq:interchange}:
\[
 (g\circ f)\otimes(j\circ h) = (g\otimes j) \circ (f\otimes h),
\]
where $\otimes$ denotes parallel composition and $\circ$ denotes sequential composition.
Given this equality, the diagrammatic notation is clearly more intuitive than the algebraic one: the interchange law is tautological in diagrams, as demonstrated in Figure \ref{fig:interchange}, whereas it is not at all obvious in the algebraic notation.

\begin{figure}
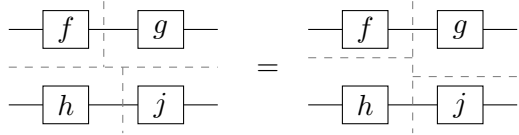

 \centering
 \input{tikz_files/circuit_interchange1.tikz} $\;$ = $\;$ \input{tikz_files/circuit_interchange2.tikz}
 \caption{The interchange law in quantum circuit notation. Dashed lines indicate where the diagrams are ``cut'' to produce the algebraic representation. Other than the different arrangements of cuts, the diagram on the left is identical to the one on the right.}
 \label{fig:interchange}
\end{figure}

The interchange law is not the only equality with this property: there are other category-theoretical structures whose defining equalities are very intuitive when represented diagrammatically.
E.g., equalities involving the swap-operation become very intuitive when the symbol of two crossing lines is used to represent \SWAP{}.
An example of this is shown in Figure \ref{fig:swap}.

\begin{figure}
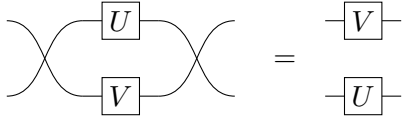

 \centering
 \input{tikz_files/circuit_swap1.tikz} $\;$ = $\;$ \input{tikz_files/circuit_swap2.tikz}
 \caption{In graphical notation, it is intuitively obvious that \SWAP{}$\,\circ(U\otimes V)\circ\,$\SWAP{} is equal to $V\otimes U$ if the symbol for \SWAP{} is a wire crossing.}
 \label{fig:swap}
\end{figure}

This notion of diagram equalities being intuitive even when two diagrams are not identical is covered by the idea of \emph{graph isomorphisms} between string diagrams.

\begin{dfn}
 Two string diagrams are \emph{equal up to graph isomorphism} if there exists a bijection between the two that keeps the inputs and outputs in the same order and assigns to each element of the first diagram an equivalent element of the same diagram such that all elements are connected up in the same way.
\end{dfn}

In particular, two diagrams are isomorphic if, keeping the inputs and outputs fixed, one can be transformed into the other by topological transformations such as crossing or uncrossing wires, stretching or shortening wires, moving nodes along wires or moving nodes around the diagram while keeping their connections the same \cite{coecke_generalised_2015}. 

Another example of diagram equalities corresponding to graph isomorphisms is given by certain diagrams involving the maps $\eta$ and $\epsilon$ representing a compact closed structure (cf.\ Definition \ref{dfn:compact_closed_category}).
Graphically, these maps are represented by curved wires with either two inputs (a  ``cap'') or two outputs (a ``cup'').
The names become clear when they are drawn in a graphical notation where sequential composition is represented vertically.
By the definition of compact structures, the cups and caps obey two equations, which -- when represented diagrammatically -- are usually called the \emph{snake equations}.
While the snake equations are not tautological in string diagrams, they do correspond to simple graph isomorphisms, as can be seen in Figure \ref{fig:snakes}.
On the other hand, the corresponding algebraic expressions are complicated, even more so when the relevant monoidal category is not strict.

There are no cups and caps in quantum circuit notation, but they play an important role in the \ZX-calculus as the (non-normalised) representations of the Bell state $\ket{00}+\ket{11}$ and its adjoint, cf. \eqref{eq:cup_Bell_state}.

\begin{figure}
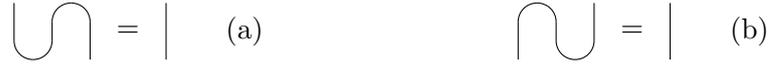

 \centering
 \input{tikz_files/snake1.tikz} $\quad$ (a) $\qquad\qquad\qquad\qquad$ \input{tikz_files/snake2.tikz} $\quad$ (b)
 \caption{The snake equations in a graphical language that is read from bottom to top, e.g.\ the \ZX-calculus.}
 \label{fig:snakes}
\end{figure}

\subsection{Graphical languages and algebraic reasoning in category theory}
\label{s:graphical_languages_algebraic}

As shown in the previous section, there are many equalities in category theory which are much more intuitive when written down diagrammatically.
Yet intuitiveness is not sufficient for the graphical language to replace the algebraic one.
A graphical language is useful only if all the intuitive diagram equalities, i.e.\ equalities that hold up to graph isomorphisms, correspond to true equations in the category.
Furthermore, ideally, all the categorical equations should be intuitive in the diagrammatic language.

In fact, for several different classes of categories, it is possible to define a corresponding graphical language in such a way that:
\begin{itemize}
 \item any equality following directly from the axioms of that category holds up to graph isomorphism in the graphical language, and
 \item if two diagrams are equal up to graph isomorphism, the corresponding algebraic terms are equal by the axioms of the category.
\end{itemize}
Thus the graphical language is entirely equivalent to algebraic reasoning for those categories \cite{selinger_dagger_2007}.

In particular, this result holds for \emph{dagger compact closed categories}.
The graphical language for a dagger category generally has boxes that are asymmetrical, cf.\ Figure \ref{fig:dagger_maps} a.
Taking the dagger corresponds to flipping the diagram upside-down and mirroring the boxes representing the maps, as shown in Figure \ref{fig:dagger_maps} b and Figure \ref{fig:dagger_diagrams}.

\begin{figure}
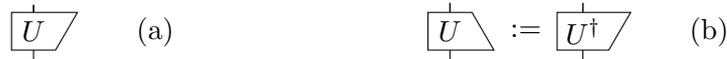

 \centering
 \input{tikz_files/dc_map.tikz} $\quad$ (a) $\qquad\qquad\qquad\qquad$ \input{tikz_files/dc_daggermap.tikz} $\quad$ (b)
 \caption{(a) Graphical representation of a map with one input and one output in a dagger category. (b) The dagger of a map is denoted by an upside-down version of the asymmetrical box.}
 \label{fig:dagger_maps}
\end{figure}

\begin{figure}
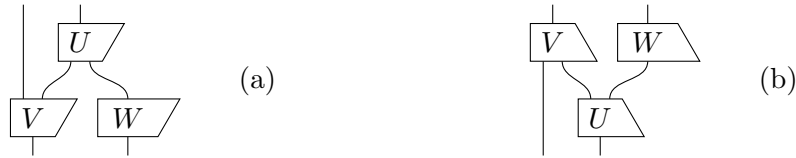

 \centering
 \input{tikz_files/dc_diagram.tikz} $\quad$ (a) $\qquad\qquad\qquad\qquad$ \input{tikz_files/dc_daggerdiagram.tikz} $\quad$ (b)
 \caption{(a) A more complicated diagram in a dagger monoidal category, and (b) its dagger.}
 \label{fig:dagger_diagrams}
\end{figure}

Graphical languages based on category theory are good candidates for intuitive, rigorous, high-level languages for quantum theory.
By \cite{selinger_dagger_2007}, two diagrams in appropriate graphical notations are equal up to isomorphism if and only if the maps they represent are equal up to the axioms of the category.
Yet this is a very general result: while quantum theory can be modelled as a dagger compact closed category, not all equalities between linear maps follow from the axioms of dagger compact closed categories.
In fact, the equations implied by the axioms of dagger compact closed categories are only those that involve the interplay of \SWAP{} maps, cups and caps, and general maps, or the relationship between a map and its dagger.

There are many different dagger compact closed categories which obey the same axioms but are not models of quantum mechanics.
An example is $\Rel$, the category of sets and relations where systems and transformations compose in parallel by Cartesian product. 
This category is in fact the setting for the categorical formulation of Spekkens' toy bit theory in Chapter \ref{s:spekkens}.

The underlying categorical structure is the same for quantum mechanics and for the toy theory, which is a local hidden variable theory.
Thus clearly graph isomorphisms are not sufficient to express all the equalities we are interested in.

\section{Graphical rewriting and properties of formal systems}
\label{s:graphical_rewriting_formal_systems}

String diagrams and graph isomorphisms are equivalent to algebraic reasoning for certain classes of categories, including the one modelling quantum theory.
Yet the equalities that follow directly from the categorical structure are only a small subset of all the equalities making up a theory.
For example, the two circuits in Figure \ref{fig:cNOT-Z} are not equal up to graph isomorphism even though they correspond to the same operator.

\begin{figure}
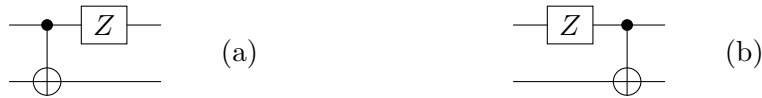

 \centering
 \input{tikz_files/cNOT-Z.tikz} $\quad$ (a) $\qquad\qquad\qquad\qquad$ \input{tikz_files/Z-cNOT.tikz} $\quad$ (b)
 \caption{(a) A controlled-\NOT{} gate, followed by a Z-gate on the control qubit, and (b) a Z-gate followed by a controlled-\NOT{} gate.}
 \label{fig:cNOT-Z}
\end{figure}

To resolve this problem, we introduce graphical rewriting.
The basic idea is the following: given two diagrams $D_1$ and $D_2$ that are equal, whenever a larger diagram contains $D_1$ as a subdiagram, this can be replaced by $D_2$ to get a new diagram equal to the original one.
There are many intricacies associated with the process of matching and replacing subdiagrams, which we ignore here.
Instead we assume that a naive ``cutting'' and ``pasting'' approach works.

Two diagrams cannot be equal unless they have matching inputs and outputs.
In the theories we are considering, there is only one type of basic system: qubits in the case of quantum theory, toy bits for Spekkens' toy theory; therefore wire types do not need to be tracked, all that matters is the number of inputs and outputs of a diagram.

\begin{ex}
 Given the following diagram equation:
 \begin{equation}\label{eq:cNOT-Z}
   \input{tikz_files/cNOT-Z.tikz} \quad = \quad \input{tikz_files/Z-cNOT.tikz} \; ,
 \end{equation}
 we can rewrite a more complicated circuit by first ``cutting out'' the diagram appearing on the left of \eqref{eq:cNOT-Z} and then ``pasting'' the right-hand side of \eqref{eq:cNOT-Z} in that place:
 \begin{equation}
  \input{tikz_files/circuit_cNOT-Z-cNOT.tikz} \quad \mapsto \quad \input{tikz_files/circuit_cNOT-Z-cNOT_cut.tikz} \quad \mapsto \quad \input{tikz_files/circuit_Z-cNOT-cNOT.tikz} \; .
 \end{equation}
\end{ex}

In this way, complicated diagram equalities can be derived from a small set of specified \emph{rewrite rules}.
Unlike most rewrite systems, we do not specify a direction for rewrite rules.
We shall not look at properties that are usually of interest in the context of rewrite systems -- like confluence or termination -- in this thesis.
Instead, we consider the rewrite rules as axioms in a formal system and look at some associated properties.

The properties of graphical languages when considered as formal systems all depend on the interpretation, i.e.\ the map from diagrams to the underlying theory.
All the graphical languages we are considering here have been constructed with a specific interpretation in mind; usually, this is takes the form of a specific map -- in fact, a functor -- from the graphical language to the matrix formulation of quantum theory.
Nevertheless, different interpretations could be chosen, including both interpretations in other theories and different maps from the graphical language into matrices.
We denote the standard interpretation functor for a graphical language by $\intf{-}$.

\begin{ex}
 The interpretation functor for a quantum circuit over the Clifford+T gate set is defined as follows:
 \begin{align}
  \intf{\sqGate{$H$}} &= \frac{1}{\sqrt{2}}\begin{pmatrix} 1 & 1 \\ 1 & -1 \end{pmatrix}, \\
  \intf{\sqGate{$S$}} &= \begin{pmatrix} 1 & 0 \\ 0 & i \end{pmatrix}, \\
  \intf{\sqGate{$T$}} &= \begin{pmatrix} 1 & 0 \\ 0 & e^{i\pi/4} \end{pmatrix}, \text{and} \\
  \intf{\input{tikz_files/circuit_cNOT.tikz}} &= \begin{pmatrix} 1 & 0 & 0 & 0 \\ 0 & 1 & 0 & 0 \\ 0 & 0 & 0 & 1 \\ 0 & 0 & 1 & 0 \end{pmatrix},
 \end{align}
 with:
 \begin{equation}
  \intf{\input{tikz_files/circuit_parallel.tikz}} = \intf{\input{tikz_files/circuit_U.tikz}} \otimes \intf{\input{tikz_files/circuit_V.tikz}}
 \end{equation}
 and:
 \begin{equation}
  \intf{\input{tikz_files/circuit_series.tikz}} = \intf{\input{tikz_files/circuit_V.tikz}} \circ \intf{\input{tikz_files/circuit_U.tikz}}
 \end{equation}
 for any circuits $U,V$ that can be composed in the required manner.
\end{ex}

\subsection{Universality}
\label{s:universality}

For a graphical language to be a viable alternative to conventional formalisms, it needs to be able to express any idea that can be expressed in the conventional formalism.
This notion is captured by the concept of \emph{universality}.
\begin{dfn}
 A graphical language is \emph{universal} under the interpretation $\intf{-}$ if for any process $P$ in the underlying theory, there exists a diagram $D$ such that:
 \begin{equation}
  \intf{D} = P.
 \end{equation}
\end{dfn}

\begin{ex}
 Quantum circuit notation with a gate set consisting of controlled-\NOT{} and arbitrary single-qubit gates is universal for unitary operations on qubits, as any unitary operator on qubits can be decomposed into single-qubit gates and controlled-\NOT{}s \cite{nielsen_quantum_2010}.
\end{ex}

A related notion that comes up in the context of gate sets in quantum computing is that of \emph{approximate universality}.
\begin{dfn}\label{dfn:approximately_universal}
 A gate set for quantum computation is called \emph{approximately universal} if any unitary operation can be approximated to arbitrary accuracy using only operations from the given set.
\end{dfn}
This is not in itself a property of graphical languages as most graphical formalisms do not have an intrinsic concept of approximation.
Nevertheless, it can be very illuminating to e.g.\ consider quantum circuits over approximately universal gate sets, even if any steps directly involving approximation have to be done by translating to matrices and then back again.

An example of such an approximately universal gate set is the Clifford+T group introduced in Section \ref{s:quantum_circuits}.

\subsection{Soundness}
\label{s:soundness}

The property of universality does not involve the rewrite rules associated with a graphical language in any way: it is a property of the notation only.
Other properties do involve the rewrite rules.
Possibly the most important property of a set of rewrite rules is that they allow only the derivation of true equalities, i.e.\ equalities that also hold in the conventional formalism.

\begin{dfn}
 Suppose $D_1=D_2$ is an equation in some graphical language, where $D_1$ and $D_2$ are diagrams. This equation is \emph{sound} under the interpretation $\intf{-}$ if:
 \begin{equation}
  \intf{D_1} = \intf{D_2}.
 \end{equation}
 A system of rewriting rules is sound for a given interpretation if all the individual rewrite rules are sound equations.
\end{dfn}

Note that an equation that is sound under some interpretation $\intf{-}$ may not be sound under alternative interpretations.

\begin{ex}
 The equality:
 \begin{equation}
  \input{tikz_files/Z-cNOT.tikz} \quad = \quad \input{tikz_files/cNOT-Z.tikz}
 \end{equation}
 is sound for quantum circuit diagrams with the usual interpretation, as:
 \begin{equation}
  C_X \circ (Z\otimes I) = (Z\otimes I) \circ C_X.
 \end{equation}
\end{ex}

\begin{ex}
 Introduce a new gate symbol \sqGate{$\diamond$}.
 With the usual interpretation for \sqGate{$Z$}, the equation:
 \begin{equation}
  \sqGate{$\diamond$}\!\sqGate{$Z$} \; = \; \sqGate{$Z$}\!\sqGate{$\diamond$}
 \end{equation}
 is sound if the interpretation of the diamond gate is:
 \begin{equation}
  \intf{\sqGate{$\diamond$}} = \begin{pmatrix} 1 & 0 \\ 0 & e^{i\pi/4} \end{pmatrix},
 \end{equation}
 but not if it is:
 \begin{equation}
  \intf{\sqGate{$\diamond$}} = \frac{1}{\sqrt{2}}\begin{pmatrix} 1 & 1 \\ 1 & -1 \end{pmatrix}.
 \end{equation}
\end{ex}

\subsection{Completeness}

A sound rewriting system ensures that only true equalities can be derived.
Relatedly, a graphical formalism is most useful if it allows \emph{all} true equalities to be derived graphically: while it is certainly possible to translate back to the conventional formalism to derive equalities, that is exactly what we are trying to avoid by introducing graphical rewrite rules in the first place.

\begin{dfn}
 A graphical rewrite system is \emph{complete} for an interpretation $\intf{-}$ if for any diagrams $D_1$ and $D_2$:
 \begin{equation}
  \intf{D_1} = \intf{D_2} \implies D_1 = D_2,
 \end{equation}
 i.e.\ any equality that can be derived using the conventional formalism can also be derived using the rewrite rules.
\end{dfn}

Like soundness, the property of completeness of a graphical rewrite system is specific to the interpretation used.
Unlike soundness, completeness depends on all the rewrite rules together and cannot be checked on a rule-by-rule basis.
It is therefore usually harder to check completeness than it is to check soundness.

It is possible to construct complete systems of rewrite rules for fragments of quantum circuits.
In \cite{selinger_generators_2013}, Selinger gives a complete set of rewrite rules for stabilizer quantum circuit diagrams built from Hadamard, phase, and controlled-\NOT{} gates.

We prove completeness results for fragments of the \ZX-calculus in Chapters \ref{ch:completeness} and \ref{ch:more_completeness}.

\section{Automated graphical reasoning}
\label{s:automated_graphical_reasoning}

While large diagrams are easier to understand than large matrices, they are still complicated and manipulating them is therefore error-prone.
This problem can be resolved by automating or semi-automating the graphical reasoning process.
The software system Quantomatic \cite{quantomatic} does just that for graphical languages based on string diagrams, including but not limited to the \ZX-calculus.

Using a specified set of rewrite rules, Quantomatic suggests possible rewrites for the current diagram that the user can then choose to apply, cf.\ Figure \ref{fig:quantomatic}.
Alternatively, with higher-level rewrite strategies, Quantomatic can rewrite diagrams automatically.
Soundness of the rewrite rules means Quantomatic can only derive true equalities.

The existence of a complete graphical language based on string diagrams for some physical theory means that this theory can be explored entirely using automated graphical reasoning.

\begin{figure}
 \centering
 \includegraphics[width=\textwidth]{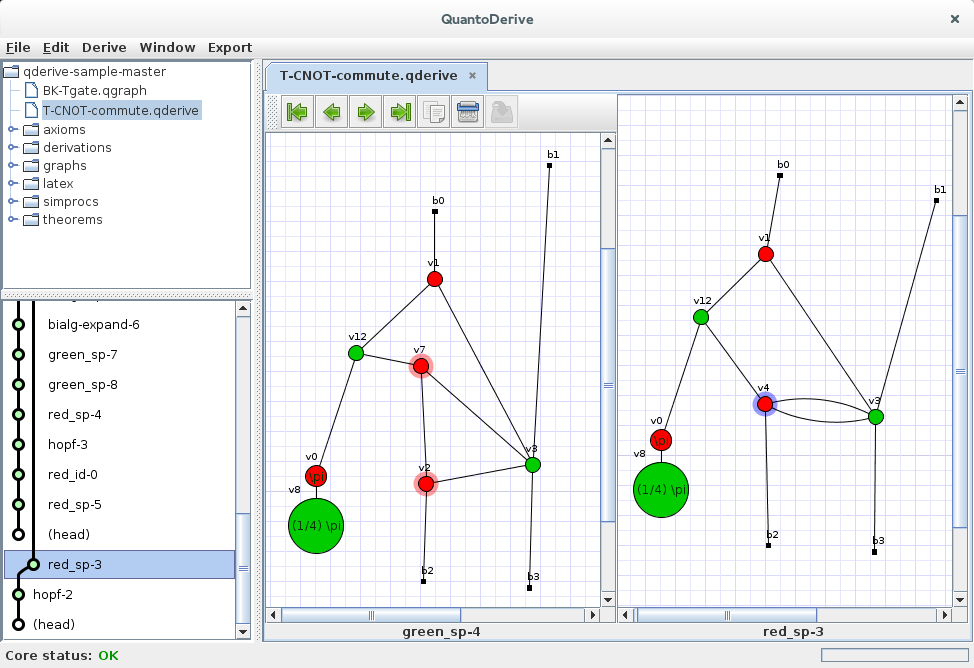}
 \caption{Screenshot of a \ZX-calculus rewrite step in \texttt{Quantomatic}: The red nodes v7 and v2 highlighted in the left diagram are merged into one highlighted red node labelled v4 in the right diagram.}
 \label{fig:quantomatic}
\end{figure}

%% file: ZX-calculus.tex
\chapter{The \ZX-calculus}
\label{s:ZX-calculus}

Having given an introduction to rigorous graphical languages in general in the previous chapter, we now focus on the \ZX-calculus.
The \ZX-calculus is a graphical language for pure state qubit quantum mechanics with post-selected measurements that comes with a set of built-in rewrite rules.
It was first introduced by Coecke and Duncan \cite{coecke_interacting_2008,coecke_interacting_2011}, based on the categorical quantum mechanics approach pioneered by Abramsky and Coecke \cite{abramsky_categorical_2004}.

In this chapter, we explain the \ZX-calculus notation and the standard interpretation for diagrams, and show that the calculus is universal.
We furthermore give the rewrite rules of the \ZX-calculus and argue that they are sound.

The \ZX-calculus completeness results in Chapters \ref{ch:completeness} and \ref{ch:more_completeness} are not for full pure state qubit quantum mechanics but for pure state qubit stabilizer quantum mechanics and the single-qubit Clifford+T group.
These fragments of pure state quantum theory and their \ZX-calculus representations are introduced in Sections \ref{s:stabilizer_QM} and \ref{s:Clifford+T_group}, respectively.

\section{The \ZX-calculus notation}
\label{s:ZX_notation}

Diagrams in the \ZX-calculus are string diagrams, consisting of wires and labelled nodes.
Where quantum circuit diagrams are read from left to right, \ZX-calculus diagrams are read from bottom to top, i.e.\ sequential operations are stacked vertically while parallel operations are put side-by-side.
Thus wires that end at the bottom edge of a diagram are inputs, wires that end at the top of a diagram are outputs.

In this section, we give the basic elements of the \ZX-calculus and define the interpretation functor for \ZX-calculus diagrams.
We then introduce some common terminology for describing \ZX-calculus diagrams.
Finally, we prove that the \ZX-calculus as defined here is universal for pure state qubit quantum mechanics with post-selected measurements.

\subsection{Basic elements of \ZX-calculus diagrams}
\label{s:ZX_elements}

There are some slight variations in the generating elements used in different versions of the \ZX-calculus, here we use the following generators shown in Figure \ref{fig:ZX-generators}:
\begin{itemize}
 \item wires,
 \item green nodes with arbitrary numbers of inputs and outputs and a label $\alpha\in(-\pi,\pi]$ called the \emph{phase},
 \item red nodes with arbitrary numbers of inputs and outputs and a phase label $\beta\in(-\pi,\pi]$,
 \item yellow square nodes, which always have one input and one output, and
 \item black star-shaped nodes, which do not have any inputs or outputs. 
\end{itemize}
Yellow and black nodes do not have phase labels.
The arbitrary number of inputs and outputs for green and red nodes includes the possibility of having no inputs, no outputs, or no connecting wires at all.
The green and red nodes are usually called \emph{spiders} due to their many ``legs''.
A phase label of $0$ is usually left implicit, e.g.\ \joinnode{} is the same as \input{tikz_files/joinnode_0.tikz}.

\begin{figure}
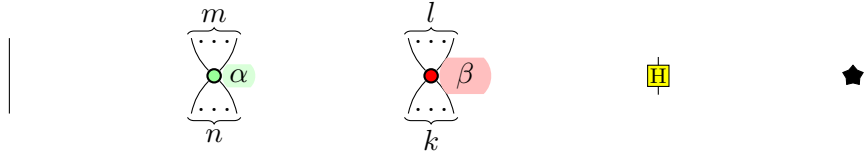

 \centering
 \input{tikz_files/wire.tikz} $\quad\quad\quad\quad\quad$ \input{tikz_files/green_spider.tikz} $\quad\quad\quad\quad$ \input{tikz_files/red_spider.tikz} $\quad\quad\quad\quad$ \Hadamard{} $\quad\quad\quad\quad\quad$ \halfscalar{}
 \caption{The basic elements of the \ZX-calculus.}
 \label{fig:ZX-generators}
\end{figure}

Diagrams are built from these basic components by either putting them side-by-side or by connecting some inputs of one component to some outputs of another.
Wires are allowed to cross and bend.

\subsection{How to interpret diagrams}
\label{s:ZX_interpretation}

\ZX-calculus diagrams can be considered as arrows in a dagger compact closed category, which is constructed as follows: The objects of this category are non-negative integers, corresponding to the number of input or output wires of a diagram.
Sequential composition is done by arranging the diagrams vertically and connecting the outputs of the first diagram to the inputs of the second diagram.
Parallel composition of objects is addition; parallel composition of arrows is performed by putting diagrams side-by-side.
The swap arrow is given by the diagram:
\begin{center}
 \input{tikz_files/swap.tikz}.
\end{center}
The compact structure on the object 1 consists of the arrows:
\begin{center}
 \input{tikz_files/cup_diagram.tikz} and \input{tikz_files/cap_diagram.tikz},
\end{center}
and the compact structure on 0 consists of two empty diagrams.
For objects $n$ with $n>1$, the compact structure consists of $n$ nested cups or caps, respectively.
The dagger functor flips diagrams upside-down and simultaneously maps any phase label $\phi$ to $-\phi$.
This is different to the general way of taking the dagger in graphical languages, cf. Section \ref{s:graphical_languages_algebraic}, because it is unnecessary to distinguish between inputs and outputs of spiders.
It is straightforward to check that these maps satisfy the definition of a dagger compact closed category given in Section \ref{s:graphical_rigorous}.

The interpretation of a \ZX-calculus diagram as a matrix is the image of this diagram under a functor from the category of \ZX-calculus diagrams into $\Mat_\CC$, the strict monoidal category corresponding to the category of finite-dimensional Hilbert spaces and linear maps defined in Example \ref{ex:Mat_CC}

The \ZX-calculus was developed with a specific interpretation functor in mind.
In the following, we denote this functor by $\intf{-}$ and sometimes call it the ``usual interpretation functor'' to distinguish it from other interpretation functors that will be constructed later.

\begin{dfn}\label{dfn:interpretation_functor}
 The usual interpretation functor for the \ZX-calculus is denoted by $\intf{-}$.
 For objects, $\intf{n}=2^n$, thus a diagram with $n$ inputs and $m$ outputs represents a $2^m$ by $2^n$ matrix.
 A diagram with no inputs or outputs denotes a 1 by 1 matrix, which is simply a complex number.
 The action of the interpretation functor on arrows is determined by its action on the generators of \ZX-calculus diagrams together with its behaviour under the different types of composition.
 The basic elements of \ZX-calculus diagrams are interpreted as follows, where for reasons of legibility the matrices are written in bra-ket notation:
 \begin{align}
  \left\llbracket \;\;\input{tikz_files/wire.tikz}\;\; \right\rrbracket &= \ket{0}\bra{0} + \ket{1}\bra{1}, \\
  \left\llbracket \input{tikz_files/green_spider.tikz} \right\rrbracket &= \ket{0}\t{m}\bra{0}\t{n} + e^{i\alpha}\ket{1}\t{m}\bra{1}\t{n}, \\
  \left\llbracket \input{tikz_files/red_spider.tikz} \right\rrbracket &= \ket{+}\t{l}\bra{+}\t{k} + e^{i\beta}\ket{-}\t{l}\bra{-}\t{k}, \\
  \left\llbracket \Hadamard \right\rrbracket &= \ket{+}\bra{0} + \ket{-}\bra{1}, \text{ and} \\
  \left\llbracket \halfscalar \right\rrbracket &= \frac{1}{2}.
 \end{align}
 Here, $\ket{\pm}=\frac{1}{\sqrt{2}}\left(\ket{0}\pm\ket{1}\right)$ and the zero-fold tensor product of any normalised bra or ket is taken to be 1.
 A wire crossing is a \SWAP{} operation:
 \begin{equation}
  \left\llbracket \input{tikz_files/swap.tikz} \right\rrbracket = \ket{00}\bra{00} + \ket{10}\bra{01} + \ket{01}\bra{10} + \ket{11}\bra{11},
 \end{equation}
 and cups and caps correspond to (non-normalised) Bell states and effects, respectively:
 \begin{equation}\label{eq:cup_Bell_state}
  \left\llbracket \input{tikz_files/cup_diagram.tikz} \right\rrbracket = \ket{00} + \ket{11} \qquad\text{and}\qquad  \left\llbracket \input{tikz_files/cap_diagram.tikz} \right\rrbracket = \bra{00} + \bra{11}.
 \end{equation}
 Let:
 \begin{center}
  \gendiagram{$D$} $\quad$ and $\quad$ \gendiagram{$D'$}
 \end{center}
 denote two arbitrary diagrams. Then:
 \begin{equation}
  \left\llbracket \gendiagram{$D$} \; \gendiagram{$D'$} \right\rrbracket = \left\llbracket \gendiagram{$D$} \right\rrbracket \otimes \left\llbracket \gendiagram{$D'$} \right\rrbracket,
 \end{equation}
 and, assuming the number of outputs of $D$ is equal to the number of inputs of $D'$:
 \begin{equation}
  \left\llbracket \input{tikz_files/composite_diagram.tikz} \right\rrbracket = \left\llbracket \gendiagram{$D'$} \right\rrbracket \circ \left\llbracket \gendiagram{$D$} \right\rrbracket.
 \end{equation}
 For wires, \Hadamard, and \halfscalar{}, the dagger is identical to the original arrow.
 For green spiders, we have:
 \begin{equation}\label{eq:dagger}
  \left\llbracket \input{tikz_files/green_spider.tikz} \right\rrbracket^\dagger = \left\llbracket \input{tikz_files/green_spider_adjoint.tikz} \right\rrbracket,
 \end{equation}
 and similarly for red spiders.
\end{dfn}

A green or red spider with no legs and phase $\pi$ represents the 1 by 1 zero matrix:
\begin{equation}
 \left\llbracket \scalar{gn, label={[gphase]right:$\pi$}} \right\rrbracket = 0 = \left\llbracket \scalar{rn, label={[rphase]right:$\pi$}} \right\rrbracket.
\end{equation}
An empty diagram on the other hand represents the 1 by 1 identity matrix:
\begin{equation}
 \llbracket\quad\rrbracket = 1.
\end{equation}

For spiders with legs, there is no strict distinction between inputs and outputs as the former can be transformed into the latter, or conversely, by simply bending the wire, e.g.:
\begin{equation}
 \left\llbracket \input{tikz_files/green_spider_bend.tikz} \right\rrbracket = \left\llbracket \input{tikz_files/green_spider_bend2.tikz} \right\rrbracket.
\end{equation}

\subsection{Terminology for \ZX-calculus diagrams}
\label{s:ZX_terminology}

We have already introduced some terminology for specific elements of \ZX-calculus diagrams; here we define additional commonly-used terms.

\begin{dfn}
 A red or green spider with exactly one input and one output is called a \emph{phase shift}.
\end{dfn}

Phase shifts are the only unitary spiders, e.g.:
\begin{equation}
 \intf{ \phase{gn,label={[gphase]right:$\alpha$}} } = \begin{pmatrix} 1 & 0 \\ 0 & e^{i\alpha} \end{pmatrix}.
\end{equation}
The yellow nodes are also called \emph{Hadamard nodes}.

\begin{dfn}
 A \emph{state diagram} is a \ZX-calculus diagram with no inputs and a non-zero number of outputs.
\end{dfn}

This agrees with the usual definition of quantum states because diagrams with no inputs and some non-zero number $n$ of outputs represent linear maps from $\CC$ to $\left(\CC^2\right)\t{n}$, which are in one-to-one correspondence with (not necessarily normalised) pure quantum states.
Similarly, diagrams with a non-zero number of inputs and no outputs are in one-to-one correspondence with (not necessarily normalised) outcomes of pure projective measurements.

\begin{dfn}
 An \emph{effect} is a \ZX-calculus diagram with no outputs and a non-zero number of inputs.
\end{dfn}

\begin{dfn}
 A \ZX-calculus diagram with no inputs and no outputs is called a \emph{scalar} or a \emph{scalar diagram}.
\end{dfn}

The \ZX-calculus thus provides a unified notation in which states, operators, and measurement effects stand on equal footing.

The above definitions also apply to parts of larger diagrams, for example we often consider the \emph{scalar part} of a diagram, by which we mean any fragments of the diagram that are disconnected from any inputs or outputs of the diagram as a whole.

\begin{ex}
 The scalar part of the following diagram:
 \begin{equation}
  \innerprodgr{} \; \input{tikz_files/controlled-NOT_acute.tikz}
 \end{equation}
 is \innerprodgr.
\end{ex}

\subsection{Universality of the \ZX-calculus}
\label{s:ZX_universal}

In Section \ref{s:universality}, we gave the definition of universality for a graphical language, which requires that any process in the underlying theory can be represented graphically.
We show that the \ZX-calculus is universal by translating a universal set of quantum gates into the \ZX-calculus.

It is well-known that quantum circuits built from controlled-\NOT{} gates and single-qubit gates are universal for unitary operators on qubits \cite{nielsen_quantum_2010}.
Furthermore, any pure quantum state on $n$ qubits can be constructed by applying some unitary operator to $n$ copies of the state $\ket{0}$.
A similar result holds for pure measurement effects.

\begin{lem}
 The \ZX-calculus with the interpretation $\intf{-}$ as given in Definition \ref{dfn:interpretation_functor} is universal for pure state qubit quantum mechanics.
\end{lem}
\begin{proof}
 The normalised computational basis states are represented by the following \ZX-calculus diagrams:
 \begin{equation}
  \ket{0} = \intf{\halfscalar \; \innerprodgr \; \state{rn}} \quad \text{and} \quad \ket{1} = \intf{\halfscalar \; \innerprodgr \; \state{rn,label={[rphase]right:$\pi$}}},
 \end{equation}
 and the controlled-\NOT{} gate $C_X$ by:
 \begin{equation}\label{eq:cNOT_acute}
  C_X = \intf{\innerprodgr \; \input{tikz_files/controlled-NOT_acute.tikz}}.
 \end{equation}
 According to the Euler decomposition rule for unitary operators, any single-qubit unitary can be expressed, up to global phase, as a product of three rotations about two different axes.
 The red and green phase shifts are such rotations.
 Furthermore, a complex phase can be represented in the \ZX-calculus as:
 \begin{equation}
  e^{i\phi} = \intf{\halfscalar \; \innerprodgr \; \innerprod{gn,label={[gphase]right:$\phi$}}{rn,label={[rphase]right:$\pi$}}}.
 \end{equation}
 Thus, for any single-qubit unitary $U$, there exist angles $\alpha,\beta,\gamma,\phi\in(-\pi,\pi]$ such that:
 \begin{equation}
  U = \left\llbracket \halfscalar \; \innerprodgr \; \innerprod{gn,label={[gphase]right:$\phi$}}{rn,label={[rphase]right:$\pi$}} \; \input{tikz_files/Euler_dec_general.tikz}\right\rrbracket.
 \end{equation}
 Hence any pure quantum operator on qubits can be represented in the \ZX-calculus by expressing it as a unitary operator sandwiched between computational basis states and effects, and decomposing the unitary operator into a circuit made of controlled-\NOT{} and single-qubit gates.
\end{proof}

Alternatively, computations in the measurement-based quantum computing (MBQC) paradigm can also be straightforwardly translated into \ZX-calculus diagrams if an extra piece of notation is added to keep track of the propagating error corrections \cite{duncan_rewriting_2010}.

\section{Rewrite rules}
\label{s:rewrite_rules}

The \ZX-calculus, as introduced in the previous section, is a rigorous graphical language for a dagger compact closed category, cf.\ Section \ref{s:graphical_languages_algebraic}.
Therefore two \ZX-calculus diagrams that are equal up to graph isomorphism automatically represent the same transformation.
Nevertheless, as explained in Section \ref{s:graphical_rewriting_formal_systems}, not all equalities satisfied by linear maps between Hilbert spaces follow from the dagger compact closed structure.

It thus becomes necessary to introduce additional rewrite rules that allow the derivation of more specific equalities between \ZX-calculus diagrams.

We first introduce some notational conventions and a ``meta-rule''.
The explicit graphical rewrite rules of the \ZX-calculus are given in Section \ref{s:ZX_rules}.
Section \ref{s:ZX_derived} contains some additional rewrite rules which can be derived from the ones in the previous section but which are used commonly enough that they merit being stated in their own right.
Finally, we argue that the rewrite rules given here are sound.

\subsection{Meta-rules and notational conventions}
\label{s:ZX_conventions}

Due to the symmetry between red and green spiders, as well as that between diagrams and their Hermitian adjoints, we shorten the list of explicitly-stated rules by adopting the following principles:
\begin{itemize}
 \item Any rewrite rule can also be applied with the colours red and green swapped.
 \item Any rewrite rule can also be applied upside-down.
\end{itemize}
These conventions are consistent with the interpretation map, as shown in Section \ref{s:soundness}.

Furthermore, diagrams of the \ZX-calculus obey the following meta-rule:
\begin{quotation}
 ``Only the topology matters.''
\end{quotation}
This ``topology rule'' is the explicit statement of the fact that two diagrams are the same whenever they are equal up to graph isomorphism (cf.\ Section \ref{s:graphical_languages_algebraic}), together with the fact that inputs and outputs of spiders do not need to be distinguished.
The latter can be derived from the spider and cup rules below.

\begin{ex}
 One application of the topology rule is the fact that scalar subdiagrams can be placed anywhere in a larger diagram:
 \begin{equation}
  \halfscalar \; \state{gn} = \input{tikz_files/scalar_plus.tikz} = \state{gn} \; \halfscalar.
 \end{equation}
\end{ex}

\begin{ex}
 Components can be moved around as long as the connections -- including connections to inputs and outputs -- remain the same:
 \begin{equation}
  \state{rn} \; \state{gn} = \input{tikz_files/states_swapped.tikz} \neq \state{gn} \; \state{rn}.
 \end{equation}
 For the first equality, the first output is connected to the red node and the second to the green node in both diagrams, this means they are equal independent of the exact placement of the red and green nodes.
 The third diagram on the other hand is not equal to the first two because in it the first output is connected to the green node and the second to the red node.
 This is consistent with the interpretation map:
 \begin{equation}
  \left\llbracket \state{rn} \; \state{gn} \right\rrbracket = 2\ket{0}\otimes\ket{+} = \left\llbracket \input{tikz_files/states_swapped.tikz} \right\rrbracket,
 \end{equation}
 whereas:
 \begin{equation}
  \left\llbracket \state{gn} \; \state{rn} \right\rrbracket = 2\ket{+}\otimes\ket{0}.
 \end{equation}
\end{ex}

As an extension of the topology meta rule, wires internal to a diagram can be drawn as horizontal lines if this does not introduce any ambiguity as to the interpretation of the diagram as a whole.
For example, consider the following equality:
\begin{equation}
 \left\llbracket \input{tikz_files/controlled-NOT_grave.tikz} \right\rrbracket = \left\llbracket \input{tikz_files/controlled-NOT_acute.tikz} \right\rrbracket = \left\llbracket \input{tikz_files/controlled-NOT_circonflexe.tikz} \right\rrbracket = \left\llbracket \input{tikz_files/controlled-NOT_vee.tikz} \right\rrbracket,
\end{equation}
i.e.\ it does not matter whether the wire connecting the two spiders is considered to be an input or output of the green spider, and an input or output of the red spider.
The diagram:
\begin{equation}
 \innerprodgr \; \input{tikz_files/controlled-NOT.tikz}
\end{equation}
therefore has a well defined interpretation -- choosing any orientation of the internal wire gives the same result -- and is hence allowed as a component of \ZX-calculus diagrams.
In fact, this diagram is often used instead of \eqref{eq:cNOT_acute} to represent the controlled-\NOT{} gate.

This bit of notational convention applies only to internal wires, i.e.\ wires that are connected to nodes on both ends.
For dangling wires it must always be clear whether the dangling end is an input or an output of the diagram, otherwise the domain and target of the diagram are not well-defined and it thus cannot be an arrow in the category of ZX-calculus diagrams.

\subsection{Explicit rewrite rules}
\label{s:ZX_rules}

There are different ways of stating the rewrite rules of the \ZX-calculus, and several new rules have been added since it was first introduced.
In this work, we use the following set of rules, where $\alpha,\beta\in(-\pi,\pi]$, addition of phases is modulo $2\pi$, and $n,m,k,l$ are non-negative integers:
\begin{itemize}
 \item the spider rule:
  \begin{equation}
   \input{tikz_files/spider.tikz},
  \end{equation}
 \item the loop rule:
  \begin{equation}
   \input{tikz_files/loop.tikz},
  \end{equation}
 \item the cup rule:
  \begin{equation}
   \input{tikz_files/cup.tikz},
  \end{equation}
 \item the bialgebra rule:
  \begin{equation}
   \input{tikz_files/bialgebra_scalars.tikz},
  \end{equation}
 \item the copy rule:
  \begin{equation}
   \input{tikz_files/copy_scalars.tikz}\;,
  \end{equation}
 \item the $\pi$-copy rule:
  \begin{equation}
   \input{tikz_files/pi-copy.tikz},
  \end{equation}
 \item the $\pi$-commutation rule:
  \begin{equation}
   \input{tikz_files/pi-commutation_scalars.tikz},
  \end{equation}
 \item the colour change rule:
  \begin{equation}
   \input{tikz_files/colour.tikz},
  \end{equation}
 \item the Euler decomposition rule:
  \begin{equation}\label{eq:Euler_dec}
   \input{tikz_files/Euler_dec_scalars.tikz},
  \end{equation}
 \item the star rule:
  \begin{equation}
   \halfscalar \; \scalar{gn} = \quad,
  \end{equation}
 \item the zero rule:
  \begin{equation}
   \input{tikz_files/zero.tikz}\;,
  \end{equation}
 \item and the zero scalar rule:
  \begin{equation}
   \scalar{gn, label={[gphase]right:$\pi$}} \; \scalar{gn, label={[gphase]right:$\alpha$}} = \scalar{gn, label={[gphase]right:$\pi$}}.
  \end{equation}
\end{itemize}

The rewrite rules for the \ZX-calculus are not directed, i.e.\ in all cases the left-hand side (LHS) can be used to replace the right-hand side (RHS), or the RHS can be used to replace the LHS.
Note that rules with varying numbers of inputs and/or outputs also hold when any of those numbers are zero.
\begin{ex}
 The $\pi$-copy rule with $m=0$ yields:
 \begin{equation}
  \input{tikz_files/zero-effect_Z.tikz} \; = \; \effect{rn},
 \end{equation}
 and the colour change rule for $n=0=m$ is:
 \begin{equation}
  \scalar{rn, label={[rphase]right:$\alpha$}} \; = \; \scalar{gn, label={[gphase]right:$\alpha$}}.
 \end{equation}
\end{ex}

The first eight rewrite rules (from the spider rule up to and including the colour change rule) are part of the original formulation of the \ZX-calculus \cite{coecke_interacting_2008}, though some of them appear here with slight modifications.
A scalar-free version of the Euler decomposition rule was first introduced in \cite{duncan_graph_2009}, where it was shown to be independent of the previously-existing \ZX-calculus rewrite rules.
The zero rule was suggested by Kissinger \cite{kissinger_communication_2014} in the context of the scalar-free \ZX-calculus, where this rule is sufficient for the derivation of a normal form for zero diagrams.
The star rule and the zero scalar rule were introduced by the author \cite{backens_making_2015}.

\subsection{Derived rewrite rules}
\label{s:ZX_derived}

Some additional rules which are often included in lists of \ZX-calculus rewrite rules in their own right, can in fact be derived from the rules given in the previous section.
\begin{itemize}
 \item The \emph{identity rule}:
  \begin{equation}
   \input{tikz_files/identity.tikz}
  \end{equation}
  follows from the spider rule, the cup rule, the upside-down cup rule, and the topology meta-rule:
  \begin{equation}
   \input{tikz_files/identity_derivation.tikz}\;.
  \end{equation}
 \item The Hadamard node can be shown to be self-inverse using the colour change rule for $n=m=1$ and the above identity rule, as well as a colour-swapped version of the identity rule:
  \begin{equation}
   \input{tikz_files/Had_inv_derivation.tikz}\;.
  \end{equation}
 \item The \emph{Hopf rule}:
  \begin{equation}
   \input{tikz_files/Hopf_scalars.tikz}
  \end{equation}
  follows from the cup and spider rules and their upside-down and/or colour-changed equivalents, the bialgebra rule, and the copy rule, in combination with the topology meta-rule:
  \begin{equation}
   \input{tikz_files/Hopf_derivation_1.tikz} \;\; = \;\; \input{tikz_files/Hopf_derivation_1-5.tikz} \;\; = \;\; \input{tikz_files/Hopf_derivation_2.tikz} \;\; = \;\; \input{tikz_files/Hopf_derivation_3.tikz} \;\; = \;\; \input{tikz_files/Hopf_derivation_4.tikz} \;\; = \;\; \input{tikz_files/Hopf_derivation_5.tikz}\;.
  \end{equation}
\end{itemize}
Furthermore, while the star rule is defined using \scalar{gn}, this scalar does not actually appear very often in the other rewrite rules.
The following variant of the star rule is therefore useful.

\begin{lem}\label{lem:variant_star_rule}
 The star node is inverse to \innerprodgr{} \innerprodgr{} :
 \begin{equation}\label{eq:halfscalar_innerprodgr2}
  \halfscalar \; \innerprodgr \; \innerprodgr \; = \quad .
 \end{equation}
\end{lem}
\begin{proof}
 Using the star rule, loop rule, cup rule, spider rule, and Hopf rule, we have:
 \begin{equation}
  \innerprodgr \; \innerprodgr \;
  = \; \halfscalar \; \innerprodgr \; \innerprodgr \; \scalar{gn} \;
  = \; \halfscalar \; \innerprodgr \; \innerprodgr \; \input{tikz_files/inverse_der4.tikz} \;
  = \; \halfscalar \; \innerprodgr \; \innerprodgr \; \input{tikz_files/inverse_der3.tikz} \;
  =  \; \halfscalar \; \innerprodgr \; \innerprodgr \; \input{tikz_files/inverse_der2.tikz} \;
  =  \; \halfscalar \; \input{tikz_files/inverse_der1.tikz} \;
  = \; \halfscalar \; \scalar{gn} \; \scalar{rn} \;
  = \; \scalar{gn}.
 \end{equation}
 Thus the desired equality follows directly from the star rule.
\end{proof}

Of course any equality derivable from the \ZX-calculus rewrite rules can potentially be used as a rewrite rule in its own right.
We derive further rewrite rules in later sections as and when they become relevant.

\subsection{Soundness of the rewrite rules}
\label{s:ZX_sound}

For most of the rewrite rules listed in Section \ref{s:rewrite_rules}, it is straightforward to check that they are sound under the interpretation given in Definition \ref{dfn:interpretation_functor} by computing the matrices corresponding to the two diagrams making up the equality.
Soundness of rules involving arbitrary number of inputs and outputs can be proved by induction over those numbers.
For soundness of the spider rule, see \cite{coecke_interacting_2011}.

The topology meta rule is sound because of the properties of dagger compact closed categories, as described in Section \ref{s:graphical_languages_algebraic}.

The convention that all rules hold with the colours red and green swapped is justified by the colour change rule.
The convention that any rewrite rule can be applied upside-down is justified by the fact that taking the adjoint of both sides of an equation preserves the equality, together with the fact that all rewrite rules continue to be sound when the signs of all angles are flipped.
This latter fact can easily be confirmed for all of the rewrite rules by checking their interpretations.

\section{Stabilizer quantum mechanics}
\label{s:stabilizer_QM}

Stabilizer quantum mechanics (QM) is an extensively studied part of quantum theory, first introduced in the context of error-correcting codes \cite{gottesman_stabilizer_1997}.
It can be operationally described as the fragment of qubit QM where the only allowed operations are preparations or measurements in the computational basis and unitary transformations belonging to the Clifford group \cite{nielsen_quantum_2010}.
While stabilizer quantum computation is significantly less powerful than general quantum computation -- it can be efficiently simulated on classical computers and is provably less powerful than even general classical computation \cite{aaronson_improved_2004} -- stabilizer QM is nevertheless of central importance in areas such as error-correcting codes \cite{gottesman_stabilizer_1997} or measurement-based quantum computation \cite{raussendorf_one-way_2001}, and it is non-local.
In the following, we introduce the operational formulation of pure state qubit stabilizer QM along with some of its properties, and then show how to adapt the \ZX-calculus to this subtheory.
We also describe the binary formalism for stabilizer quantum theory \cite{calderbank_quantum_1997}, which is both at the heart of the efficient simulation and has also been used to derive interesting results about stabilizer states in their own right \cite{van_den_nest_graphical_2004}.

\subsection{The Pauli group and the Clifford group}
\label{s:Pauli_Clifford}

The Pauli operators:
 \begin{equation}
  X = \begin{pmatrix}0&1\\1&0\end{pmatrix},\quad Y = \begin{pmatrix}0&-i\\i&0\end{pmatrix},\quad\text{and}\quad Z = \begin{pmatrix}1&0\\0&-1\end{pmatrix}
 \end{equation}
have a central role in quantum mechanics because, together with the identity, they form a basis for all single-qubit unitaries under linear combinations.
Under multiplication, this set of operators gives rise to the following group.

\begin{dfn}
 The \emph{Pauli group} $P_1$ is the closure of the set $\{I,X,Y,Z\}$ under multiplication.
 It consists of the identity and Pauli matrices with multiplicative factors $\{\pm 1,\pm i\}$.
 This definition generalises to multiple qubits as follows: The Pauli group on $n$ qubits, $P_n$, consists of all tensor products of Pauli and identity matrices with phase factors $\{\pm 1, \pm i\}$, i.e.:
 \begin{equation}
  P_n = \Big\{\alpha g_1\otimes g_2\otimes\ldots\otimes g_n \,\Big|\, \alpha\in\{\pm 1,\pm i\}\text{ and } g_k\in\{I,X,Y,Z\}\text{ for } k=1,\ldots,n\Big\}.
 \end{equation}
 Elements of $P_n$ are often called \emph{Pauli products}.
\end{dfn}

A closely related groups of operators is the Clifford group, whose elements map the Pauli group back to itself under conjugation.

\begin{dfn}\label{dfn:Clifford}
 The \emph{Clifford group} on $n$ qubits, denoted $\mathcal{C}_n$, is the group of operators which normalise the Pauli group, i.e.:
 \begin{equation}
  \mathcal{C}_n = \left\{ U \,\middle|\, \forall g\in P_n: UgU^\dagger\in P_n\right\}.
 \end{equation}
\end{dfn}

While all global phase operators, i.e.\ unitaries of the form $e^{i\phi}I$ for some $\phi\in(-\pi,\pi]$, map the Pauli group back to itself, the Clifford group is usually taken to contain only those global phase operators for which $\phi$ is an integer multiple of $\pi/4$.
This is because those are the only global phase operators that can arise from products of other Clifford unitaries.
We follow that convention here.

Therefore the Clifford group for any $n>1$ is generated by the global phase operator $\omega=e^{i\pi/4}I$, as well as two single-qubit operators and one two-qubit operator \cite{nielsen_quantum_2010}, namely the phase operator:
\begin{equation}
 S = \begin{pmatrix}1&0\\0&i\end{pmatrix},
\end{equation}
the Hadamard operator:
\begin{equation}
 H = \frac{1}{\sqrt{2}}\begin{pmatrix}1&1\\1&-1\end{pmatrix},
\end{equation}
and the controlled-\NOT{} operator:
\begin{equation}
 C_X = \begin{pmatrix}1&0&0&0\\0&1&0&0\\0&0&0&1\\0&0&1&0\end{pmatrix}.
\end{equation}
Ignoring global phases, the group $\mathcal{C}_1$ of single-qubit Clifford unitaries has 24 elements.
It is generated by the phase and Hadamard operators, or, alternatively, by $R_Z$ and $R_X$, where $R_Z = S$ and $R_X = HSH$.

\begin{dfn}\label{dfn:local_Clifford_group}
 The \emph{local Clifford group} on $n$ qubits, $\mathcal{C}_1\t{n}$, consists of all $n$-fold tensor products of single-qubit Clifford operators.
\end{dfn}

The set of quantum states that can be prepared by applying a Clifford unitary to a computational basis state are the stabilizer states.
As Pauli-X is a Clifford operator, it suffices to consider state preparations starting from the all-zero state $\ket{0}\t{n}$.

\begin{dfn}
 A pure $n$-qubit quantum state is called a \emph{stabilizer state} if it can be prepared by applying an $n$-qubit Clifford unitary to the state $\ket{0}\t{n}$.
\end{dfn}

The term ``stabilizer quantum mechanics'' originates from the following property.

\begin{dfn}
 A unitary operator $U$ is said to \emph{stabilize} a quantum state $\ket{\psi}$ if:
 \begin{equation}
  U\ket{\psi}=\ket{\psi}.
 \end{equation}
\end{dfn}

The unitaries stabilizing a given quantum state can easily be seen to form a group: the identity operator stabilizes all states, multiplying two stabilizers of the same state gives another stabilizer, and if some unitary $U$ stabilizes a state $\ket{\psi}$ then so does its inverse $U^\dagger$.

\begin{thm}
 For each $n$-qubit stabilizer state $\ket{\psi}$, there exists some Abelian subgroup $S\subseteq P_n$ such that $\ket{\psi}$ is the unique state stabilized by all elements of $S$.
\end{thm}
\begin{proof}
 First, consider the state $\ket{0}\t{n}$ and the set:
 \begin{equation}
  S_{\ket{0}\t{n}} = \Big\{g_1\otimes g_2\otimes\ldots\otimes g_n \,\Big|\, g_k\in\{I,Z\}\text{ for } k=1,\ldots,n\Big\}.
 \end{equation}
 It is straightforward to check that $S_{\ket{0}\t{n}}$ is an Abelian subgroup of $P_n$ and that each $\sigma\in S_{\ket{0}\t{n}}$ satisfies $\sigma\ket{0}\t{n}=\ket{0}\t{n}$.
 What remains to be shown is that $\ket{0}\t{n}$ is the unique state stabilized by all elements of $S_{\ket{0}\t{n}}$.

 Denote by $Z_k$ the Pauli product that consists of a Pauli-Z operator on the $k$-th qubit and identities everywhere else, i.e.:
 \begin{equation}
  Z_k = \underbrace{I\otimes\ldots\otimes I}_{k-1} \otimes Z \otimes \underbrace{I\otimes\ldots\otimes I}_{n-k}.
 \end{equation}
 Then $Z_k\in S_{\ket{0}\t{n}}$ for $k=1,\ldots,n$.
 Let:
 \begin{equation}
  \ket{\psi} = \sum_{x_1,\ldots,x_n\in\{0,1\}} \psi_{x_1\ldots x_n} \ket{x_1\ldots x_n}
 \end{equation}
 be an $n$-qubit state, where $\psi_{x_1\ldots x_n}\in\CC$ for all $x_1,\ldots,x_n\in\{0,1\}$.
 Now:
 \begin{equation}
  Z_k \ket{\psi} = \sum_{x_1,\ldots,x_n\in\{0,1\}} (-1)^{x_k} \psi_{x_1\ldots x_n} \ket{x_1\ldots x_n}.
 \end{equation}
 Thus, by component-wise comparison, $Z_k$ stabilizes $\ket{\psi}$ if and only if $\psi_{x_1\ldots x_n}=0$ whenever $x_k=1$.
 Therefore, as $Z_k\in S_{\ket{0}\t{n}}$ for $k=1,\ldots,n$, the only state stabilized by all elements of $S_{\ket{0}\t{n}}$ is the all-zero state $\ket{0}\t{n}$.

 Next, consider some stabilizer state $\ket{\phi}=U\ket{0}\t{n}$, where $U\in\mathcal{C}_n$.
 Let:
 \begin{equation}
  S_{\ket{\phi}} = \Big\{ U \sigma U^\dagger \,\Big|\, \sigma\in S_{\ket{0}\t{n}} \Big\}.
 \end{equation}
 It is straightforward to check that this, too, is an Abelian subgroup of $P_n$.
 Now for any $U \sigma U^\dagger\in S_{\ket{\phi}}$:
 \begin{equation}
  \left(U \sigma U^\dagger\right)\ket{\phi} = U \sigma U^\dagger U\ket{0}\t{n} = U \sigma \ket{0}\t{n} = U \ket{0}\t{n} = \ket{\phi},
 \end{equation}
 so all elements of $S_{\ket{\phi}}$ stabilize $\ket{\phi}$.
 To prove uniqueness, suppose there exists some other state $\ket{\phi'}$ that is stabilized by all elements of $S_{\ket{\phi}}$.
 For each $\sigma\in S_{\ket{0}\t{n}}$ there exists $\tau\in S_{\ket{\phi}}$ such that $\sigma = U^\dagger \tau U$.
 Therefore, $U^\dagger\ket{\phi'}$ must be stabilized by all elements of $S_{\ket{0}\t{n}}$.
 But we showed above that $\ket{0}\t{n}$ is the unique state with that property.
 Hence, $U^\dagger\ket{\phi'}=\ket{0}\t{n}$, i.e. $\ket{\phi'}=U\ket{0}\t{n}=\ket{\phi}$.
 Thus $\ket{\phi}$ is the unique state stabilized by all elements of $S_{\ket{\phi}}$.
\end{proof}

The group of Pauli products stabilizing a given state is often called its \emph{stabilizer group}.

\emph{Stabilizer scalars} are those complex numbers that can arise as outcomes of a stabilizer computation, i.e.\ a computation consisting of the preparation of the state $\ket{0}\t{n}$, application of a Clifford unitary, and a computational basis measurement on all $n$ qubits.
It is straightforward to check that stabilizer scalars take values of the form $2^{-r/2}e^{i\phi}$, where $r$ is a non-negative integer and $\phi$ is an integer multiple of $\pi/4$.

\subsection{Graph states}
\label{s:graph_states}

An important subset of the stabilizer states are the graph states, which consist of a number of qubits entangled with each other according to the structure of a mathematical graph.

\begin{dfn}\label{dfn:graph_state_QM}
 A \emph{finite graph} is a pair $G=(V,E)$ where $V$ is a finite set of vertices and $E$ is a collection of edges, which are denoted by pairs of vertices.
 A graph is \emph{undirected} if its edges are unordered pairs of vertices.
 It is \emph{simple} if it has no self-loops and there is at most one edge connecting any two vertices.
\end{dfn}

In the following, all graphs are assumed to be finite, undirected, and simple.
For such graphs, the collection of edges is in fact a set (as opposed to, say, a multi-set) and each edge is an unordered set of size two (rather than a tuple).
For an $n$-vertex graph, we often take $V=\{1,2,\ldots,n\}$.

\begin{dfn}\label{dfn:graph_state}
 Given a graph $G=(V,E)$ with $n=\abs{V}$ vertices, the corresponding \emph{graph state} $\ket{G}$ is the $n$-qubit state prepared as follows:
 \begin{itemize}
  \item for each vertex $v\in V$, a qubit prepared in the state $\ket{+}=H\ket{0}$, and
  \item for each edge $e=\{v,w\}\in E$, a controlled-Z operator applied to the appropriate qubits.
 \end{itemize}
\end{dfn}

Controlled-Z operators commute, therefore the order in which they are applied does not matter in the above definition.

All graph states are pure stabilizer states, as $H$ and controlled-Z are Clifford unitaries.
On the other hand, it is easy to see that not all stabilizer states are graph states: for example, the state $\ket{0}$ is a stabilizer state but not a graph state.
Yet there exists an interesting relationship between arbitrary stabilizer states and graph states.
Consider the equivalence relation on stabilizer states given by the local Clifford group.

\begin{dfn}
 Two $n$-qubit stabilizer states $\ket{\psi}$ and $\ket{\phi}$ are \emph{equivalent under local Clifford operations} if there exists $U\in\mathcal{C}_1\t{n}$ such that $\ket{\psi}=U\ket{\phi}$.
\end{dfn}

\begin{thm}[\cite{van_den_nest_graphical_2004}]\label{thm:stabilizer_graph_state}
 Any pure stabilizer state is equivalent to some graph state under local Clifford operations, i.e.\ any $n$-qubit stabilizer state $\ket{\psi}$ can be written, not generally uniquely, as $U\ket{G}$, where $U\in\mathcal{C}_1\t{n}$ and $\ket{G}$ is an $n$-qubit graph state.
\end{thm}

A single stabilizer state may well be equivalent to more than one graph state under local Clifford operations.
To organise these equivalence classes, we require the following definition and theorem.

\begin{dfn}\label{dfn:local_complementation}
 Let $G=(V,E)$ be a graph and let $v\in V$ be a vertex.
 The \emph{local complementation about $v$} is the operation that inverts the subgraph generated by the neighbourhood of $v$ (but not including $v$ itself).
 Formally, a local complementation about $v\in V$ sends $G$ to the graph:
 \begin{equation}\label{eq:graph_local_complementation}
  G\star v = \left(V,E \triangle \big\{\{b,c\}\big|\{b,v\},\{c,v\}\in E\wedge b\neq c\big\}\right),
 \end{equation}
 where $\triangle$ denotes the symmetric set difference, i.e.\ $A\triangle B$ contains all elements that are contained either in $A$ or in $B$ but not in both.
\end{dfn}

\begin{ex}
 Consider the line graph on four vertices.
 Applying local complementations about vertex 3 and then vertex 2 yields the following sequence of graphs:
 \begin{equation}
  \input{tikz_files/lc_example.tikz}
 \end{equation}
\end{ex}

\begin{thm}[\cite{van_den_nest_graphical_2004}]\label{thm:graph_states_LC}
 Two graph states on the same number of qubits are equivalent under local Clifford operations if and only if there is a sequence of local complementations that transforms one graph into the other.
\end{thm}

\subsection{The binary formalism for stabilizer quantum mechanics}
\label{s:binary_stabilizer_formalism}

As laid out in Section \ref{s:Pauli_Clifford}, any $n$-qubit stabilizer state corresponds to an Abelian subgroup of $P_n$, the Pauli group on $n$ qubits.
While the size of such a subgroup is exponential in $n$, it can be represented by a set of generators for the group, whose size is linear in $n$ \cite{gottesman_stabilizer_1997}.
In fact, a generating set for the stabilizer group of a pure $n$-qubit stabilizer state contains exactly $n$ independent Pauli products \cite{nielsen_quantum_2010}.
``Independent'' here means that no element can be removed from the set without making the generated group smaller.
Each generating set uniquely determines the subgroup, but there are different choices of generators for the same group.

\begin{ex}\label{ex:Bell_state_stabilizer}
 The Bell state $\frac{1}{\sqrt{2}}\left( \ket{00} + \ket{11} \right)$ is a stabilizer state with stabilizer group:
 \begin{equation}
  \{ I\otimes I,\; Z\otimes Z,\; X\otimes X,\; -Y\otimes Y \}.
 \end{equation}
 This group can be represented, for example, by the generating set $\avg{Z\otimes Z, X\otimes X}$, or by $\avg{Z\otimes Z, -Y\otimes Y}$.
\end{ex}

Any Pauli product with phase $\pm 1$ can be uniquely expressed as a binary vector using the following encoding \cite{calderbank_quantum_1997,nielsen_quantum_2010}.

\begin{dfn}
 Consider an $n$-qubit Pauli product $g = (-1)^a g_1\otimes g_2\otimes\ldots\otimes g_n$, where $a\in\{0,1\}$ and $g_1,\ldots,g_n\in P_1$.
 The $2n+1$ bit \emph{check vector} associated with $g$ is:
 \begin{equation}
  ( z_1,\ldots, z_n | x_1,\ldots, x_n | a ),
 \end{equation}
 where, for $i=1,\ldots,n$:
 \begin{equation}
  x_i = \begin{cases} 0 &\text{if } g_i \in\{I,Z\} \\ 1 &\text{if } g_i \in\{X,Y\},\end{cases} \qquad\text{and}\qquad z_i = \begin{cases} 0 &\text{if } g_i \in\{I,X\} \\ 1 &\text{if } g_i \in\{Y,Z\}.\end{cases}
 \end{equation}
\end{dfn}

The check vectors corresponding to all Pauli products in a generating set can be combined into the columns of a $(2n+1)$ by $n$ matrix.
Yet for many applications it is easier to ignore the phase bits associated with the Pauli operators and focus only on the bits determining the matrices.
In the context of local Clifford equivalence, ignoring the phases of the Pauli products is reasonable because the phase of any generator of a stabilizer subgroup can be changed by a local Clifford operation that leaves the phases of all other generators invariant \cite{van_den_nest_graphical_2004}.

\begin{dfn}
 The \emph{check matrix} for a pure stabilizer state is the $2n$ by $n$ matrix that results from combining the check vectors associated with the generators of a stabilizer subgroup into the columns of a matrix, but ignoring the last bits (i.e.\ ignoring the phases of the Pauli products).
\end{dfn}

As there are different generating sets for the same stabilizer subgroup, there are different check matrices associated with the same stabilizer state.

\begin{ex}
 The check matrices for the Bell state derived from the generating sets given in Example \ref{ex:Bell_state_stabilizer} are:
 \begin{equation}
  \left( \begin{array}{cc} 1 & 0 \\ 1 & 0 \\ \hline 0 & 1 \\ 0 & 1 \end{array} \right) \qquad\text{and}\qquad \left( \begin{array}{cc} 1 & 1 \\ 1 & 1 \\ \hline 0 & 1 \\ 0 & 1 \end{array} \right).
 \end{equation}
\end{ex}

The commutativity condition for stabilizer subgroups translates into a self-orthogonality condition on check matrices.

\begin{lem}[\cite{van_den_nest_graphical_2004}]\label{lem:check_matrix}
 Let $J$ be the $2n$ by $2n$ matrix that has $n$ by $n$ identity matrices in its off-diagonal quadrants and zeroes elsewhere, i.e.:
 \begin{equation}\label{eq:J_definition}
  J = \begin{pmatrix} 0 & I \\ I & 0 \end{pmatrix},
 \end{equation}
 where $I$ is the $n$ by $n$ identity matrix.
 Then any $2n$ by $n$ check matrix $S$ is self-orthogonal under the symplectic inner product, i.e.\ it satisfies:
 \begin{equation}
  S^T J S = 0.
 \end{equation}
 Conversely, any self-orthogonal $2n$ by $n$ matrix is a valid check matrix, i.e.\ it corresponds to a commuting stabilizer subgroup.
\end{lem}

Graph states have particularly straightforward representations as check matrices, making use of the following matrix encoding for finite simple graphs.

\begin{dfn}
 A graph $G$ with $n$ vertices can be described by a symmetric $n$ by $n$ matrix $\theta$ with binary entries such that $\theta_{ij}=1$ if and only if there is an edge connecting vertices $i$ and $j$.
 This matrix is known as the \emph{adjacency matrix}.
\end{dfn}

The adjacency matrix can be used to construct a generating set for the stabilizer subgroup of a graph state.

\begin{prop}[\cite{van_den_nest_graphical_2004}]
 Let $G=(V,E)$ be a graph with adjacency matrix $\theta$, and let $\ket{G}$ be the graph state corresponding to $G$ according to Definition \ref{dfn:graph_state}.
 Then the stabilizer group of $\ket{G}$ is generated by the following $n$ Pauli products:
 \begin{equation}
  X_v\otimes\bigotimes_{u\in V} Z_u^{\theta_{uv}} \quad\text{for all } v\in V.
 \end{equation}
 Here, subscripts indicate to which qubit the operator is applied.
\end{prop}

These generators yield the following check matrix.

\begin{lem}[\cite{van_den_nest_graphical_2004}]
 Let $G$ be a graph with adjacency matrix $\theta$, and let $\ket{G}$ be the associated graph state. Then:
 \begin{equation}\label{eq:graph_check_matrix}
  \left( \begin{array}{cc} \theta \\ \hline I \end{array} \right)
 \end{equation}
 is a check matrix for $\ket{G}$, where $I$ is the $n$ by $n$ identity matrix.
\end{lem}

In the check matrix formalism, Clifford operations are represented by binary $2n$ by $2n$ matrices that multiply the check matrices from the left.
The definition of the Clifford group as the normaliser of the Pauli group (cf.\ Definition \ref{dfn:Clifford}) translates into the binary formalism as the following condition.

\begin{lem}[\cite{van_den_nest_graphical_2004}]
 A binary $2n$ by $2n$ matrix $Q$ corresponds to a Clifford operation if and only if it preserves the symplectic inner product, i.e.\ it satisfies:
 \begin{equation}
  Q^T J Q = J,
 \end{equation}
 where $J$ is defined in \eqref{eq:J_definition}.
\end{lem}

\begin{lem}[\cite{van_den_nest_graphical_2004}]\label{lem:binary_operation}
 A local Clifford operation is represented by a binary $2n$ by $2n$ matrix of the form:
 \begin{equation}
  Q = \begin{pmatrix} A & B \\ C & D \end{pmatrix},
 \end{equation}
 where each $n$ by $n$ submatrix $A,B,C,D$ is diagonal.
\end{lem}

It is this check matrix formalism that was used to prove Theorems \ref{thm:stabilizer_graph_state} and \ref{thm:graph_states_LC} in \cite{van_den_nest_graphical_2004}.
We use the same formalism to prove corresponding results for Spekkens' toy bit theory in Section \ref{s:binary}.

\subsection{Stabilizer quantum mechanics in the \ZX-calculus}
\label{s:stabilizer_ZX}

Operationally, the theory of stabilizer QM is defined as pure state qubit QM with the following restrictions: state preparations and measurements have to be in the computational basis, and unitary operations are required to be in the Clifford group.
This group is generated by the single-qubit operators $S$ and $H$, together with the two-qubit controlled-\NOT{} operator.
In the \ZX-calculus, computational basis states and effects are denoted by:
\begin{equation}
 \ket{0} = \left\llbracket\halfscalar\;\innerprodgr\;\state{rn}\right\rrbracket, \quad
 \ket{1} = \left\llbracket\halfscalar\;\innerprodgr\;\state{rn, label={[rphase]right:$\pi$}}\right\rrbracket, \quad
 \bra{0} = \left\llbracket\halfscalar\;\innerprodgr\;\effect{rn}\right\rrbracket, \quad \text{and} \quad
 \bra{1} = \left\llbracket\halfscalar\;\innerprodgr\;\effect{rn, label={[rphase]right:$\pi$}}\right\rrbracket.
\end{equation}
The Clifford group generators can be translated into the \ZX-calculus as follows:
\begin{equation}\label{eq:stabilizer_translations}
 S = \left\llbracket \phase{gn, label={[gphase]right:$\pi/2$}}\right\rrbracket, \quad
 H = \left\llbracket \Hadamard \right\rrbracket, \quad \text{and} \quad
 C_X = \left\llbracket \innerprodgr \; \input{tikz_files/controlled-NOT.tikz} \right\rrbracket.
\end{equation}

\begin{lem}\label{lem:stabilizer-ZX}
 Any stabilizer state or operation with post-selected measurements can be represented by a \ZX-calculus diagram in which all phase angles are integer multiples of $\pi/2$.
\end{lem}
\begin{proof}
 Given a stabilizer state or operation, find a circuit representation in terms of $S$, $H$, $C_X$, computational basis states and post-selected computational basis measurements.
 Translate the circuit into the \ZX-calculus using the above representation.
 The result is a \ZX-calculus diagram in which all phase angles are integer multiples of $\pi/2$.
\end{proof}

In fact, the converse is also true. Note first that, rather than defining the \ZX-calculus in terms of phased spiders with arbitrary numbers of legs, we can also define it in terms of four types of basic spiders with small fixed numbers of inputs and outputs and phase 0, together with phase shifts and Hadamard nodes.

\begin{lem}\label{lem:basic_elements}
 Any \ZX-calculus diagram can be written as a combination of four basic spiders:
 \begin{equation}\label{eq:basic_spiders}
  \splitnode{}, \quad\quad \effect{gn}, \quad\quad \joinnode{}, \quad\text{and}\quad \state{gn},
 \end{equation}
 together with phase shifts \phase{gn, label={[gphase]right:$\alpha$}} and \phase{rn, label={[rphase]right:$\beta$}} for $\alpha,\beta\in(-\pi,\pi]$, \Hadamard, and \halfscalar.
\end{lem}
\begin{proof}
 If a red or green spider has a non-zero phase, it can be decomposed into a phase 0 spider and a single-qubit phase operator using the spider law.
 Furthermore, again using the spider law, any green spider with phase 0 can be ``pulled apart'' into a diagram composed of the four elements given above.
 Spiders with no inputs or outputs can be rewritten into a state composed with a phase shift, composed with an effect: for any $\alpha\in(-\pi,\pi]$:
 \begin{equation}
  \scalar{gn, label={[gphase]right:$\alpha$}} = \input{tikz_files/scalar_decomposed.tikz}.
 \end{equation}
 Any red spider can be turned into a green spider using the colour change law, introducing a Hadamard node on each leg.
 Thus any red spider can be written as a combination of Hadamard nodes, phase shifts, and the basic green spiders.

 Therefore, any diagram in the \ZX-calculus can be written as a combination of the four spiders given in \eqref{eq:basic_spiders}, Hadamard nodes, \halfscalar, and phase shifts.
\end{proof}

In the above decomposition, the red phase shifts could be removed and replaced with green phase shifts and Hadamard nodes without changing the result.
Nevertheless, as we often decompose single-qubit operators into red and green phase shifts rather than into green phase shifts and Hadamards, we include red phase shifts here.

\begin{lem}\label{lem:ZX-stabilizer}
 Any \ZX-calculus diagram in which all phase angles are integer multiples of $\pi/2$ represents a (not necessarily normalised) stabilizer operation with post-selected measurements.
\end{lem}
\begin{proof}
 Firstly, note that the class of \ZX-calculus diagram in which all phase angles are integer multiples of $\pi/2$ is closed under the rewrite rules.

 Secondly, by Lemma \ref{lem:basic_elements}, any \ZX-calculus diagram can be decomposed into four basic spiders plus phase shifts and Hadamard nodes.
 For each of these diagram generators, we exhibit a decomposition of the corresponding operator into the Clifford generators, computational basis states, and computational basis effects.
 In addition to the translations for \phase{gn, label={[gphase]right:$\pi/2$}} and \Hadamard{} given in \eqref{eq:stabilizer_translations}, this gives the following decompositions:
 \begin{align}
  \left\llbracket \state{gn} \right\rrbracket = \left\llbracket \input{tikz_files/state_generators.tikz} \right\rrbracket &= \sqrt{2}\, H \ket{0} \\
  \left\llbracket \effect{gn} \right\rrbracket = \left\llbracket \input{tikz_files/effect_generators.tikz} \right\rrbracket &= \sqrt{2} \bra{0} H \\
  \left\llbracket \splitnode \right\rrbracket = \left\llbracket \input{tikz_files/split_generators.tikz} \right\rrbracket &= C_X \circ (I\otimes\ket{0}) \\
  \left\llbracket \joinnode\right\rrbracket = \left\llbracket \input{tikz_files/join_generators.tikz} \right\rrbracket &= (I\otimes\bra{0}) \circ C_X
 \end{align}
 Thus any \ZX-calculus diagram in which all phase angles are integer multiples of $\pi/2$ can be translated into a stabilizer operation with preparation of states in the computational basis and post-selected computational basis measurements.
\end{proof}

It is straightforward to normalise diagrams in the stabilizer \ZX-calculus by adding copies of \halfscalar{} and/or \innerprodgr{} as needed.
Thus, combining Lemmas \ref{lem:stabilizer-ZX} and \ref{lem:ZX-stabilizer}, we have the following:

\begin{thm}\label{thm:ZX_stabilizer}
 The \ZX-calculus for stabilizer quantum mechanics with post-selected measurements consists exactly of those diagrams in which all phase angles are integer multiples of $\pi/2$.
\end{thm}

By the spider law, the set of all phase shifts for the \ZX-calculus or for some fragment of the \ZX-calculus form a group called the \emph{phase group}.
The group operation is given by the merging of spiders, \phase{gn} is the group identity, and the Hermitian adjoint of a phase shift is its group inverse.

\begin{lem}[\cite{coecke_phase_2011}]
 The phase group for the stabilizer \ZX-calculus is isomorphic to the cyclic group $\ZZ_4$.
\end{lem}

The idea of phase groups is relevant also in the construction of a graphical calculus for Spekkens' toy theory in Chapter \ref{s:spekkens}.

\section{The Clifford+T group}
\label{s:Clifford+T_group}

The Clifford group, i.e.\ the group of unitary stabilizer operations, is significantly less powerful than the group of general unitary operations on qubits.
In particular, the number of distinct Clifford unitaries acting on a fixed number of qubits is finite.
Hence stabilizer quantum mechanics encompasses only a small fragment of all pure quantum operations on qubits.
Nevertheless it is possible to construct an \emph{approximately universal} set of operations as defined in Section \ref{s:universality} by adding an appropriate non-stabilizer gate to the Clifford group \cite{nielsen_quantum_2010}.

\begin{dfn}
 The Clifford+T group is the group of unitary operations generated by:
 \begin{equation}
  T = \begin{pmatrix}1&0\\0&e^{i\pi/4}\end{pmatrix}
 \end{equation}
 and the Clifford unitaries.
\end{dfn}

\begin{prop}[\cite{boykin_universal_1999}]
 The Clifford+T group is approximately universal.
\end{prop}

The \ZX-calculus representations of the Clifford gates are given in \eqref{eq:stabilizer_translations}.
The $T$ gate is represented by:
\begin{equation}
 T=\left\llbracket\phase{gn, label={[gphase]right:$\pi/4$}}\right\rrbracket.
\end{equation}
Similar to the stabilizer \ZX-calculus introduced in the previous section, we find that the notion of a Clifford+T \ZX-calculus is well-defined.

\begin{thm}
 The \ZX-calculus for the Clifford+T group with state preparation and post-selected measurements in the computational basis consists exactly of those diagrams in which all phase angles are integer multiples of $\pi/4$.
\end{thm}

The proof is analogous to that of Theorem \ref{thm:ZX_stabilizer}.

In Section \ref{s:Clifford+T}, we present results specific to the \emph{scalar-free single-qubit Clifford+T} group, which is the subgroup of unitary Clifford+T operations acting on one qubit only, and where equality is taken to be up to some non-zero scalar factor.
In the \ZX-calculus, the scalar-free single-qubit Clifford+T group is represented by those diagrams in which all phase angles are integer multiples of $\pi/4$ and which have the structure of a \emph{line graph}, i.e.\ any node in the diagram has exactly one input and one output, and there are no cycles.

%% file: tikz_files/wire.tikz
\begin{tikzpicture}
	\begin{pgfonlayer}{nodelayer}
		\node [style=none] (0) at (0, 1) {};
		\node [style=none] (1) at (0, -1) {};
	\end{pgfonlayer}
	\begin{pgfonlayer}{edgelayer}
		\draw (0.center) to (1.center);
	\end{pgfonlayer}
\end{tikzpicture}

%% file: tikz_files/scalar_plus.tikz
\begin{tikzpicture}
	\begin{pgfonlayer}{nodelayer}
		\node [style=bscalar] (0) at (0, -0.75) {};
		\node [style=gn] (1) at (0, -0) {};
		\node [style=none] (2) at (0, 0.5) {};
	\end{pgfonlayer}
	\begin{pgfonlayer}{edgelayer}
		\draw (2.center) to (1);
	\end{pgfonlayer}
\end{tikzpicture}

%% file: tikz_files/states_swapped.tikz
\begin{tikzpicture}
	\begin{pgfonlayer}{nodelayer}
		\node [style=gn] (0) at (-0.5, -0.5) {};
		\node [style=rn] (1) at (0.5, -0.5) {};
		\node [style=none] (2) at (0.5, 0.5) {};
		\node [style=none] (3) at (1.5, 0.5) {};
	\end{pgfonlayer}
	\begin{pgfonlayer}{edgelayer}
		\draw (2.center) to (1);
		\draw [in=-90, out=90, looseness=1.00] (0) to (3.center);
	\end{pgfonlayer}
\end{tikzpicture}

%% file: tikz_files/zero-effect_Z.tikz
\begin{tikzpicture}
	\begin{pgfonlayer}{nodelayer}
		\node [style=rn] (0) at (0, 0.75) {};
		\node [label={[gphase]right:$\pi$}, style=gn] (1) at (0, -0) {};
		\node [style=none] (2) at (0, -0.75) {};
	\end{pgfonlayer}
	\begin{pgfonlayer}{edgelayer}
		\draw (0) to (2.center);
	\end{pgfonlayer}
\end{tikzpicture}

%% file: tikz_files/Hopf_derivation_5.tikz
\begin{tikzpicture}
	\begin{pgfonlayer}{nodelayer}
		\node [style=rn] (0) at (0.5, -0.5) {};
		\node [style=gn] (1) at (0.5, 0.5) {};
		\node [style=none] (2) at (0.5, 1.25) {};
		\node [style=none] (3) at (0.5, -1.25) {};
	\end{pgfonlayer}
	\begin{pgfonlayer}{edgelayer}
		\draw (2.center) to (1);
		\draw (0) to (3.center);
	\end{pgfonlayer}
\end{tikzpicture}

%% file: tikz_files/inverse_der1.tikz
\begin{tikzpicture}
	\begin{pgfonlayer}{nodelayer}
		\node [style=gn] (0) at (-1, 0.5) {};
		\node [style=gn] (1) at (-1, 1.25) {};
		\node [style=rn] (2) at (-1, -0.5) {};
		\node [style=rn] (3) at (-1, -1.25) {};
	\end{pgfonlayer}
	\begin{pgfonlayer}{edgelayer}
		\draw (1) to (0);
		\draw (2) to (3);
	\end{pgfonlayer}
\end{tikzpicture}

%% file: tikz_files/scalar_decomposed.tikz
\begin{tikzpicture}
	\begin{pgfonlayer}{nodelayer}
		\node [style=gn] (0) at (0, 1) {};
		\node [label={[gphase]right:$\alpha$}, style=gn] (1) at (0, -0) {};
		\node [style=gn] (2) at (0, -1) {};
	\end{pgfonlayer}
	\begin{pgfonlayer}{edgelayer}
		\draw (0) to (2);
	\end{pgfonlayer}
\end{tikzpicture}

%% file: tikz_files/state_generators.tikz
\begin{tikzpicture}
	\begin{pgfonlayer}{nodelayer}
		\node [style=rn] (0) at (0, -0.5) {};
		\node [style=Hadamard] (1) at (0, 0.25) {};
		\node [style=none] (2) at (0, 0.75) {};
	\end{pgfonlayer}
	\begin{pgfonlayer}{edgelayer}
		\draw (2.center) to (0);
	\end{pgfonlayer}
\end{tikzpicture}

%% file: tikz_files/effect_generators.tikz
\begin{tikzpicture}
	\begin{pgfonlayer}{nodelayer}
		\node [style=rn] (0) at (0, 0.5) {};
		\node [style=Hadamard] (1) at (0, -0.25) {};
		\node [style=none] (2) at (0, -0.75) {};
	\end{pgfonlayer}
	\begin{pgfonlayer}{edgelayer}
		\draw (2.center) to (0);
	\end{pgfonlayer}
\end{tikzpicture}

%% file: completeness.tex
\chapter{The \ZX-calculus and completeness}
\label{ch:completeness}

In the previous chapter, we introduced the \ZX-calculus notation and rewrite rules, and argued that this graphical calculus for pure state qubit quantum mechanics with post-selected measurements is universal and sound.
We now look at the completeness properties of the \ZX-calculus, i.e.\ the question of whether any equality that can be derived using matrices can also be derived graphically.

As shown by Schr\"{o}der and Zamdzhiev \cite{schroeder_incomplete_2014}, the full \ZX-calculus is incomplete.
We recap this result in Section \ref{s:incompleteness}.
Next, we argue that the incompleteness of the universal \ZX-calculus does not preclude completeness for fragments of the \ZX-calculus where the phase angles are restricted.
The \ZX-calculus exhibits map-state duality, and therefore general results about the calculus can be derived by considering only state diagrams.
For the remainder of the chapter, we focus on the stabilizer \ZX-calculus as introduced in Section \ref{s:stabilizer_ZX}, showing first that there exists a non-unique normal form for stabilizer state diagrams, and then that it is possible to derive, using the rewrite rules given in the previous chapter, all equalities between stabilizer state diagrams that are true up to non-zero scalar factor.

\section{Incompleteness of the universal \ZX-calculus}
\label{s:incompleteness}

A graphical language is incomplete under some interpretation $\intf{-}$ if there exist two diagrams $D_1$ and $D_2$ that represent the same process under this interpretation -- i.e.\ $\intf{D_1} = \intf{D_2}$ -- but it is not possible to transform one diagram into the other using the rewrite rules.
As impossibility is hard to prove directly, incompleteness results are usually derived by constructing an alternative interpretation functor $\intf{-}'$ for which the rewrite rules are all sound, but under which $D_1$ and $D_2$ have distinct interpretations.
If $\intf{D_1}'\neq\intf{D_2}'$ and the rewrite rules are sound for $\intf{-}'$, it follows that it cannot be possible to rewrite $D_1$ into $D_2$, or conversely, as a sound interpretation means that any pair of diagrams that are equal according to the rewrite rules have to be mapped to equal processes in the underlying theory.
This approach was used, for example, to show that the Euler decomposition of the Hadamard node is independent of the previously existing rules of the \ZX-calculus \cite{duncan_graph_2009,duncan_pivoting_2013}.

Let $\intf{-}$ denote the usual interpretation functor for \ZX-calculus diagrams as given in Definition \ref{dfn:interpretation_functor}.
The rewrite rules introduced in Section \ref{s:rewrite_rules} are all sound with respect to this interpretation, i.e.\ they translate to true equalities between matrices.
A whole family of alternative interpretation functors can be defined by multiplying all phases in a \ZX-calculus diagram by some integer with respect to the usual interpretation.

\begin{dfn}
 Let $j$ be an integer and define the linear map $\llbracket - \rrbracket_j$ from \ZX-calculus diagrams to matrices as follows:
 \begin{itemize}
  \item the interpretations of wires (including wire crossings, caps, and cups), \Hadamard, and \halfscalar{} under $\llbracket - \rrbracket_j$ are the same as under $\llbracket - \rrbracket$, and
  \item phases of spiders are multiplied by $j$ as compared to $\llbracket-\rrbracket$:
   \begin{align}
    \left\llbracket \input{tikz_files/green_spider.tikz} \right\rrbracket_j &:= \left\llbracket \input{tikz_files/green_spider_j.tikz} \right\rrbracket \\
    \left\llbracket \input{tikz_files/red_spider.tikz} \right\rrbracket_j &:= \left\llbracket \input{tikz_files/red_spider_j.tikz} \right\rrbracket
   \end{align}
 \end{itemize}
\end{dfn}

This family of interpretation functors was first introduced in \cite{duncan_graph_2009}.
The functor $\intf{-}_0$ is used to show that the Euler decomposition rule is independent of the other rules (excluding the zero rule and zero scalar rule, which were only introduced later) \cite{duncan_pivoting_2013}.

\begin{lem}\label{lem:alternative_interpretations}
 For odd $j$, the rewrite rules in Section \ref{s:rewrite_rules} are sound under $\llbracket - \rrbracket_j$.
\end{lem}
\begin{proof}
 Rewrite rules involving only zero phases are clearly sound under the new interpretation map.
 Rules that only involve one non-zero phase which is the same on both sides of the equality are also sound.
 The spider rule is sound, as $j\alpha+j\beta = j(\alpha+\beta)$.
 With $j$ odd, we have $j\pi \equiv \pi \mod 2\pi$, so the $\pi$-copy, $\pi$-commutation, zero, and zero scalar rules hold under the new interpretation.
 This leaves only the Euler decomposition rule, for which there are two cases.

 If $j=4l+1$ for some integer $l$, then:
 \begin{align}
  j\pi/2 &\equiv \pi/2\mod 2\pi, \text{ and} \\
  j(-\pi/2) &\equiv -\pi/2\mod 2\pi,
 \end{align}
 thus the Euler decomposition rule remains unchanged.
 Otherwise, if $j=4l+3$ for some integer $l$, then:
 \begin{align}
  j\pi/2 &\equiv -\pi/2\mod 2\pi, \text{ and} \\
  j(-\pi/2) &\equiv \pi/2\mod 2\pi.
 \end{align}
 The Euler decomposition rule with the signs of all phases flipped is just the Hermitian adjoint of the original Euler decomposition rule.
 This is because scalar diagrams can be rotated upside-down using the topology meta rule without changing their value.
 Thus the Euler decomposition rule is sound under the map $\llbracket - \rrbracket_j$, completing the proof.
\end{proof}

As all the rewrite rules are sound under $\llbracket - \rrbracket_j$ for any odd integer $j$, any diagram equality derivable in the \ZX-calculus must hold not just under the normal interpretation map but also under $\llbracket - \rrbracket_j$ for odd $j$.
Conversely, any diagram equality that does not hold under $\llbracket - \rrbracket_j$ for some odd integer $j$ cannot be derivable using the rules of the \ZX-calculus.
This idea forms the basis of the following argument.

\begin{thm}[\cite{schroeder_incomplete_2014}]\label{thm:incompleteness}
 The exactly universal \ZX-calculus is incomplete.
\end{thm}
\begin{proof}
 It is well-known that any single-qubit unitary operator can be decomposed into three single-qubit rotations up to a global phase \cite{nielsen_quantum_2010}.
 This decomposition is called \emph{Euler angle decomposition} by analogy with the decomposition of an arbitrary rotation in three-dimensional space into rotations about two fixed axes.
 The Euler decomposition of the Hadamard \eqref{eq:Euler_dec} is one example of this decomposition for single-qubit unitaries.

 Expressed in terms of \ZX-calculus diagrams, this result takes the following form. For any unitary \genunitary{$U$}, angles $\phi,\alpha,\beta,\gamma\in(-\pi,\pi]$ can be chosen such that:
 \begin{equation}\label{eq:Euler_dec_general}
  \left\llbracket \; \innerprod{gn}{rn}\;\; \genunitary{$U$} \; \right\rrbracket = \left\llbracket \innerprod{gn,label={[gphase]right:$\phi$}}{rn,label={[rphase]right:$\pi$}} \; \input{tikz_files/Euler_dec_general.tikz}\right\rrbracket.
 \end{equation}
 There are double square brackets in this equation because, while we are representing the operators graphically, this is a result that holds for matrices.
 In fact, we show that the corresponding diagram equality does not always follow from the current set of rewrite rules of the \ZX-calculus.

 To do this, consider the following single-qubit unitary:
 \begin{equation}\label{eq:incompleteness_example}
  \genunitary{$U$}\; = \; \input{tikz_files/incompleteness_ex.tikz}
 \end{equation}
 It is straightforward to check that \eqref{eq:Euler_dec_general} holds for this \genunitary{$U$} if:
 \begin{align}
  \alpha = \gamma &= -\arccos\left(\frac{5}{2\sqrt{13}}\right) \approx 0.2561 \pi \label{eq:incompleteness_alpha} \\
  \beta &= -2\arcsin\left(\frac{\sqrt{3}}{4}\right) \approx -0.2851 \pi \label{eq:incompleteness_beta} \\
  \phi &= \arcsin\left(\frac{\sqrt{3}}{4}\right) -\alpha \approx 0.3987 \pi. \label{eq:incompleteness_phi}
 \end{align}
 Now consider:
 \begin{equation}
  \left\llbracket \; \innerprod{gn}{rn}\;\; \genunitary{$U$} \; \right\rrbracket_3 \quad\text{and}\quad \left\llbracket \input{tikz_files/Euler_dec_general.tikz}\right\rrbracket_3
 \end{equation}
 for the \genunitary{$U$} and $\alpha,\beta,\gamma,\phi$ defined above.
 The first diagram is just a scaled version of the identity:
 \begin{equation}
  \left\llbracket \; \innerprod{gn}{rn}\;\; \genunitary{$U$} \; \right\rrbracket_3 = \sqrt{2} \left(\ket{0}\bra{0} + \ket{1}\bra{1}\right).
 \end{equation}
 On the other hand, for the given values of $\alpha,\beta,\gamma$, and $\phi$, the second diagram is clearly not a scaled identity.

 But, as shown in Lemma \ref{lem:alternative_interpretations}, the rewrite rules given in Section \ref{s:rewrite_rules} are sound under the interpretation map $\llbracket - \rrbracket_3$, i.e.\ any equality derived using those rules must be true under the interpretation $\llbracket - \rrbracket_3$.
 Thus it is impossible to rewrite:
 \begin{center}
  \innerprod{gn}{rn} \input{tikz_files/incompleteness_ex.tikz} $\quad$ into $\quad$ \input{tikz_files/Euler_dec_general.tikz}
 \end{center}
 with the given values of $\alpha,\beta,\gamma$, and $\phi$.

 Therefore, the exactly universal \ZX-calculus with the rewrite rules given in Section \ref{s:rewrite_rules} is incomplete.
\end{proof}

We have shown that there are pairs of diagrams that represent the same matrix but cannot be rewritten into each other using the current rewrite rules of the \ZX-calculus.
The current rule set being incomplete does not mean the \ZX-calculus can never be complete.
For example, the Euler decomposition of the Hadamard operator was not contained in the original set of rewrite rules for the \ZX-calculus but added later after it was found to not be derivable from the other rewrite rules \cite{duncan_graph_2009}.
Thus the obvious approach to solving the incompleteness problem is to add more rewrite rules.

Yet, as Schr\"{o}der and Zamdzhiev argue, it may not be possible to complete the exactly universal \ZX-calculus by adding a finite (or even finitely generated) set of new rewrite rules \cite{schroeder_incomplete_2014}.
Thus the \ZX-calculus for universal pure state qubit quantum mechanics may always be incomplete.

\section{Completeness results are possible for fragments of the \ZX-calculus}
\label{s:possible_completeness}

There is another approach that can lead to completeness results for the \ZX-calculus: rather than trying to complete the exactly universal calculus, one can consider fragments of the \ZX-calculus, corresponding to more restricted theories.
The incompleteness proof relies centrally on the fact that arbitrary single qubit unitaries -- or, in \ZX-calculus terms, arbitrary phase angles -- are allowed.
In a fragment of quantum theory which imposes restrictions on the phase angles, e.g.\ stabilizer quantum mechanics (cf.\ Section \ref{s:stabilizer_ZX}), this does not hold.
Therefore the incompleteness proof does not preclude completeness results for fragments of the pure state qubit quantum theory.

For example, consider the following result for single-qubit unitaries within stabilizer quantum mechanics.

\begin{lem}\label{lem:single-qubit_Clifford_normal_form}
 Assuming that non-zero stabilizer scalars have inverses, any unitary single-qubit Clifford operator can be written uniquely in one of the forms:
 \begin{equation}\label{eq:local_Clifford_normal_form}
  \genscalar{$r$} \; \input{tikz_files/single_Clifford_normal_form_1.tikz} \qquad \text{or} \qquad \genscalar{$s$} \; \input{tikz_files/single_Clifford_normal_form_2.tikz},
 \end{equation}
 where $\alpha,\beta,\gamma\in\{0,\pi/2,\pi,-\pi/2\}$ and \genscalar{$r$}, \genscalar{$s$} are (not necessarily connected) scalar diagrams with modulus $1$.
\end{lem}
\begin{proof}
 Any single-qubit Clifford operator can be written entirely in terms of green and red phase shifts by replacing Hadamard nodes using the Euler decomposition rule.
 The two copies of \innerprodgr{} required to match the left-hand side of the Euler decomposition rule can be created together with a copy of \halfscalar{} using the variant star rule \eqref{eq:halfscalar_innerprodgr2}.

 Furthermore, each single-qubit unitary Clifford operator must have a representation with no more than three red or green nodes in the non-scalar part, as given any diagram with at least four nodes in the non-scalar part, at least one of the following applies:
 \begin{itemize}
  \item There are two adjacent nodes of the same colour, in which case they can be merged by the spider rule.
  \item There is a node with a phase that is a multiple of $2\pi$, in which case it can be removed by the identity rule.
  \item There is a node with a phase of $\pi$, in which case it can be moved past a node of the other colour using the $\pi$-commutation rule and then merged with another node of the same colour.
  The scalar \innerprodgr{}, which is required to match the left-hand side of the $\pi$-commutation rule, can be created using the variant star rule as above.
  As the $\pi$-commutation rule also holds upside-down and/or with colours flipped, it is never necessary to apply it backwards, thus we do not need to worry about matching the scalar on its right-hand side in our diagram.
  \item If none of the above cases apply, the nodes in the diagram must all have phases in the set $\{\pm\pi/2\}$ and their colours must alternate.
  Now assume that there is a scalar diagram \genscalar{$s$} which is inverse to \innerprod{gn,label={[gphase]right:$-\pi/2$}}{rn,label={[rphase]right:$-\pi/2$}}, i.e.\ it satisfies:
  \begin{equation}\label{eq:inverse_for_omega_dagger}
   \genscalar{$s$} \; \innerprod{gn,label={[gphase]right:$-\pi/2$}}{rn,label={[rphase]right:$-\pi/2$}} = \qquad,
  \end{equation}
  Then by the Euler decomposition and the self-inverse property of the Hadamard operator, together with the colour change rule, we have:
  \begin{equation}\label{eq:Euler_colour-swap}
   \input{tikz_files/Euler_der_1.tikz} \;
   = \; \genscalar{$s$} \; \innerprod{gn,label={[gphase]right:$-\pi/2$}}{rn,label={[rphase]right:$-\pi/2$}} \; \input{tikz_files/Euler_der_1.tikz} \;
   = \; \genscalar{$s$} \; \innerprodgr \; \innerprodgr \; \Hadamard \;
   = \; \genscalar{$s$} \; \innerprodgr \; \innerprodgr \; \input{tikz_files/Euler_der_2.tikz} \;
   = \; \input{tikz_files/Euler_der_3.tikz} \;
   = \; \input{tikz_files/Euler_der_4.tikz} .
  \end{equation}
  We show in Lemma \ref{lem:omega_inverses} that the above assumption is satisfied, i.e.\ there exists a scalar diagram \genscalar{$s$} such that \eqref{eq:inverse_for_omega_dagger} is true.

  By pre- and post-composing the first and last diagrams in \eqref{eq:Euler_colour-swap} with \phase{gn,label={[gphase]right:$\pi$}} or \phase{rn,label={[rphase]right:$\pi$}} and using the spider and $\pi$-commutation rules, similar results can be derived for any combination of plus and minus signs in the phases.
  Hence if there is a sequence of four nodes of alternating colours, all of which have phases in the set $\{\pm\pi/2\}$, we can change the colours of three of them, and thus get two adjacent nodes of the same colour, which can be merged.
 \end{itemize}
 In each of the cases listed above, the number of nodes in the non-scalar part of the diagram can be reduced by applying suitable rewrite rules.
 The strategy works until there are no more than three nodes left.
 Having reduced all diagrams to at most three nodes, it is straightforward -- albeit somewhat tedious -- to check that the given expressions indeed give a unique representation for each Clifford operator.
 The condition on the modulus of the scalar subdiagram is imposed by the unitarity of the Clifford operator as a whole.
\end{proof}

From this lemma it follows immediately that the Euler decomposition of any single-qubit Clifford operator involves only phase angles that are integer multiples of $\pi/2$.
Thus, Theorem \ref{thm:incompleteness} does not apply to the \ZX-calculus restricted to stabilizer quantum mechanics.
In fact, the above lemma implies that -- under the assumption of non-zero stabilizer scalars being invertible -- the \ZX-calculus is complete for unitary single-qubit Clifford operators.
Similar arguments can be made for other fragments of the \ZX-calculus, e.g.\ the single-qubit Clifford+T group, see Section \ref{s:Clifford+T_completeness}.

In the proof of Lemma \ref{lem:single-qubit_Clifford_normal_form}, we have ignored scalar subdiagrams for simplicity; assuming only that every scalar appearing in a rewrite rule has an inverse so that rewrite rules can be applied whenever their non-scalar part matches.
In Chapter \ref{ch:more_completeness}, we show that the assumption of non-zero scalars being invertible is justified.

We continue under that assumption for the rest of this chapter to show that the \ZX-calculus is complete for pure state stabilizer quantum mechanics with post-selected measurements and ignoring scalars.
This work was first published in \cite{backens_zx-calculus_2013}.

\section{Map-state duality in the \ZX-calculus}
\label{s:Choi-Jamiolkowski}

Ignoring scalars in the stabilizer completeness proof simplifies matters significantly.
There is another useful simplifying assumption: the fact that arguments about arbitrary processes in the calculus can be made by only considering states.
This assumption is justified by map-state duality.

Map-state duality, also known as the Choi-Jamio{\l}\-kowski isomorphism, relates quantum states and linear operators:

\begin{thm}[Map-state duality or Choi-Jamio{\l}kowski isomorphism]\label{thm:Choi-Jamiolkowski}
 For any pair of positive integers $n$ and $m$, there exists a bijection between the linear operators from $n$ to $m$ qubits and the states on $n+m$ qubits.
\end{thm}

In the \ZX-calculus, states are diagrams with no inputs.
Therefore the Choi-Jamio{\l}kowski isomorphism as a transformation consists of just ``bending around'' the inputs of the operator so that they become outputs.
This process can also be thought of as composing the operator with an appropriate entangled state.
In the reverse direction, we bend around some of the outputs to become inputs, or alternatively compose the diagram with the appropriate effect.

The isomorphism implies that for any operator $A$ from $n$ to $m$ qubits and for any $n+m$-qubit state $B$:
\begin{equation}\label{eq:Choi-Jamiolkowski}
 \input{tikz_files/choi-jamiolkowski.tikz}.
\end{equation}
This follows directly from the rule that \emph{only the topology matters}, which allows us to ``yank straight'' any inputs and outputs.

In the remainder of this chapter, we thus focus without loss of generality on state diagrams only.
Any results we derive about stabilizer state diagrams can also be applied to non-scalar stabilizer diagrams with arbitrary numbers of inputs and outputs.

\section{A normal form for stabilizer state diagrams}
\label{s:ZX_graph_states}

Even with the phase angles restricted to integer multiples of $\pi/2$, there is an infinite number of \ZX-calculus diagrams for any given number of inputs and outputs, and diagrams can get arbitrarily large.
To prove completeness, we therefore start by showing that any stabilizer diagram can be brought into a normal form, called ``GS-LC form'', which is based on graph states and local Clifford operators.
The term ``normal form'' is used very loosely here: the GS-LC form is not unique, i.e.\ two  GS-LC diagrams may be equal without being identical.
We also do not prove any confluence or termination properties for the rewrite system.
Instead we exhibit an algorithm that can be used to bring any stabilizer \ZX-calculus diagram into GS-LC form.
As GS-LC diagrams are fairly compact -- a normalised GS-LC diagram with $n$ inputs and $m$ outputs contains at most $O\left((n+m)^2\right)$ nodes, as shown in Lemma \ref{lem:GS-LC_size} -- this nevertheless simplifies the derivation of equalities between diagrams.

Similar to the proof of Lemma \ref{lem:single-qubit_Clifford_normal_form} above, this normal form proof relies on the following assumption:
\begin{quotation}
 \noindent Any stabilizer scalar diagram that appears as part of a non-zero rewrite rule has an inverse, i.e.\ for any:
 \begin{equation}
  \genscalar{$s$} \in \left\{ \innerprod{gn}{rn}, \; \innerprod{gn,label={[gphase]right:$\pi$}}{rn}, \; \innerprod{gn,label={[gphase]right:$\pi$}}{rn,label={[rphase]right:$\pi/2$}}, \; \innerprod{gn,label={[gphase]right:$\pi$}}{rn,label={[rphase]right:$\pi$}}, \; \innerprod{gn,label={[gphase]right:$\pi$}}{rn,label={[rphase]right:$-\pi/2$}}, \; \innerprod{gn,label={[gphase]right:$-\pi/2$}}{rn,label={[rphase]right:$-\pi/2$}} \right\},
 \end{equation}
 there exists a stabilizer scalar diagram \genscalar{$r$} such that:
 \begin{equation}
  \genscalar{$r$} \; \genscalar{$s$} = \qquad.
 \end{equation}
\end{quotation}
This assumption is necessary to ensure that rewrite rules can be applied whenever their non-scalar part matches.
We show in Section \ref{s:scalar_completeness} that the assumption is justified.
\begin{ex}
 In \eqref{eq:Euler_colour-swap} above, the non-scalar part of the right-hand side of the Euler decomposition rule matches the first diagram.
 Yet without an inverse to \innerprod{gn,label={[gphase]right:$-\pi/2$}}{rn,label={[rphase]right:$-\pi/2$}}, the rewrite rule could not be applied because the scalar part does not match.
\end{ex}

\subsection{Graph states and local Clifford operators}

Graph states as defined in Definition \ref{dfn:graph_state_QM} can be represented in the graphical calculus in an especially elegant way.

\begin{prop}\label{prop:graph_state}
 A graph state $\ket{G}$, where $G=(E,V)$ is an $n$-vertex graph, is represented in the graphical calculus as follows:
 \begin{itemize}
  \item for each vertex $v\in V$, a green node with one output and a normalising factor \halfscalar{} \innerprodgr{}, and
  \item for each edge $\{u,v\}\in E$, a Hadamard node connected to the green nodes representing vertices $u$ and $v$, as well as a normalising factor \innerprodgr{}.
 \end{itemize}
\end{prop}
\begin{proof}[Proof]
 As:
 \begin{equation}
  \intf{\halfscalar \; \innerprodgr \; \state{gn}} = \ket{+} \qquad\text{and}\qquad \intf{\innerprodgr \; \input{tikz_files/controlled-Z.tikz}} = C_Z,
 \end{equation}
 this is just the translation into the \ZX-calculus of Definition \ref{dfn:graph_state}.
\end{proof}

Graph states were first introduced into the scalar-free \ZX-calculus in \cite{duncan_graph_2009}.
The proof that this definition agrees with the definition of graph states in terms of their stabilizers translates straightforwardly from the argument given there to the scaled \ZX-calculus.

\begin{ex}
 Let $G$ be the graph:
 \begin{equation}
  \input{tikz_files/graph_example.tikz}
 \end{equation}
 The stabilizer group of the corresponding 4-qubit graph state is generated by the operators:
\begin{align}
 & X\otimes Z\otimes Z\otimes Z, \nonumber \\
 & Z\otimes X\otimes Z\otimes Z, \nonumber \\
 & Z\otimes Z\otimes X\otimes I, \text{ and} \nonumber \\
 & Z\otimes Z\otimes I\otimes X.
\end{align}
 By Proposition \ref{prop:graph_state}, the corresponding diagram in the \ZX-calculus is:
 \begin{equation}
  \innerprod{gn}{rn} \; \input{tikz_files/graph_state_example.tikz}
 \end{equation}
 where the vertices are rearranged so that the qubits are in the correct order.
 We verify that this is an eigenstate of the operator $X\otimes Z\otimes Z\otimes Z$.
 Indeed, ignoring the scalar \innerprodgr{} as that remains unchanged throughout:
 \begin{equation}
  \input{tikz_files/graph_state_stabilizer_example.tikz}
 \end{equation}
 using the $\pi$-copy law and the spider rule in the first step, the colour change law and the self-inverse property of Hadamard nodes in the second step, and the spider rule again in the third step.

 The same process can be applied to the other Pauli operators given above.
\end{ex}

The stabilizer graph formalism \cite{elliott_graphical_2008} (cf.\ also Section \ref{s:stabilizer_graphs}) provides a notation for arbitrary stabilizer states in terms of graph states and local Clifford operators.
Translating this idea into the \ZX-calculus yields the following definition.

\begin{dfn}\label{dfn:GS-LC}
 A diagram in the stabilizer \ZX-calculus is called a \emph{GS-LC diagram} if it consists of a graph state diagram with arbitrary single-qubit Clifford operators applied to each output, together with an arbitrary scalar.
 Following \cite{anders_fast_2005}, we call the Clifford operator associated with one of the qubits in the graph state its \emph{vertex operator}.
\end{dfn}

As a consequence of Lemma \ref{lem:single-qubit_Clifford_normal_form}, whenever a  single-qubit Clifford operator appears in the following arguments, we can assume without loss of generality that the operator is one of the diagrams given in \eqref{eq:local_Clifford_normal_form}.

Any graph state or GS-LC diagram can be decomposed into two parts: the non-scalar part, containing any nodes that are in some way connected to an output, and the normalising factor, containing any nodes that are disconnected from all outputs.
We use a white ellipse labelled with the name of the graph as short-hand notation for the non-scalar part of a graph diagram, and a white diamond as short-hand for a scalar.
For example, \gengraph{$G$} denotes the non-scalar part of the diagram representing $\ket{G}$.
The non-scalar part of an $n$-qubit GS-LC diagram with underlying graph $G$ is denoted by:
\begin{equation}
 \input{tikz_files/GS-LC_diagram.tikz},
\end{equation}
where $U_1,\ldots,U_n\in\mathcal{C}_1$.

Note that GS-LC diagrams do not need to be normalised, and they may be zero.
Nevertheless, it is straightforward to see the following.

\begin{lem}\label{lem:GS-LC_zero}
 A GS-LC diagram is zero if and only if its scalar part is zero.
\end{lem}
\begin{proof}
 The non-scalar part of a GS-LC diagram consists of a unitary operator -- controlled-Z gates to create the edges, then single-qubit Clifford operators -- applied to the state $\state{gn}\ldots\state{gn}$, which cannot be zero.
 Thus the desired result follows.
\end{proof}

Van den Nest et al. showed, using the binary formalism for stabilizer quantum mechanics, that two graph states are related to each other by local Clifford operations if and only if the underlying graphs are related by a sequence of local complementations \cite{van_den_nest_graphical_2004} (cf.\ Section \ref{s:binary_stabilizer_formalism}).
Duncan and Perdrix translated this result into the scalar-free \ZX-calculus \cite{duncan_graph_2009}, and it carries over straightforwardly to the scaled \ZX-calculus.

\begin{thm}\label{thm:Van_den_Nest}
 Let $G=(V,E)$ be a graph with adjacency matrix $\theta$ and let $G\star v$ be the graph that results from applying a local complementation about some $v\in V$ (cf. Definition \ref{dfn:local_complementation}).
 Then the equality:
 \begin{equation}
  \ket{G\star v} = R_{X,v}\otimes\bigotimes_{u\in V}R_{Z,u}^{-\theta_{uv}}\ket{G}
 \end{equation}
 holds in the \ZX-calculus, i.e.\ we have:
 \begin{equation}
  \input{tikz_files/local_complementation_1.tikz} \; = \; \genscalar{$s$} \; \input{tikz_files/local_complementation_2.tikz}
 \end{equation}
 where $\alpha_k = -\theta_{kv}\pi/2$ for $k\in V\setminus\{v\}$ and \genscalar{$s$} is a scalar diagram satisfying:
 \begin{equation}
  \left\llbracket \genscalar{$s$} \right\rrbracket = \sqrt{2^{\abs{E}-\abs{E'}}},
 \end{equation}
 where $E'$ is the set of edges of $G\star v$ and $\abs{-}$ denotes the number of elements of a set.
 An explicit form for \genscalar{$s$} can be found using Lemma \ref{lem:modulus_nf}.
\end{thm}

\begin{lem}\label{lem:modulus_nf}
 Let $r$ be an integer.
 Then the scalar $\sqrt{2^r}$ can be represented by one of the following \ZX-calculus diagrams: 
 \begin{mitem}
  \item for $r>0$, $r$ copies of \innerprodgr,
  \item for $r=0$, the empty diagram,
  \item for $r<0$ and $r$ even, $\abs{r}/2$ copies of \halfscalar, and
  \item for $r<0$ and $r$ odd, one copy of \innerprodgr{} and $(1-r)/2$ copies of \halfscalar.
 \end{mitem}
\end{lem}
\begin{proof}
 It is straightforward to check the interpretations of the given diagrams.
\end{proof}

In fact, any diagram consisting solely of copies of \innerprodgr{} and \halfscalar{} can be brought into the normal form given in Lemma \ref{lem:modulus_nf} by using the variant star rule proved in Lemma \ref{lem:variant_star_rule}.

While a general stabilizer diagram with a fixed number of inputs and outputs can contain arbitrarily many nodes, GS-LC diagrams have a more manageable size.

\begin{lem}\label{lem:GS-LC_size}
 The non-scalar part of a GS-LC diagram on $n$ qubits contains at most $(n^2+7n)/2$ nodes.
\end{lem}
\begin{proof}
 Ignoring scalars, the graph state underlying a GS-LC diagram contains one node for each vertex and one node for each edge.
 The number of vertices is $n$.
 As the graph is simple, i.e.\ there are no self-loops and at most one edge between any pair of vertices, the maximum number of edges is equal to $n(n-1)/2$.
 The non-scalar parts of the single-qubit Clifford operators can be written using at most three nodes per qubit.
 Thus the total number of nodes is:
 \begin{equation}
  n + n(n-1)/2 + 3n = \frac{1}{2}(n^2 + 7n). \qedhere
 \end{equation}
\end{proof}

We are ignoring the scalar part of the GS-LC diagram here as that can be arbitrarily large for a non-normalised operator.

\subsection{Equivalence transformations of GS-LC diagrams}
\label{s:equivalence_GS-LC}

Consider the non-normalised $n$-qubit GS-LC diagram:
\begin{equation}
 \input{tikz_files/GS-LC_diagram.tikz}
\end{equation}
where $G=(V,E)$ is a graph with adjacency matrix $\theta$, and $U_v\in\mathcal{C}_1$ for $v\in V$.
It is useful to set out explicitly three equivalence transformations of GS-LC diagrams, i.e.\ operations that take a GS-LC diagram to an equal but generally not identical GS-LC diagram.
These operations are later used to prove that any stabilizer state diagram can be brought into GS-LC form, and that the scalar-free stabilizer \ZX-calculus is complete.

\textbf{Local complementation about a qubit $v$}: Let $G\star v$ denote the graph that results from $G$ through application of the graph-theoretical local complementation about some vertex $v\in V$.
Then by Theorem \ref{thm:Van_den_Nest}:
\begin{equation}
 \input{tikz_files/GS-LC_diagram.tikz} = \; \genscalar{$s$} \; \input{tikz_files/GS-LC_diagram_lc.tikz},
\end{equation}
where $\alpha_u = \theta_{uv}\pi/2$ for $u\in V\setminus\{v\}$ and $s=\sqrt{2^{\abs{E'}-\abs{E}}}$ with $E'$ the set of edges of $G\star v$.
In the following, when we say ``local complementation'', we usually mean this transformation, which consists of a graph operation, a change to the local Clifford operators, and a change to the scalar part of the diagram.

\textbf{Fixpoint operation on a qubit $v$}: Let $v\in V$, then:
\begin{equation}
 \input{tikz_files/GS-LC_diagram_fp.tikz},
\end{equation}
where $\alpha_u = \theta_{uv}\pi$ for $u\in V\setminus\{v\}$.
This equality holds by the definition of graph states, or, alternatively, by a double local complementation about $v$.
Note that as the number of edges in the graph does not change, the normalisation does not change either.

\textbf{Local complementation along an edge $\{v,w\}$}: Let $v,w\in V$ such that $\{v,w\}\in E$. Then:
\begin{equation}
 \input{tikz_files/GS-LC_diagram.tikz} = \; \genscalar{$s$} \; \input{tikz_files/GS-LC_diagram_pivot.tikz},
\end{equation}
where:
\begin{equation}
 U_j' = \begin{cases} U_j\circ R_Z\circ R_X^{-1}\circ R_Z & \text{if }j\in\{v,w\} \\ U_j\circ Z & \text{if } j\in V\setminus\{v,w\}\wedge(\{j,v\}\in E \vee \{j,w\}\in E) \\ U_j & \text{otherwise,} \end{cases}
\end{equation}
and $G' = (V,E')$ satisfies the following properties:
\begin{itemize}
 \item $G' = ((G\star v)\star w)\star v = ((G\star w)\star v)\star w$;
 \item $\{v,w\}\in E'$;
 \item for $j\in V\setminus\{v,w\}$, $\{j,v\}\in E'\Leftrightarrow \{j,w\}\in E$ and $\{j,w\}\in E'\Leftrightarrow \{j,v\}\in E$, i.e.\ a vertex $j$ is adjacent to $v$ in $G'$ if and only if $j$ was adjacent to $w$ in $G$ and correspondingly with $v$ and $w$ exchanged;
 \item for $p,q\in V\setminus\{v,w\}$, let $P$ be the intersection of $p$'s neighbourhood with $\{v,w\}$, i.e.\ $v\in P$ if $\{p,v\}\in E$ and $w\in P$ if $\{p,w\}\in E$, and define $Q$ correspondingly.
 Then the edge $\{p,q\}$ is toggled if and only if $P,Q$ and $\emptyset$ are pairwise distinct.
\end{itemize}
The scalar factor, as usual, takes the value $s=\sqrt{2^{\abs{E'}-\abs{E}}}$.

This is an equivalence transformation because it consists of three subsequent local complementations on qubits, but it is worth a separate mention because of non-obvious properties like the symmetry under interchange of $v$ and $w$.

\subsection{Any stabilizer state diagram is equal to some GS-LC diagram}
\label{s:GS-LC}

From the binary formalism for stabilizer quantum mechanics, we know that any stabilizer state is local Clifford-equivalent to some graph state \cite{van_den_nest_graphical_2004} (cf.\ Theorem \ref{thm:stabilizer_graph_state}).
In the following, we show that a corresponding statement holds in the \ZX-calculus: any stabilizer state diagram is equal to some GS-LC diagram.
The proof of this result is strongly inspired by Anders and Briegel's proof that stabilizer quantum mechanics can be simulated efficiently on classical computers using a representation based on graph states and local Clifford operators \cite{anders_fast_2005}.

\begin{thm}\label{thm:ZX_GS-LC}
 Any stabilizer state diagram is equal to some GS-LC diagram within the \ZX-calculus.
\end{thm}
\begin{proof}
 For clarity, the proof has been split into various lemmas, which are stated and proved after the theorem.

 By Lemma \ref{lem:basic_elements}, any \ZX-calculus diagram can be written in terms of four basic spiders together with phase shifts, Hadamard nodes, and star nodes.
 Recall that a \ZX-calculus diagram represents a quantum state if and only if it has no inputs.
 Of the basic elements given in Lemma \ref{lem:basic_elements}, \state{gn} is the only one with no inputs.
 Thus any diagram representing a state must contain at least one such component (or a cup, which can be replaced by spiders).
 Clearly \state{gn} is a GS-LC diagram.
 We can now proceed by induction: If, for each of the basic components, applying it to a GS-LC diagram yields a diagram that can be rewritten to some GS-LC diagram, then any stabilizer state diagram can be rewritten to some GS-LC diagram.
 Lemmas \ref{lem:basic_state}, \ref{lem:LC}, \ref{lem:halfscalar}, \ref{lem:measurement}, \ref{lem:split} and \ref{lem:join} show these inductive steps.
 Therefore any stabilizer state diagram can be decomposed into the basic elements and then converted, step by step, into a GS-LC diagram.
\end{proof}

\begin{lem}\label{lem:basic_state}
 A stabilizer state diagram which consists of a GS-LC diagram and \state{gn} is equal to some GS-LC diagram within the \ZX-calculus.
\end{lem}
\begin{proof}
 Adding a vertex to a graph yields another graph, so adding \state{gn} to a graph state diagram yields another (not necessarily normalised) graph state diagram.
 The same holds for GS-LC diagrams.
\end{proof}

\begin{lem}\label{lem:LC}
 A stabilizer state diagram which consists of a basic single-qubit Clifford unitary, e.g.\ a phase shift or a Hadamard operation, applied to some GS-LC diagram is equal to a GS-LC diagram within the \ZX-calculus.
\end{lem}
\begin{proof}
 This follows directly from Definition \ref{dfn:GS-LC}, the definition of GS-LC diagrams.
\end{proof}

\begin{lem}\label{lem:halfscalar}
 A stabilizer state diagram which consists of a star node applied to some GS-LC diagram is equal to a GS-LC diagram within the \ZX-calculus.
\end{lem}
\begin{proof}
 GS-LC diagrams do not need to be normalised, so adding a star node to a GS-LC diagram yields another GS-LC diagram.
\end{proof}

\begin{lem}\label{lem:measurement}
 A stabilizer state diagram which consists of \effect{gn} applied to some GS-LC diagram is equal to a GS-LC diagram within the \ZX-calculus.
\end{lem}
\begin{proof}
 Call the vertex of the GS-LC diagram to which the post-selected measurement \effect{gn} is applied the \emph{operand vertex}.
 There are two cases.

 \textit{The operand vertex has no neighbours}: There are six single-qubit pure stabilizer states.
 In each, case the measurement effect combines with the state into a scalar.

 \textit{The operand vertex has at least one neighbour}: Computational basis measurements on graph states are straightforward.

 In the \ZX-calculus, the computational basis states are denoted (somewhat counter-intuitively) by red effects: \effect{rn} represents $\bra{0}$ and \effect{rn,label={[rphase]right:$\pi$}} represents $\bra{1}$, in either case up to some scalar factor.
 By repeated application of the copy rule:
 \begin{equation}
  \input{tikz_files/measurement_lemma_1a.tikz} \;
  = \; \genscalar{$s$} \; \input{tikz_files/measurement_lemma_1b.tikz} \;
  = \; \genscalar{$s$} \; \input{tikz_files/measurement_lemma_1c.tikz}
 \end{equation}
 where \genscalar{$s$} consists of $n-1$ copies of \innerprodgr{} \halfscalar.
 Using the $\pi$-copy rule, the same holds for \effect{rn,label={[rphase]right:$\pi$}}.
 Thus if the vertex operator of the operand vertex is:
 \begin{equation}
  \Hadamard{} \qquad \text{or} \qquad \input{tikz_files/Hadamard_X.tikz},
 \end{equation}
 the measured vertex is simply removed from the graph state.

 Otherwise, we can pick one neighbour of the operand vertex; following \cite{anders_fast_2005}, this neighbour is called the \emph{swapping partner}.
 A local complementation about the operand vertex adds \phase{rn,label={[rphase]right:$\pi/2$}} to its vertex operator.
 A local complementation about the swapping partner adds \phase{gn,label={[gphase]right:$-\pi/2$}} to the vertex operator on the operand vertex.
 Now these two single-qubit operators together generate all of $\mathcal{C}_1$.
 Note that local complementations about the operand vertex or its swapping partner do not remove the edge between these two vertices.
 Therefore, after each local complementation, the operand vertex still has a neighbour, enabling further local complementations.
 Hence it is always possible to change the vertex operator on the operand vertex to \Hadamard{} using local complementations.
 Then, the measurement is straightforward as above.
\end{proof}

\begin{lem}\label{lem:split}
 A stabilizer state diagram which consists of \splitnode applied to some GS-LC diagram is equal to a GS-LC diagram within the \ZX-calculus.
\end{lem}
\begin{proof}
 As before, call the vertex we are acting upon the operand vertex.
 Again, there are two cases.

 \textit{The operand vertex has no neighbours}: In this case, the part of the diagram representing the non-operand qubits does not change, hence if it was in GS-LC form originally, it remains that way.
 The overall diagram will be in GS-LC form if and only if \splitnode{} applied to the operand vertex can be transformed into a GS-LC diagram.
 Now, up to scalar factor, the six single-qubit stabilizer states can be written as:
 \begin{equation}\label{eq:single_qubit_states}
  \state{gn}, \quad\quad \state{gn, label={[gphase]right:$\pi/2$}}, \quad\quad \state{gn, label={[gphase]right:$\pi$}}, \quad\quad \state{gn, label={[gphase]right:$-\pi/2$}}, \quad\quad \state{rn}, \quad\text{and}\quad \state{rn, label={[rphase]right:$\pi$}}.
 \end{equation}
 By the spider law, the identity law, and the self-inverse property of the Hadamard operator:
 \begin{equation}
  \input{tikz_files/split_lemma_1.tikz}
 \end{equation}
 for $\alpha\in\{0,\pi/2,\pi,-\pi/2\}$.
 Using the copy law and the $\pi$-copy law, for $\beta\in\{0,\pi\}$:
 \begin{equation}
  \input{tikz_files/split_lemma_2a.tikz} \;
  = \; \innerprodgr \; \halfscalar \; \state{rn, label={[rphase]right:$\beta$}} \; \state{rn, label={[rphase]right:$\beta$}} \;
  = \; \innerprodgr \; \halfscalar \; \input{tikz_files/split_lemma_2b.tikz} \; \input{tikz_files/split_lemma_2b.tikz}.
 \end{equation}
 In each of these cases, the right hand side of the equation can easily be seen to be a GS-LC diagram.

 \textit{The operand vertex has at least one neighbour}: Note that:
 \begin{equation}
  \input{tikz_files/split_lemma_3.tikz},
 \end{equation}
 so if the vertex operator on the operand vertex is trivial, we just add a new vertex and edge to the graph.
 Now, as described in the proof for Lemma \ref{lem:measurement}, we can use local complementations about the operand vertex and a swapping partner to change the vertex operator on the operand vertex to the identity.
 Thus whenever we apply \splitnode{} to a GS-LC diagram, the result is equal to some GS-LC diagram.
\end{proof}

\begin{lem}\label{lem:join}
 A stabilizer state diagram which consists of \joinnode{} applied to some GS-LC diagram is equal to a GS-LC diagram within the \ZX-calculus.
\end{lem}
\begin{proof}
 As usual, call the qubits to which \joinnode{} is applied the operand qubits.
 This time there are two of them, and there are four cases to consider.

 \textit{Operand vertices are connected only to each other}: Since neither operand vertex is connected to any other vertex, we can neglect all non-operand vertices.
 Now, for any $U,V\in\mathcal{C}_1$:
 \begin{equation}
  \input{tikz_files/join_lemma_1.tikz}
 \end{equation}
 where $W$ is again in $\mathcal{C}_1$.
 From Lemma \ref{lem:single-qubit_Clifford_normal_form}, it is straightforward to see that any single-qubit Clifford unitary $W$ can be written as:
 \begin{equation}\label{eq:LC}
  \genunitary{$W$} \; = \; \genscalar{$s$} \; \input{tikz_files/join_lemma_2.tikz}
 \end{equation}
 for some scalar $s$ with $\abs{s}=1$ and angles $\alpha,\beta,\gamma\in\{0,\pi/2,\pi,-\pi/2\}$.
 Thus, using the spider law and the Hopf law:
 \begin{equation}
  \input{tikz_files/join_lemma_1.tikz} \;
  = \; \genscalar{$s$} \; \input{tikz_files/join_lemma_3a.tikz} \;
  = \; \genscalar{$s$} \; \input{tikz_files/join_lemma_3b.tikz} \;
  = \; \genscalar{$s$} \; \halfscalar \; \input{tikz_files/join_lemma_3c.tikz}
 \end{equation}
 Hence, the diagram can always be brought into GS-LC form.

 \textit{One operand vertex has no neighbours}: If one of the operand vertices has no neighbours, it must be in one of the six single-qubit states given in \eqref{eq:single_qubit_states}.
 Now for $\alpha\in\{0,\pi/2,\pi,-\pi/2\}$ and $\beta\in\{0,\pi\}$:
 \begin{equation}
  \input{tikz_files/join_lemma_4a.tikz} \; = \; \phase{gn, label={[gphase]right:$\alpha$}}
 \end{equation}
 and:
 \begin{equation}
  \input{tikz_files/join_lemma_4b.tikz} \; = \; \input{tikz_files/join_lemma_4c.tikz} \; = \; \innerprodgr \; \halfscalar \; \input{tikz_files/join_lemma_4d.tikz}.
 \end{equation}
 Thus by Lemmas \ref{lem:basic_state}, \ref{lem:LC}, \ref{lem:halfscalar} and \ref{lem:measurement}, no matter what the properties of the other operand vertex are, the resulting state is always equal to a GS-LC diagram.

 Having resolved the case where the two operand vertices are connected only to each other and the case where one of the operand vertices has no neighbours, we are left with the cases where both operand vertices have neighbours and at least one of the operand vertices has a non-operand neighbour.

 \textit{Both operand vertices have non-operand neighbours}: Denote the two operand vertices by $a$ and $b$. Pick one of $a$'s non-operand neighbours to be a swapping partner.
 As laid out in the proof of Lemma \ref{lem:measurement}, we can use local complementations about $a$ and its swapping partner to change the vertex operator of $a$ to the identity.
 We can then do the same for $b$, picking a new swapping partner from among its neighbours.
 If $a$ is connected to $b$ or to $b$'s swapping partner, these operations may result in adding some phase operators of the form \phase{gn,label={[gphase]right:$-\pi/2$}} to $a$'s vertex operator.
 This is not a problem, as green phase operators commute with \joinnode.
 Once the vertex operators for both operand vertices are identities or green phase operators, we can move the green phases through the spider and then merge the vertices.
 Note that:
 \begin{equation}
  \input{tikz_files/join_lemma_5a.tikz}
 \end{equation}
 and:
 \begin{equation}
  \input{tikz_files/join_lemma_5b.tikz}
 \end{equation}
 so we can remove any double edges or self-loops resulting from the merging.
 The result is a GS-LC diagram on $n-1$ qubits, where $n$ is the number of qubits in the original GS-LC diagram.

 \textit{One operand vertex is connected only to the other, but the latter has a non-operand neighbour}: We can change the vertex operator of the second operand vertex to the identity as in the previous case.
 In the process, the first operand vertex may acquire one or more non-operand neighbours; in that case we proceed as above.
 Else, by \eqref{eq:LC}, for any vertex operator $V$ on the first operand qubit:
 \begin{equation}
  \input{tikz_files/join_lemma_6a.tikz}
  = \; \input{tikz_files/join_lemma_6b.tikz}
  = \; \genscalar{$s$} \input{tikz_files/join_lemma_6c.tikz}
  = \; \genscalar{$s$} \; \input{tikz_files/join_lemma_6d.tikz} \;
  = \; \genscalar{$s$} \; \halfscalar \; \input{tikz_files/join_lemma_6e.tikz}
 \end{equation}
 where $W=V\circ H$ is decomposed as in \eqref{eq:LC} and we have used the spider law and the Hopf law.
 Again, the resulting diagram is a GS-LC diagram.

 The four cases we have considered cover all the possible configurations of the graph underlying the original GS-LC diagram, hence the proof is complete.
\end{proof}

\section{Completeness for the scalar-free stabilizer \ZX-calculus}
\label{s:scalar-free_equalities}

Having shown that any stabilizer state diagram can be brought into GS-LC form, we now look at equalities between such diagrams.
While there are significantly fewer GS-LC diagrams than general stabilizer state diagrams, it is possible to reduce the number of diagrams needing to be considered even further by moving to a ``reduced GS-LC form'' instead.
We give the equivalence operations of reduced GS-LC diagrams and then show how to use these operations to derive any sound equalities between reduced GS-LC diagrams.
This means the \ZX-calculus is complete for scalar-free stabilizer diagrams.
Finally, we give an example application of this result in Section \ref{s:example}.

Throughout this section, we continue to assume that any non-zero stabilizer scalar is invertible.
Additionally, we now drop all scalars from diagrams during derivations to save space.
The equalities derived here are thus only true up to some scalar factor.
By inspection of the applied rewrite rules, it is straightforward to put the scalars back in.

\subsection{Reduced GS-LC diagrams}
\label{s:reduced_GS-LC}

Following \cite{elliott_graphical_2008}, we define a more restricted form of GS-LC diagrams.
The resulting diagrams are still not unique, but the number of equivalent diagrams is significantly smaller, making it easier to derive equalities between them.

\begin{dfn}\label{dfn:rGS-LC}
 A \emph{stabilizer state diagram in reduced GS-LC (or rGS-LC) form} is a diagram in GS-LC form satisfying the following additional constraints.
 \begin{menum}
  \item\label{it:rGS-LC:VO} All vertex operators belong to the set:
   \begin{equation}\label{eq:reduced_vertex_operators}
    R = \left\{ \input{tikz_files/wire.tikz}\; , \; \phase{gn, label={[gphase]right:$\pi/2$}}, \; \phase{gn, label={[gphase]right:$\pi$}}, \; \phase{gn, label={[gphase]right:$-\pi/2$}}, \; \input{tikz_files/rGS-LC-Cliffords_1.tikz}, \;\; \input{tikz_files/rGS-LC-Cliffords_2.tikz} \right\}.
   \end{equation}
  \item\label{it:rGS-LC:H} Two adjacent vertices must not both have vertex operators that include red nodes.
 \end{menum}
\end{dfn}

\begin{thm}\label{thm:ZX_rGS-LC}
 Any stabilizer state diagram is equal to some rGS-LC diagram within the \ZX-calculus.
\end{thm}
\begin{proof}
 Note that any single-qubit state can be brought into rGS-LC form: for $a\in\{0,1\}$ and $\alpha\in\{0,\pi/2,\pi,-\pi/2\}$:
 \begin{equation}
  \input{tikz_files/single-qubit_rGS-LC-states.tikz}.
 \end{equation}
 By Theorem \ref{thm:ZX_GS-LC}, any stabilizer state diagram is equal to some GS-LC diagram within the \ZX-calculus.
 Lemma \ref{lem:single-qubit_Clifford_normal_form} shows that each vertex operator in the GS-LC diagram can be brought into one of the forms:
 \begin{equation}
  \input{tikz_files/single-qubit_Cliffords_normal_form.tikz}.
 \end{equation}
 Note that the cases $\beta=0$ and $\gamma=0$ of the above normal forms correspond exactly to the elements of  $R$, the restricted set of vertex operators given in \eqref{eq:reduced_vertex_operators}.
 A local complementation about a vertex $v$ pre-multiplies the vertex operator of $v$ with \phase{rn,label={[rphase]right:$-\pi/2$}}, so the vertex operator on any vertex with at least one neighbour can be brought into one of the forms \eqref{eq:reduced_vertex_operators} by some number of local complementations about this vertex.
 The other effects of local complementations are to toggle some of the edges in the graph state and to pre-multiply the vertex operators of neighbouring vertices by \phase{gn,label={[gphase]right:$\pi/2$}}.
 The set $R$ is not mapped to itself under repeated pre-multiplication with green phase operators: this operation sends the set $\{$\phase{gn,label={[gphase]right:$\alpha$}}$\}$ to itself, but it maps:
 \begin{equation}
  \input{tikz_files/rGS-LC_Cliffords_orbit.tikz}
 \end{equation}
 The normal form of a vertex operator contains at most two red nodes.
 Once a vertex operator is in one of the forms in $R$, pre-multiplication by green phase operators does not change the number of red nodes it contains when expressed in normal form.
 Thus the process of removing red nodes from the vertex operators by applying local complementations must terminate after at most $2n$ steps for an $n$-qubit diagram, at which point all vertex operators are elements of the set $R$.

 With all vertex operators in $R$, suppose there are two adjacent qubits $u$ and $v$ which both have red nodes in their vertex operators, i.e.\ there is a subdiagram of the form:
  \begin{equation}\label{eq:neighbouring_red_nodes}
   \input{tikz_files/adjacent_red_nodes.tikz}
  \end{equation}
 for $a,b\in\{\pm 1\}$.
 A local complementation along the edge $\{u,v\}$ maps the vertex operator of $u$ to:
 \begin{equation}\label{eq:pi_2+api_2+lc-edge}
  \input{tikz_files/removing_adjacent_reds.tikz},
 \end{equation}
 and similarly for $v$.
 If $a=1$, we apply a fixpoint operation to $u$ and if $b=1$, we apply one to $v$.
 After this, the vertex operators on both $u$ and $v$ are green phase operators.
 Vertex operators of qubits adjacent to $u$ or $v$ are pre-multiplied with some power of \phase{gn,label={[gphase]right:$\pi$}}.
 Thus each such operation removes the red nodes from a pair of adjacent qubits and leaves all vertex operators in the set $R$.
 Hence after at most $n/2$ such operations, it becomes impossible to find a subdiagram as in \eqref{eq:neighbouring_red_nodes}.
 Thus, the diagram is in reduced GS-LC form.
\end{proof}

\subsection{Equivalence transformations of rGS-LC diagrams}
\label{s:equivalence_rGS-LC}

It is obvious that local complementations and applications of the fixpoint rule do not in general take rGS-LC diagrams to rGS-LC diagrams.
Still, rGS-LC forms are not necessarily unique, as the following two propositions show.
The propositions are adapted from similar results in \cite{elliott_graphical_2008}.

\begin{prop}\label{prop:rGS-LC_transformation1}
 Suppose a rGS-LC diagram contains a pair of neighbouring qubits $p$ and $q$ in the following configuration:
 \begin{equation}
  \input{tikz_files/rGS-LC_transformation_1.tikz}
 \end{equation}
 where $a\in\{\pm 1\}$ and $b\in\{0,1\}$.
 Then a local complementation about $q$ followed by a local complementation about $p$ yields a diagram which can be brought into rGS-LC form by at most two applications of the fixpoint rule.
\end{prop}

\begin{proof}
 Consider first the effect of the two local complementations on the vertex operators of $p$ and $q$.
 For $p$ we have:
 \begin{equation}
  \input{tikz_files/rGS-LC_transformation_1_proof_1.tikz},
 \end{equation}
 and for $q$:
 \begin{equation}
  \input{tikz_files/rGS-LC_transformation_1_proof_2.tikz}.
 \end{equation}
 If $a=+1$, we apply a fixpoint operation to $p$ and if $b=0$, we apply a fixpoint operation to $q$; then the vertex operators of $p$ and $q$ are in $R$.
 The fixpoint operations add \phase{gn,label={[gphase]right:$\pi$}} to neighbouring qubits, which maps the set $R$ to itself.
 As fixpoint operations do not change any edges, we do not have to worry about them when considering whether the rest of the diagram satisfies Definition \ref{dfn:rGS-LC}.

 We now need to check that the two local complementations map all vertex operators to allowed ones.
 Vertices not adjacent to $p$ or $q$ can safely be ignored because their vertex operators remain unchanged.
 As the local complementations and fixpoint operations add only green phase operators to vertices other than $p$ and $q$, any vertex operator on another qubit that started out as a green phase remains a green phase.
 It remains to check the effect of the transformation on qubits whose vertex operator contains a red node and which are adjacent to $p$ or $q$.
 By Definition \ref{dfn:rGS-LC}, such qubits cannot be adjacent to $p$.
 So suppose $w$ is a qubit in the original graph state with a red node in its vertex operator and suppose the initial diagram contains an edge $\{w,q\}$.
 Then the local complementation about $q$ adds a phase factor \phase{gn,label={[gphase]right:$\pi/2$}} to the vertex operator of $w$ and it creates an edge between $w$ and $p$.
 The complementation about $p$ adds another \phase{gn,label={[gphase]right:$\pi/2$}} to $w$ and removes the edge between $w$ and $q$.
 Thus the vertex operator of $w$ remains in the set $R$, i.e. the transformation preserves property 1 of the definition of rGS-LC diagrams.

 Suppose there are two qubits $v,w$ in the original graph state, both of which have red nodes in their vertex operators and are adjacent to $q$.
 Since the original diagram is in rGS-LC form, there is no edge between $v$ and $w$.
 Now the local complementation about $q$ adds an edge between $v$ and $w$ and creates edges $\{p,v\}$ and $\{p,w\}$.
 The local complementation about $p$ removes the edge $\{v,w\}$, so once again $v$ and $w$ are not adjacent.
 Edges involving any qubits that are not adjacent to $p$ or $q$ remain unchanged.
 Thus the transformation preserves property 2 of Definition \ref{dfn:rGS-LC}.
 Hence, the resulting diagram is in rGS-LC form.
\end{proof}

\begin{prop}\label{prop:rGS-LC_transformation2}
 Suppose a rGS-LC diagram contains a pair of neighbouring qubits $p$ and $q$ in the following configuration:
 \begin{equation}
  \input{tikz_files/rGS-LC_transformation_2.tikz}
 \end{equation}
 where $a,b\in\{\pm 1\}$.
 Then a local complementation along the edge $\{p,q\}$ yields a diagram which can be brought into rGS-LC form by at most two applications of the fixpoint rule.
\end{prop}

\begin{proof}
 After the local complementation along the edge, the vertex operator of $p$ is given by \eqref{eq:pi_2+api_2+lc-edge}.
 For the vertex operator of $q$, we have:
 \begin{equation}
  \input{tikz_files/rGS-LC_transformation_2_proof.tikz}
 \end{equation}
 Thus if $a=1$, we apply a fixpoint operation to $p$ and if $b=-1$, we apply a fixpoint operation to $q$.
 From the properties of local complementations along edges (see Section \ref{s:equivalence_GS-LC}) it follows that this transformation preserves the two properties of rGS-LC states. 
\end{proof}

These two types of equivalence operation suffice to derive all equalities between rGS-LC diagrams, as shown in the next section.

\subsection{Comparing rGS-LC diagrams}
\label{s:equality_testing}

In Theorem \ref{thm:ZX_rGS-LC}, we have proved that any stabilizer state diagram is equal to some rGS-LC diagram.
Thus, the \ZX-calculus is complete for stabilizer states if, given two rGS-LC diagrams representing the same state, we can show that they are equal using the rules of the \ZX-calculus.
To do this, we again follow \cite{elliott_graphical_2008}.

\begin{dfn}
 A pair of rGS-LC diagrams on the same number of qubit is called \emph{simplified} if there are no pairs of qubits $p,q$ such that $p$ has a red node in its vertex operator in the first diagram but not in the second, $q$ has a red node in the second diagram but not in the first, and $p$ and $q$ are adjacent in at least one of the diagrams.
\end{dfn}

\begin{prop}\label{prop:simplified}
 Any pair of rGS-LC diagrams on $n$ qubits is equal to a simplified pair.
\end{prop}
\begin{proof}
 Suppose there exists a pair of qubits $p,q$ such that $p$ has a red node in its vertex operator in the first diagram but not in the second, $q$ has a red node in the second diagram but not in the first, and $p$ and $q$ are adjacent in at least one of the diagrams.
 Then in the diagram in which they are adjacent, we can apply the appropriate one of the equivalence transformations given in Section \ref{s:equivalence_rGS-LC}.
 The equivalence rules do not change the total number of red nodes among the vertex operators.
 Each such application pairs up red nodes between the two diagrams.
 Paired up qubits do not participate further in these transformations, therefore in less than $n$ steps the pair of diagrams is simplified.
\end{proof}

\begin{lem}\label{lem:unpaired_red_node}
 Consider a simplified pair of rGS-LC diagrams and suppose there exists an unpaired red node, i.e.\ there is a qubit $p$ which has a red node in its vertex operator in one of the diagrams, but not in the other.
 Then the two diagrams are not equal.
\end{lem}
\begin{proof}
 Let $D_1$ be the diagram in which $p$ has the red node, $D_2$ the other diagram.
 There are multiple cases:

 \emph{In either diagram, $p$ has no neighbours}: In this case, the overall quantum state factorises and the two diagrams are equal only if the two states of $p$ are the same.
 But:
 \begin{equation}
  \input{tikz_files/unpaired_red_node_case_disconnected.tikz}
 \end{equation}
 for $\alpha\in\{0,\pi/2,\pi,-\pi/2\}$ and $b\in\{\pm 1\}$, so the diagrams must be unequal.

 \emph{$p$ is isolated in one of the diagrams but not in the other}: Two graph states are equivalent under local Clifford operations only if one graph can be transformed into the other via a sequence of graph-theoretical local complementations.
 A local complementation never turns a vertex with neighbours into a vertex without neighbours, or conversely.
 Thus the two diagrams cannot be equal.

 \emph{$p$ has neighbours in both diagrams}: Let $N_1$ be the set of all qubits that are adjacent to $p$ in $D_1$, and define $N_2$ similarly.
 The vertex operators of any qubit in $N_1$ must be green phases in both diagrams.
 In $D_1$, this is because of the definition of rGS-LC states, in $D_2$ it is because the pair of diagrams is simplified.
 To both diagrams apply the operation:
 \begin{equation}
  U = \left(\bigotimes_{v\in N_1} C_{X,v\to p}\right)\circ R_{Z,p},
 \end{equation}
 where $R_{Z,p}$ denotes \phase{gn,label={[gphase]right:$\pi/2$}} on $p$, and $C_{X,v\to p}$ is a controlled-\NOT{} operation with control $v$ and target $p$.
 The controlled-\NOT{} gates with different controls and the same target commute, so this is well-defined.
 Furthermore, $U$ is invertible, so (in a slight abuse of notation) $U\circ D_1 = U\circ D_2 \Leftrightarrow D_1 = D_2$.
 We show that, no matter what the properties of $D_1$ and $D_2$ are (beyond the existence of an unpaired red node on $p$):
 \begin{itemize}
  \item in $U\circ D_1$, qubit $p$ is in state \state{rn} or \state{rn,label={[rphase]right:$\pi$}};
  \item in $U\circ D_2$, $p$ is either entangled with other qubits, or in one of the states \state{gn,label={[gphase]right:$\phi$}}, where $\phi\in\{0,\pi/2,\pi,-\pi/2\}$.
 \end{itemize}
 By the arguments used in the first two cases, this implies that $U\circ D_1\neq U\circ D_2$ and therefore $D_1\neq D_2$.

 Let $n=\abs{N_1}$, $m=\abs{N_1\cap N_2}$, and suppose the qubits are arranged in such a way that the first $m$ elements of $N_1$ are those which are also elements of $N_2$, if there are any.
 Consider first the effect on diagram $D_1$.
 The local Clifford operator on $p$ combines with the $R_Z$ from $U$ to:
 \begin{equation}
  R_Z\circ R_X\circ R_Z^{\pm 1} = H\circ Z^{a},
 \end{equation}
 where $a = (1\mp 1)/2$.
 Thus $U\circ D_1$ is equal to:
 \begin{equation}
  \input{tikz_files/unpaired_red_node_D1a.tikz} = \!\!\! \input{tikz_files/unpaired_red_node_D1b.tikz} = \input{tikz_files/unpaired_red_node_D1c.tikz}.
 \end{equation}
 Here, $\alpha_k\in\{0,\pi/2,\pi,-\pi/2\}$ for $k=1,\ldots,n$ and we have used the fact that green nodes can be moved past each other.
 Note that at the end, qubit $p$ is isolated and in the state \state{rn,label={[rphase]right:$a\pi$}}.

 Next consider $D_2$.
 As $N_2$ is not in general equal to $N_1$, there may be qubits adjacent to $p$ which do not have controlled-\NOT{} gates applied to them, qubits which have controlled-\NOT{} gates applied to them but are not adjacent to $p$, and qubits which are adjacent to $p$ and have controlled-\NOT{} gates applied to them.
 In the following diagram, $\beta$ and $\gamma_1,\ldots,\gamma_n$ are multiples of $\pi/2$ as usual.
 The phase $\beta$ is the combination of the original local Clifford operator on $p$ and the $R_Z$ part of $U$.
 As before, we do not care about edges that do not involve $p$.
 This time we also neglect edges between $p$ and vertices not in $N_1$:
 \begin{equation}
  \input{tikz_files/unpaired_red_node_D2.tikz}
 \end{equation}
 We distinguish different cases, depending on the value of $\beta$.

 If $\beta=\pi/2$, apply a local complementation and a fixpoint operation about $p$.
 This does not change the edges incident on $p$:
 \begin{center}
  \input{tikz_files/unpaired_red_node_D2_case_pi_2_1.tikz}
 \end{center}
 \begin{center}
  \input{tikz_files/unpaired_red_node_D2_case_pi_2_2.tikz}
 \end{center}
 \begin{equation}
  \input{tikz_files/unpaired_red_node_D2_case_pi_2_3.tikz}
 \end{equation}
 where $\gamma_k'=\gamma_k-\pi/2$ for $k=1,\ldots,m$.
 Now if $N_1=N_2$, $p$ has no more neighbours and is in the state \state{rn,label={[rphase]right:$-\pi/2$}}.
 This is not the same as the state $p$ has in diagram 1, so the diagrams are not equal.
 Else, after the application of $U$, $p$ still has some neighbours in diagram 2.
 Local complementations do not change this fact.
 Thus the two diagrams cannot be equal.
 The case $\beta=-\pi/2$ is entirely analogous, except that there is no fixpoint operation at the beginning.

 If $\beta=0$, there are two subcases.
 Firstly, suppose there exists $v\in N_2$ such that $v\notin N_1$.
 Apply a local complementation about this $v$.
 This operation changes the vertex operator on $p$ to \phase{gn,label={[gphase]right:$\pi/2$}}.
 It also changes the edges incident on $p$, but the important thing is that $p$ still has at least one neighbour.
 Thus we can proceed as in the case $\beta=\pi/2$.
 Secondly, suppose there is no $v\in N_2$ which is not in $N_1$.
 Since $N_2\neq\emptyset$ -- $N_2=\emptyset$ corresponds to the case ``$p$ has no neighbours in $D_2$'', which was considered above -- we must then be able to find $v\in N_1\cap N_2$.
 The diagram looks as follows, where now $m>0$ and, again, we are ignoring edges that do not involve $p$:
 \begin{equation}
  \input{tikz_files/unpaired_red_node_D2_case_0_1.tikz}.
 \end{equation}
 To show that the two diagrams are unequal it suffices to show that in diagram 2 the state of $p$ either factors out, but is not \state{rn} or \state{rn,label={[rphase]right:$\pi$}}, or that it remains entangled with other qubits.
 We are thus justified in ignoring large portions of the above diagram to focus only on $p$, $v$ and the controlled-Z between the two.
 In particular, we ignore for the moment the controlled-Z gates between $p$ and qubits other than $v$, as well as the last Hadamard gate on $p$.
 Then:
 \begin{align}
  \input{tikz_files/unpaired_red_node_D2_case_0_2_a.tikz} \;
  &= \; \input{tikz_files/unpaired_red_node_D2_case_0_2_b.tikz} \;
  = \; \input{tikz_files/unpaired_red_node_D2_case_0_2_c.tikz} \nonumber \\
  &= \; \input{tikz_files/unpaired_red_node_D2_case_0_2_d.tikz} \;
  = \; \input{tikz_files/unpaired_red_node_D2_case_0_2_e.tikz},
 \end{align}
 where for the second equality we have applied a local complementation and a fixpoint operation to $v$ and used the Euler decomposition, the third equality follows by a local complementation on $p$, and the last one comes from the merging of $p$ with the green node in the bottom left.
 Note that, in the end, $p$ and $v$ are still connected by an edge.
 None of the operations we ignored in picking out this part of the diagram can change that.
 Thus, as before, the state of $p$ cannot be the same as in diagram 1.
 The two diagrams are unequal.

 The case $\beta=\pi$ is analogous to $\beta=0$, except we start with a fixpoint operation on the same qubit as the first local complementation.

 We have shown that a simplified pair of rGS-LC diagrams are not equal if there are any unpaired red nodes.
\end{proof}

\begin{thm}\label{thm:rGS-LC_equality}
 The two diagrams making up a simplified pair of rGS-LC diagram are equal, i.e.\ they correspond to the same quantum state, if and only if they are identical.
\end{thm}
\begin{proof}
 By Lemma \ref{lem:unpaired_red_node}, the diagrams are unequal if there are any unpaired red nodes.
 We can therefore assume that all red nodes are paired up.

 Let the diagrams be $D_1$ and $D_2$.
 Suppose the graph underlying $D_1$ is $G_1=(V,E_1)$ and that underlying $D_2$ is $G_2=(V,E_2)$.
 For simplicity, suppose $V=\{1,2,\ldots,n\}$.
 We can draw the two diagrams as:
 \begin{equation}
  \input{tikz_files/rGS-LC_equality_1.tikz},
 \end{equation}
 where, for all $v\in V$, $\alpha_v\in\{0,\pi/2\}$ and:
 \begin{equation}
  \beta_v,\gamma_v\in \begin{cases} \{\pm\pi/2\} & \text{if }\alpha_v = \pi/2 \\ \{0,\pi/2,\pi,-\pi/2\} & \text{otherwise.} \end{cases}
 \end{equation}
 Let $V'=\{v\in V|\alpha_v=\pi/2\}$ and define the operators:
 \begin{equation}
  U = \bigotimes_{v\in V'} R_{X,v}^{-1} \quad\text{and}\quad W = \bigotimes_{\{u,w\}\in E_1} C_{Z,uw},
 \end{equation}
 where $C_{Z,uw}$ denotes a controlled-Z operator applied to qubits $u$ and $w$.
 Controlled-Z operators commute, therefore both $U$ and $W$ are invertible and we have $(W\circ U)\circ D_1=(W\circ U)\circ D_2$ if and only if $D_1=D_2$.
 Now in $U\circ D_1$ and $U\circ D_2$, all vertex operators are green nodes, which can be moved past controlled-Z operators.
 Thus $(W\circ U)\circ D_1$ is equal to \state{gn,label={[gphase]right:$\beta_1$}}~$\ldots$~\state{gn,label={[gphase]right:$\beta_n$}}.
 Now $(W\circ U)\circ D_2$ can be rewritten as follows:
 \begin{equation}
  \input{tikz_files/rGS-LC_equality_2.tikz}
 \end{equation}
 Here, the white ellipse labelled $G_1$ denotes the graph state $G_1$ with an additional input for each vertex. $E_1\triangle E_2$ is the symmetric difference of the two sets $E_1$ and $E_2$, i.e.\ the graph $(V,E_1\triangle E_2)$ contains all edges which are contained in either $G_1$ or $G_2$, but not in both.
 Clearly this is equal to a product of single-qubit states only if $E_1\triangle E_2=\emptyset$.
 That condition is satisfied if and only if $E_1=E_2$, i.e.\ $G_1=G_2$.

 Assuming that the underlying graphs are equal, we have $(W\circ U)\circ D_1=(W\circ U)\circ D_2$ if and only if $\beta_v=\gamma_v$ for all $v\in V$.
 Thus $(W\circ U)\circ D_1=(W\circ U)\circ D_2$ if and only if $D_1$ and $D_2$ are identical.
 By unitarity of $(W\circ U)$, this implies that the diagrams $D_1$ and $D_2$ are equal if and only if they are identical, as required.
\end{proof}

We have shown that if two stabilizer state diagrams in the \ZX-calculus are brought into rGS-LC form and then simplified, the result are two identical diagrams whenever the original diagrams represented the same state.
Thus, by inverting some of the rewrite steps, one of the original diagrams can be transformed into the other.
Therefore any equality between two stabilizer state diagrams that holds up to a non-zero scalar factor can be derived from the rewrite rules of the \ZX-calculus.

\subsection{Example: Two circuit decompositions for controlled-Z}
\label{s:example}

In quantum circuit notation, there are two common ways of writing the controlled-Z operator in terms of controlled-\NOT{} gates and different types of single-qubit gates, see Figure \ref{fig:circuit_CZ}.

\begin{figure}
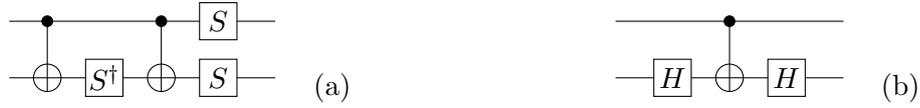

 \centering
 \input{tikz_files/circuit_CZ_2.tikz} $\qquad\qquad\qquad\qquad$ \input{tikz_files/circuit_CZ_1.tikz}
 \caption{Two quantum circuit decompositions for the controlled-Z operator in terms of controlled-\NOT{} and single-qubit gates: (a) using phase gates and (b) using Hadamard gates.}
 \label{fig:circuit_CZ}
\end{figure}

The two quantum circuit diagrams translate straightforwardly to the following \ZX-calculus diagrams:
\begin{equation}\label{eq:example_diagrams}
 \innerprodgr \; \innerprodgr \; \input{tikz_files/CZ_a.tikz} \quad\text{and}\quad \innerprodgr \; \input{tikz_files/CZ_b.tikz}.
\end{equation}
Since these two diagrams have been constructed to represent the same operator, we expect them to be equal.
To confirm this, we follow the steps given in the preceding sections.
We ignore scalars in the rewrite process because the circuits are unitary and thus their normalisation is fixed.
Furthermore, complex phases are physically irrelevant.

To bring the diagrams into GS-LC form, they first need to be mapped to the corresponding state diagrams via the Choi-Jamio{\l}kowski isomorphism.
It is useful to note that:
\begin{equation}
 \input{tikz_files/cup_graph_state.tikz},
\end{equation}
and to convert the elements of the diagrams into those given in Lemma \ref{lem:basic_elements} before transforming the diagram to a state.
Thus, the first diagram becomes:
\begin{equation}
 \input{tikz_files/CZ_a_der_1.tikz},
\end{equation}
where the grey box in the last section of the equality encloses the GS-LC part of the diagram.
The operators still outside the grey box are applied to the GS-LC state, one by one, using local complementations on the graph to change the vertex operators as necessary, until the whole diagram is in GS-LC form:
\begin{center}
 \input{tikz_files/CZ_a_der_2a.tikz}
\end{center}
\begin{equation}
 \input{tikz_files/CZ_a_der_2b.tikz}.
\end{equation}
Lastly, the vertex operators are decomposed into red and green phase operators only, so the diagram can be brought into rGS-LC form:
\begin{equation}\label{eq:diagram1}
 \input{tikz_files/CZ_a_der_3.tikz}.
\end{equation}

Similarly, the second diagram becomes:
\begin{equation}
 \input{tikz_files/CZ_b_der_1.tikz},
\end{equation}
which turns into:
\begin{equation}\label{eq:diagram2}
 \input{tikz_files/CZ_b_der_2.tikz}.
\end{equation}

The last parts of \eqref{eq:diagram1} and \eqref{eq:diagram2} form a pair of rGS-LC diagrams, which we now simplify.
Numbering the qubits from left to right, we find that both diagrams have red nodes in the vertex operator of qubit 2, and that there are further red nodes in the vertex operator of qubit 4 in the first diagram and qubit 1 in the second diagram.
Qubits 1 and 4 are connected in the first diagram, so we can apply the rGS-LC transformation given in Proposition \ref{prop:rGS-LC_transformation2} to transfer the red node from one to the other.
First, apply a local complementation about the edge $\{1,4\}$:
\begin{equation}
 \input{tikz_files/CZ_pair-up_1.tikz}.
\end{equation}
Then rewrite the vertex operators into standard form and apply a fixpoint operation about qubit 4, to get a diagram that is once again in rGS-LC form:
\begin{equation}\label{eq:diagram1_simplified}
 \input{tikz_files/CZ_pair-up_2.tikz}.
\end{equation}
The pair of diagrams given by the last parts of \eqref{eq:diagram2} and \eqref{eq:diagram1_simplified} is now simplified.
In fact, these two diagrams are identical -- as expected, considering we started with two different circuit representations of the same quantum-mechanical operator.

By taking the sequence of diagrams derived here and bending outputs in those diagrams back into inputs, we can now derive a sequence of rewrites which show directly that the two diagrams given in \eqref{eq:example_diagrams} are equal up to scalar factor.

%% file: tikz_files/Euler_der_2.tikz
\begin{tikzpicture}
	\begin{pgfonlayer}{nodelayer}
		\node [style=Hadamard] (0) at (-1.5, 1) {};
		\node [style=Hadamard] (1) at (-1.5, -0) {};
		\node [style=Hadamard] (2) at (-1.5, -1) {};
		\node [style=none] (3) at (-1.5, 1.5) {};
		\node [style=none] (4) at (-1.5, -1.5) {};
	\end{pgfonlayer}
	\begin{pgfonlayer}{edgelayer}
		\draw (3.center) to (4.center);
	\end{pgfonlayer}
\end{tikzpicture}

%% file: tikz_files/Hadamard_X.tikz
\begin{tikzpicture}
	\begin{pgfonlayer}{nodelayer}
		\node [style=Hadamard] (0) at (0.0, 0.4) {};
		\node [style=none] (1) at (0.0, 1.0) {};
		\node [style=none] (2) at (0.0, -1.0) {};
		\node [style=rn,label={[rphase]right:$\pi$}] (3) at (0.0, -0.5) {};
	\end{pgfonlayer}
	\begin{pgfonlayer}{edgelayer}
		\draw (0.center) to (1);
		\draw (2.center) to (1);
	\end{pgfonlayer}
\end{tikzpicture}

%% file: tikz_files/split_lemma_2b.tikz
\begin{tikzpicture}
	\begin{pgfonlayer}{nodelayer}
		\node [style=gn] (0) at (0, -1.5) {};
		\node [style=Hadamard] (1) at (0, -0.5) {};
		\node [label={[rphase]right:$\beta$}, style=rn] (2) at (0, 0.5) {};
		\node [style=none] (3) at (0, 1.25) {};
	\end{pgfonlayer}
	\begin{pgfonlayer}{edgelayer}
		\draw (3.center) to (0);
	\end{pgfonlayer}
\end{tikzpicture}

%% file: tikz_files/join_lemma_3c.tikz
\begin{tikzpicture}
	\begin{pgfonlayer}{nodelayer}
		\node [style=none] (0) at (0, 1.25) {};
		\node [style=gn] (1) at (0, -0.5) {};
		\node [label={[rphase]right:$\beta$}, style=rn] (2) at (0, -1.5) {};
		\node [label={[gphase]right:$\alpha+\gamma$}, style=gn] (3) at (0, 0.5) {};
	\end{pgfonlayer}
	\begin{pgfonlayer}{edgelayer}
		\draw (0.center) to (1);
	\end{pgfonlayer}
\end{tikzpicture}

%% file: tikz_files/rGS-LC-Cliffords_1.tikz
\begin{tikzpicture}
	\begin{pgfonlayer}{nodelayer}
		\node [style=none] (0) at (1, -1.2) {};
		\node [style=none] (1) at (1, 1.2) {};
		\node [label={[rphase]right:$\pi/2$}, style=rn] (2) at (1, 0.6) {};
		\node [label={[gphase]right:$\pi/2$}, style=gn] (3) at (1, -0.6) {};
	\end{pgfonlayer}
	\begin{pgfonlayer}{edgelayer}
		\draw (1.center) to (0.center);
	\end{pgfonlayer}
\end{tikzpicture}

%% file: more_completeness.tex
\chapter{Expanding \ZX-calculus completeness}
\label{ch:more_completeness}

In the previous chapter, we showed that it is possible to derive completeness results for fragments of the \ZX-calculus where the phase angles are restricted, despite the universal \ZX-calculus being incomplete.
We proved in particular that the \ZX-calculus is complete for pure state stabilizer quantum mechanics with post-selected measurements if equality is taken to be up to some non-zero scalar factor and stabilizer scalars appearing in the rewrite rules are assumed to be invertible.

In this chapter, we expand those completeness results.
First, we prove that the \ZX-calculus is complete for non-zero stabilizer scalars, and show that the assumption of invertibility for non-zero scalars was justified.
Building up on the scalar-free completeness result and the completeness result for scalars, we then show that the \ZX-calculus is complete for general pure state stabilizer quantum mechanics with post-selected measurements, including zero operators.
Finally, we show that the \ZX-calculus is also complete for the single-qubit Clifford+T group.

\section{Completeness for non-zero stabilizer scalars}
\label{s:scalar_completeness}

In the latter sections of Chapter \ref{ch:completeness}, we considered only stabilizer state diagrams.
By map-state duality, any results for state diagrams automatically extend to arbitrary diagrams with at least one input or output.
We now change the focus to scalar diagrams, i.e.\ diagrams or subdiagrams with neither inputs nor outputs.

Using the result that any stabilizer state diagram can be brought into GS-LC form (cf.\ Section \ref{s:GS-LC}), we show that stabilizer scalar diagrams can be decomposed into small disconnected segments.
We then construct a normal form for non-zero stabilizer scalar diagrams and show that it is unique.
This implies that the \ZX-calculus is complete for non-zero stabilizer scalars.
Furthermore, we show that the non-zero stabilizer scalars in the \ZX-calculus form a group, thus justifying the scalar invertibility assumption used for most of Chapter \ref{ch:completeness}.

\subsection{Decomposing scalar diagrams}
\label{s:decompose_scalars}

The process used to bring stabilizer state diagrams into GS-LC form in Theorem \ref{thm:ZX_GS-LC} can also be used to simplify stabilizer scalar diagrams.

\begin{cor}\label{cor:decompose_scalars}
 Any stabilizer scalar diagram can be decomposed into disconnected segments containing at most two nodes and one edge each.
\end{cor}
\begin{proof}
 Take any connected component of the scalar diagram which contains more than two nodes.
 Rewrite it into the inner product between some (possibly complicated) state diagram and \effect{gn}.
 This can be done for any connected scalar diagram by decomposing it into basic spiders as in Lemma \ref{lem:basic_elements}.
 Since \effect{gn} is the only basic spider with no outputs, the scalar must contain at least one copy of \effect{gn} (or a cap, which can be rewritten into spiders).

 The remainder of the scalar diagram, which represents a single-qubit state, can then be brought into GS-LC form.
 Any scalar subdiagrams that ``split off'' the main part of the diagram in this rewriting process  consist of disconnected segments containing at most two nodes each, as can easily be checked by looking at the basic rewrite rules and the proofs of Lemmas \ref{lem:measurement} through \ref{lem:join}.
 The non-scalar part of a single-qubit GS-LC diagram can be rewritten to consist of only a single node as in \eqref{eq:single_qubit_states}.
 This node combines with the green effect into a two-node diagram.
\end{proof}

Decomposing scalars into two-node segments is not just a step towards a normal form, it also allows zero diagrams -- i.e. diagrams that are interpreted as a zero matrix under the usual interpretation functor -- to be recognised without needing to construct the interpretation matrix explicitly.

\begin{cor}\label{cor:recognize_zero}
 A stabilizer diagram in GS-LC form and in which all scalar subdiagrams are decomposed as in Corollary \ref{cor:decompose_scalars} is zero if and only if it explicitly contains at least one of the following as a subdiagram:
 \begin{equation}\label{eq:zero_scalars}
  \scalar{gn, label={[gphase]right:$\pi$}}, \quad\quad \scalar{rn, label={[rphase]right:$\pi$}}, \quad\quad \innerprod{gn, label={[gphase]right:$\pi/2$}}{rn, label={[rphase]right:$-\pi/2$}}, \quad \text{ or } \quad \innerprod{gn, label={[gphase]right:$-\pi/2$}}{rn, label={[rphase]right:$\pi/2$}}.
 \end{equation}
\end{cor}
\begin{proof}
 By Lemma \ref{lem:GS-LC_zero}, a diagram in GS-LC form is zero if and only if its scalar part is zero.
 Now, for any two scalar diagrams \genscalar{$s$} and \genscalar{$r$}:
 \begin{equation}
  \left\llbracket \genscalar{$s$} \; \genscalar{$r$} \right\rrbracket = \left\llbracket \genscalar{$s$} \right\rrbracket \left\llbracket \genscalar{$r$} \right\rrbracket,
 \end{equation}
 i.e.\ putting two scalar diagrams next to each other corresponds to taking the product of their values.
 Thus a scalar diagram is zero if and only if at least one of the disconnected components is zero.
 It is straightforward to check that out of all connected scalar diagrams containing at most two nodes, the diagrams given in \eqref{eq:zero_scalars} are exactly the ones that are zero.
 The result follows.
\end{proof}

Being able to recognise zero diagrams allows us to first ignore zero diagrams and derive a normal form for non-zero stabilizer scalars.
A normal form for zero diagrams is then derived in Section \ref{s:zero_nf}

\subsection{A unique normal form for non-zero stabilizer scalars}
\label{s:scalar_nf}

In this section, whenever we talk about scalars, we mean non-zero stabilizer scalars.
Using Corollary \ref{cor:decompose_scalars}, it is straightforward to show the following.

\begin{prop}\label{prop:scalar_interpretation}
 Let $D$ be a non-zero stabilizer scalar diagram.
 Then:
 \begin{equation}\label{eq:scalar_values}
  \left\llbracket D \right\rrbracket \in \left\{ \sqrt{2^r}e^{is\pi/4} \right|\left.\vphantom{e^{is\pi/4}\sqrt{2^r}} r,s\in\ZZ\right\}.
 \end{equation}
\end{prop}
\begin{proof}
 By Corollary \ref{cor:decompose_scalars}, any scalar diagram can be decomposed into disconnected segments containing at most two nodes each.
 It is straightforward to check that each non-zero stabilizer scalar diagram $D$ consisting of at most two nodes satisfies \eqref{eq:scalar_values}.
 The values of disconnected scalar diagrams multiply, therefore the value of any scalar diagram $D$ satisfies \eqref{eq:scalar_values}.
\end{proof}

We define a normal form for scalar diagrams as follows: Pick one representative diagram for each $s\in\{0,1,\ldots,7\}$ from \eqref{eq:scalar_values}.
Combine these with the smallest number of copies of \halfscalar{} and/or \innerprodgr{} required to achieve the correct normalisation.

The simplest representative for $s=0$ is the empty diagram.

\begin{lem}
 The set:
 \begin{equation}\label{eq:complex_phases}
  \left\{ \innerprod{gn,label={[gphase]right:$\pi/2$}}{rn,label={[rphase]right:$\pi/2$}}, \quad \innerprod{gn,label={[gphase]right:$\pi/2$}}{rn,label={[rphase]right:$\pi$}}, \quad \innerprod{gn,label={[gphase]right:$\pi/2$}}{rn,label={[rphase]right:$\pi$}} \innerprod{gn,label={[gphase]right:$\pi/2$}}{rn,label={[rphase]right:$\pi/2$}}, \quad \innerprod{gn,label={[gphase]right:$\pi$}}{rn,label={[rphase]right:$\pi$}}, \quad \innerprod{gn,label={[gphase]right:$-\pi/2$}}{rn,label={[rphase]right:$\pi$}} \innerprod{gn,label={[gphase]right:$-\pi/2$}}{rn,label={[rphase]right:$-\pi/2$}}, \quad \innerprod{gn,label={[gphase]right:$-\pi/2$}}{rn,label={[rphase]right:$\pi$}}, \quad \innerprod{gn,label={[gphase]right:$-\pi/2$}}{rn,label={[rphase]right:$-\pi/2$}} \right\}
 \end{equation}
 contains one diagram for each complex phase $e^{is\pi/4}$ with $s\in\{1,\ldots,7\}$.
\end{lem}
\begin{proof}
 Apply the interpretation map to each diagram in turn, ignoring normalisation.
\end{proof}

A normal form for scalars consisting only of \innerprodgr{} and \halfscalar{} is derived in Lemma \ref{lem:modulus_nf}.

We now state and prove the normal form theorem for non-zero scalars.
For ease of understanding, some parts of the proof are given as separate lemmas, which are stated after the main theorem.

\begin{thm}\label{thm:scalar_nf}
 Any non-zero stabilizer scalar diagram in the \ZX-calculus can be uniquely represented as a combination of one of the diagrams in \eqref{eq:complex_phases}, or the empty diagram, with one of the diagrams listed in Lemma \ref{lem:modulus_nf}.
\end{thm}
\begin{proof}
 By Corollary \ref{cor:decompose_scalars}, we only need to consider diagrams made up of disconnected segments of at most two nodes each.
 Using Lemma \ref{lem:y-states}, any disconnected single node \scalar{gn, label={[gphase]right:$\alpha$}} or \scalar{rn, label={[rphase]right:$\beta$}} can be rewritten into a diagram that is made up of copies of \halfscalar{} and components consisting of exactly one red and one green node.
 Given this, Lemmas \ref{lem:innerprod_wlog} and \ref{lem:overlap_with_ket_zero}, and the fact that we are considering non-zero diagrams only, we can restrict our attention without loss of generality to diagrams built only from \halfscalar{}, \innerprodgr{}, and the following two-node diagrams:
 \begin{equation}\label{eq:basic_overlaps}
  \innerprod{gn,label={[gphase]right:$\pi/2$}}{rn,label={[rphase]right:$\pi/2$}}, \quad
  \innerprod{gn,label={[gphase]right:$\pi/2$}}{rn,label={[rphase]right:$\pi$}}, \quad
  \innerprod{gn,label={[gphase]right:$\pi$}}{rn,label={[rphase]right:$\pi$}}, \quad
  \innerprod{gn,label={[gphase]right:$-\pi/2$}}{rn,label={[rphase]right:$\pi$}}, \; \text{ and } \;
  \innerprod{gn,label={[gphase]right:$-\pi/2$}}{rn,label={[rphase]right:$-\pi/2$}}.
 \end{equation}
 These diagrams can easily be seen to also be in the set \eqref{eq:complex_phases}.
 Thus all that remains to show is that a diagram consisting of several of these elements can be rewritten to a diagram that consists of only one of the diagrams given in \eqref{eq:complex_phases} (or the empty diagram) plus any number of copies of \innerprodgr{} and \halfscalar.
 This rewriting can be done in a step-by-step fashion, so it suffices to look at pairs of diagrams from \eqref{eq:basic_overlaps}.

 The combination of any two diagrams that both contain \state{rn, label={[rphase]right:$\pi$}} can be simplified using Lemma \ref{lem:pi_multiplication}. The products:
 \begin{equation}
  \innerprod{gn,label={[gphase]right:$\pi/2$}}{rn,label={[rphase]right:$\pi/2$}} \innerprod{gn,label={[gphase]right:$-\pi/2$}}{rn,label={[rphase]right:$-\pi/2$}}, \qquad \innerprod{gn,label={[gphase]right:$\pi/2$}}{rn,label={[rphase]right:$\pi$}} \innerprod{gn,label={[gphase]right:$\pi/2$}}{rn,label={[rphase]right:$\pi/2$}}, \quad \text{and} \quad \innerprod{gn,label={[gphase]right:$-\pi/2$}}{rn,label={[rphase]right:$\pi$}} \innerprod{gn,label={[gphase]right:$-\pi/2$}}{rn,label={[rphase]right:$-\pi/2$}}
 \end{equation}
 are straightforward as well: the first diagram is shown to be simplifiable in Lemma \ref{lem:omega_inverses}, the latter two are elements of \eqref{eq:complex_phases}, so do not need to be simplified.

 For other combinations, note first that using the spider rule, the $\pi$-copy rule, and the $\pi$-commutation rule, we have:
 \begin{equation}\label{eq:scalar_pi_2_equality}
  \input{tikz_files/multiplication4_der2.0.tikz}.
 \end{equation}
 Thus:
 \begin{equation}\label{eq:omega-dagger_squared}
  \innerprod{gn,label={[gphase]right:$-\pi/2$}}{rn,label={[rphase]right:$-\pi/2$}} \; \innerprod{gn,label={[gphase]right:$-\pi/2$}}{rn,label={[rphase]right:$-\pi/2$}} \;
  = \; \halfscalar \; \innerprod{gn}{rn} \; \innerprod{gn,label={[gphase]right:$-\pi/2$}}{rn,label={[rphase]right:$\pi$}} \; \innerprod{gn,label={[gphase]right:$\pi/2$}}{rn,label={[rphase]right:$\pi/2$}} \; \innerprod{gn,label={[gphase]right:$-\pi/2$}}{rn,label={[rphase]right:$-\pi/2$}} \;
  = \; \innerprod{gn}{rn} \; \innerprod{gn}{rn} \; \innerprod{gn}{rn} \; \innerprod{gn,label={[gphase]right:$-\pi/2$}}{rn,label={[rphase]right:$\pi$}},
 \end{equation}
 where the last equality is by Lemma \ref{lem:omega_inverses}.
 Furthermore, using \eqref{eq:scalar_pi_2_equality} and Lemma \ref{lem:pi_multiplication}, we get:
 \begin{equation}\label{eq:minus_omega}
  \innerprod{gn,label={[gphase]right:$\pi$}}{rn,label={[rphase]right:$\pi$}} \; \innerprod{gn,label={[gphase]right:$\pi/2$}}{rn,label={[rphase]right:$\pi/2$}} \;
  = \; \halfscalar \; \innerprod{gn}{rn} \; \innerprod{gn,label={[gphase]right:$-\pi/2$}}{rn,label={[rphase]right:$\pi$}} \; \innerprod{gn,label={[gphase]right:$-\pi/2$}}{rn,label={[rphase]right:$\pi$}} \; \innerprod{gn,label={[gphase]right:$\pi/2$}}{rn,label={[rphase]right:$\pi/2$}} \;
  = \; \innerprod{gn,label={[gphase]right:$-\pi/2$}}{rn,label={[rphase]right:$\pi$}} \; \innerprod{gn,label={[gphase]right:$-\pi/2$}}{rn,label={[rphase]right:$-\pi/2$}}.
 \end{equation}
 Equalities similar to \eqref{eq:scalar_pi_2_equality}, \eqref{eq:omega-dagger_squared}, and \eqref{eq:minus_omega} can be derived with all the signs flipped.

 This covers all the combinations of two diagrams from \eqref{eq:basic_overlaps}.
 More complicated diagrams can be dealt with step-by-step.
 Once the subdiagrams that involve complex phases are brought into one of the forms in \eqref{eq:complex_phases}, the real parts of the diagram can be brought into the normal form given in Lemma \ref{lem:modulus_nf}.

 The resulting normal form for complex non-zero scalars can easily be seen to be unique.
\end{proof}

\begin{lem}\label{lem:innerprod_wlog}
 The inner product between a red and a green node with phase angles $\alpha$ and $\beta$ is defined only by the set $\{\alpha,\beta\}$, it does not matter which label is assigned to the green and which to the red node, or whether it is a green state and red effect, or conversely. Diagrammatically:
 \begin{equation}
  \innerprod{gn, label={[gphase]right:$\alpha$}}{rn, label={[rphase]right:$\beta$}} \; = \; \innerprod{rn, label={[rphase]right:$\beta$}}{gn, label={[gphase]right:$\alpha$}} \; = \; \innerprod{gn, label={[gphase]right:$\beta$}}{rn, label={[rphase]right:$\alpha$}} \; = \; \innerprod{rn, label={[rphase]right:$\alpha$}}{gn, label={[gphase]right:$\beta$}}.
 \end{equation}
\end{lem}
\begin{proof}
 The first equality results from the topology meta rule.
 The equality between the second and third diagram follows from the colour change rule and the fact that the Hadamard node is self-inverse:
 \begin{equation}
  \input{tikz_files/colour_change_der2.0.tikz}.
 \end{equation}
 The last equality is again by the topology meta rule.
\end{proof}

\begin{lem}\label{lem:pi_multiplication}
 For any pair of phase angles $\alpha$ and $\beta$, the complex phases resulting from the inner products of \state{rn, label={[rphase]right:$\pi$}} with phased green effects can be combined into just one subdiagram and a normalising factor \innerprodgr:
 \begin{equation}
  \innerprod{gn, label={[gphase]right:$\alpha$}}{rn, label={[rphase]right:$\pi$}} \; \innerprod{gn, label={[gphase]right:$\beta$}}{rn, label={[rphase]right:$\pi$}} \; = \; \innerprod{gn}{rn} \; \innerprod{gn, label={[gphase]right:$\alpha+\beta$}}{rn, label={[rphase]right:$\pi$}}.
 \end{equation}
\end{lem}
\begin{proof}
 We have:
 \begin{equation}\label{eq:scalar_mult}
  \innerprod{gn, label={[gphase]right:$\alpha$}}{rn, label={[rphase]right:$\pi$}} \; \innerprod{gn, label={[gphase]right:$\beta$}}{rn, label={[rphase]right:$\pi$}} \; = \; \innerprod{gn}{rn} \; \input{tikz_files/multiplication_pi_der.tikz} \; = \; \innerprod{gn}{rn} \; \innerprod{gn, label={[gphase]right:$\alpha+\beta$}}{rn, label={[rphase]right:$\pi$}}
 \end{equation}
 using the copy rule, $\pi$-copy rule, and spider rule.
\end{proof}

\begin{lem}\label{lem:overlap_with_ket_zero}
 For any phase angle $\alpha$, the inner product of a green effect of phase $\alpha$ with \state{rn} is equal to \innerprodgr{}:
 \begin{equation}
  \innerprod{gn, label={[gphase]right:$\alpha$}}{rn} \; = \; \innerprod{gn}{rn}. 
 \end{equation}
\end{lem}
\begin{proof}
 The case $\alpha=0$ is trivial. For $\alpha=\pi$, note that by the spider and $\pi$-copy rules:
 \begin{equation}\label{eq:pi_remove}
  \input{tikz_files/pi_remove_der2.0.tikz}.
 \end{equation}
 Now, using \eqref{eq:pi_remove}, \eqref{eq:halfscalar_innerprodgr2}, and Lemma \ref{lem:pi_multiplication} with $\beta=-\alpha$, together with various rewrite rules, yields:
 \begin{align}
  \innerprod{gn, label={[gphase]right:$\alpha$}}{rn} \;
  &= \; \halfscalar \;\; \innerprod{gn}{rn} \;\; \innerprod{gn}{rn} \;\; \innerprod{gn, label={[gphase]right:$\alpha$}}{rn} \;
  = \; \halfscalar \;\; \innerprod{gn}{rn} \;\; \innerprod{gn}{rn,label={[rphase]right:$\pi$}} \; \innerprod{gn, label={[gphase]right:$\alpha$}}{rn} \;
  = \; \halfscalar \; \innerprod{gn, label={[gphase]right:$-\alpha$}}{rn, label={[rphase]right:$\pi$}} \; \innerprod{gn, label={[gphase]right:$\alpha$}}{rn, label={[rphase]right:$\pi$}} \; \innerprod{gn, label={[gphase]right:$\alpha$}}{rn} \;
  = \; \halfscalar \; \innerprod{gn, label={[gphase]right:$-\alpha$}}{rn, label={[rphase]right:$\pi$}} \;\; \innerprod{gn}{rn} \;\; \input{tikz_files/overlap_with_ket_zero_1.tikz} \nonumber \\
  &= \; \halfscalar \; \innerprod{gn, label={[gphase]right:$-\alpha$}}{rn, label={[rphase]right:$\pi$}} \;\; \innerprod{gn}{rn} \;\; \input{tikz_files/overlap_with_ket_zero_2.tikz} \;
  = \; \halfscalar \; \innerprod{gn, label={[gphase]right:$-\alpha$}}{rn, label={[rphase]right:$\pi$}} \;\; \innerprod{gn}{rn} \;\; \input{tikz_files/overlap_with_ket_zero_3.tikz} \;
  = \; \halfscalar \; \innerprod{gn, label={[gphase]right:$-\alpha$}}{rn, label={[rphase]right:$\pi$}} \; \innerprod{gn, label={[gphase]right:$\alpha$}}{rn, label={[rphase]right:$\pi$}} \;\; \innerprod{gn}{rn} \;\;
  = \;\; \innerprod{gn}{rn},
 \end{align}
 thus proving the result for any $\alpha$.
\end{proof}

\begin{lem}\label{lem:y-states}
 The states \state{gn, label={[gphase]right:$\pi/2$}} and \state{rn, label={[rphase]right:$-\pi/2$}} are equal up to a complex phase:
 \begin{equation}
  \state{rn, label={[rphase]right:$-\pi/2$}} = \halfscalar \; \innerprod{gn, label={[gphase]right:$-\pi/2$}}{rn, label={[rphase]right:$-\pi/2$}} \; \state{gn, label={[gphase]right:$\pi/2$}}.
 \end{equation}
\end{lem}
\begin{proof}
 The desired equality can be derived as follows:
 \begin{align}
  \state{rn, label={[rphase]right:$-\pi/2$}} \;
  &= \; \input{tikz_files/y_state_der_1.tikz} \;
  = \; \halfscalar \; \innerprod{gn, label={[gphase]right:$-\pi/2$}}{rn, label={[rphase]right:$-\pi/2$}} \; \input{tikz_files/y_state_der_2.tikz} \;
  = \; \halfscalar \; \innerprod{gn, label={[gphase]right:$-\pi/2$}}{rn, label={[rphase]right:$-\pi/2$}} \; \input{tikz_files/y_state_der_3.tikz} \nonumber \\
  &= \; \halfscalar \; \halfscalar \;\; \innerprod{gn}{rn} \;\; \innerprod{gn}{rn} \;\; \innerprod{gn, label={[gphase]right:$-\pi/2$}}{rn, label={[rphase]right:$-\pi/2$}} \; \input{tikz_files/y_state_der_3.tikz} \;
  = \; \halfscalar \; \halfscalar \;\; \innerprod{gn}{rn} \;\; \innerprod{gn, label={[gphase]right:$-\pi/2$}}{rn, label={[rphase]right:$-\pi/2$}} \; \innerprod{gn}{rn, label={[rphase]right:$\pi/2$}} \; \state{gn, label={[gphase]right:$\pi/2$}} \nonumber \\
  &= \; \halfscalar \; \innerprod{gn, label={[gphase]right:$-\pi/2$}}{rn, label={[rphase]right:$-\pi/2$}} \; \state{gn, label={[gphase]right:$\pi/2$}},
 \end{align}
 using the colour change rule, the Euler decomposition rule, the copy rule, \eqref{eq:halfscalar_innerprodgr2}, and Lemma \ref{lem:overlap_with_ket_zero}.
\end{proof}

\begin{lem}\label{lem:omega_inverses}
 The scalar diagrams \halfscalar{} \innerprod{gn,label={[gphase]right:$-\pi/2$}}{rn,label={[rphase]right:$-\pi/2$}} and \halfscalar{} \innerprod{gn,label={[gphase]right:$\pi/2$}}{rn,label={[rphase]right:$\pi/2$}} are inverse to each other:
 \begin{equation}
  \halfscalar \innerprod{gn,label={[gphase]right:$-\pi/2$}}{rn,label={[rphase]right:$-\pi/2$}} \halfscalar \innerprod{gn,label={[gphase]right:$\pi/2$}}{rn,label={[rphase]right:$\pi/2$}} = \quad.
 \end{equation}
\end{lem}
\begin{proof}
 By the identity rule, spider rule, Euler decomposition rule, colour change rule, copy rule, and Lemma \ref{lem:overlap_with_ket_zero}:
 \begin{align}\label{eq:scalar_pi_2_inverse}
  \innerprod{gn,label={[gphase]right:$-\pi/2$}}{rn,label={[rphase]right:$-\pi/2$}} \; \innerprod{gn,label={[gphase]right:$\pi/2$}}{rn,label={[rphase]right:$\pi/2$}}
  &= \; \input{tikz_files/multiplication5_der_a.tikz} \nonumber \\
  &= \; \innerprod{gn}{rn} \; \innerprod{gn}{rn} \; \innerprod{gn}{rn} \; \input{tikz_files/multiplication5_der_b.tikz} \;
  = \; \innerprod{gn}{rn} \; \innerprod{gn}{rn} \; \innerprod{gn}{rn} \; \innerprod{gn,label={[gphase]right:$-\pi/2$}}{rn} \;
  = \; \innerprod{gn}{rn} \; \innerprod{gn}{rn} \; \innerprod{gn}{rn} \; \innerprod{gn}{rn}.
 \end{align}
 The desired equality follows by multiplying both sides with \halfscalar{} \halfscalar{} and using \eqref{eq:halfscalar_innerprodgr2}.
\end{proof}

The existence of a unique normal form for non-zero stabilizer scalars immediately implies the following.

\begin{thm}
 The \ZX-calculus is complete for non-zero stabilizer scalars.
\end{thm}
\begin{proof}
 Suppose \genscalar{$s$} and \genscalar{$r$} are two diagrams in the stabilizer \ZX-calculus such that:
 \begin{equation}
  \intf{\genscalar{$s$}} = \intf{\genscalar{$r$}}.
 \end{equation}
 Then both diagrams must be rewritable to the same normal form diagram.
 As all rewrite rules are invertible, this implies that \genscalar{$s$} can be rewritten into \genscalar{$r$}, or conversely.
\end{proof}

Lemmas \ref{lem:innerprod_wlog}--\ref{lem:omega_inverses} actually hold in the general \ZX-calculus, not just in the stabilizer fragment.
Nevertheless, the normal form for stabilizer scalars is unlikely to be extendable to arbitrary scalars as the incompleteness result \cite{schroeder_incomplete_2014} implies the existence of scalar diagrams in the general \ZX-calculus which cannot be decomposed into  two node-segments.

As shown in Proposition \ref{prop:scalar_interpretation}, the non-zero stabilizer scalar diagrams represent numbers of the form:
\begin{equation}
 \sqrt{2^r} e^{i\pi s/4},
\end{equation}
where $r,s$ are integers.
Thus we find the following.

\begin{thm}\label{thm:scalars_group}
 The non-zero stabilizer scalar normal form diagrams form a group, which is isomorphic to $\ZZ_8\times\ZZ$.
\end{thm}
\begin{proof}
 By Theorem \ref{thm:scalar_nf}, the set of stabilizer scalars is closed under multiplication.
 The unit of this multiplication is the empty diagram.
 Given a scalar \genscalar{$s$} in normal form, its inverse can be constructed by taking the Hermitian adjoint of \genscalar{$s$}, doing the following replacement:
 \begin{align*}
  \halfscalar \; &\mapsto \; \innerprodgr \; \innerprodgr, \text{ and} \\
  \input{tikz_files/innerprodrg.tikz} \; &\mapsto \; \halfscalar \; \innerprodgr,
 \end{align*}
 and then applying the variant star rule to simplify the resulting diagram.
 Thus the stabilizer scalar normal form diagrams form a group.

 The isomorphism to $\ZZ_8\times\ZZ$ follows from Proposition \ref{prop:scalar_interpretation}.
\end{proof}

This theorem justifies the assumption of invertibility of non-zero scalars made throughout most of Chapter \ref{ch:completeness}.

\section{Completeness for scaled stabilizer diagrams}
\label{s:scaled_completeness}

Given the completeness result for scalar-free stabilizer diagrams derived in Chapter \ref{ch:completeness} and the completeness result for non-zero stabilizer scalars from Section \ref{s:scalar_completeness}, the obvious next step is to combine them into a completeness result for non-zero scaled stabilizer diagrams.
All that is missing for a full completeness result is a completeness proof for stabilizer zero diagrams.

First, we combine the previous stabilizer completeness results into a completeness proof for arbitrary non-zero stabilizer diagrams.
We then give a normal form for stabilizer zero diagrams and prove that it is unique; this immediately implies the stabilizer \ZX-calculus is complete for zero diagrams.
Section \ref{s:full_stabilizer_completeness} contains the full stabilizer completeness proof.
Finally, we give an example application for the scalar completeness results.

\subsection{Completeness for non-zero stabilizer diagrams}
\label{s:non-zero_completeness}

Using the completeness result for non-zero stabilizer scalars derived in Section \ref{s:scalar_completeness} and the ability to derive equalities between non-scalar stabilizer diagrams as in Section \ref{s:scalar-free_equalities}, we can now prove that the \ZX-calculus is complete for arbitrary non-zero diagrams.

\begin{thm}\label{thm:non-zero_completeness}
 The \ZX-calculus is complete for non-zero scaled diagrams, i.e.\ diagrams that contain both scalar and non-scalar parts.
\end{thm}
\begin{proof}
 Given two scaled diagrams $D$ and $D'$, first consider the non-scalar parts.
 Assume these are both operators from $n$ to $m$ qubits; if the numbers of inputs or outputs do not match up the diagrams are trivially distinct.
 Proceed according to the following steps:
 \begin{enumerate}
  \item Use the Choi-Jamio{\l}kowski isomorphism to bend all inputs into outputs.
  \item Bring the resulting states into GS-LC form, and thus into rGS-LC form, keeping track of the scalars somewhere on the side.
  \item Simplify the pair of rGS-LC diagrams, again keeping track of scalars on the side.
 \end{enumerate}
 Now if the resulting rGS-LC diagrams, ignoring scalars, are not identical then the original diagrams cannot be equal.
 This is because multiplying two different linear operators by scalars can only make them equal if one was a scaled version of the other to begin with.

 If the simplified rGS-LC diagrams are identical, proceed by bringing the scalars into normal form as described in Theorem \ref{thm:scalar_nf}.
 This normalisation does not change the non-scalar part of the diagram.

 Now if the two resulting diagrams are identical, apply the Choi-Jamio{\l}kowski isomorphism to transform the rewrites performed in the non-scalar part of this process into rewrites of the original diagrams.
 The Choi-Jamio{\l}kowski isomorphism preserves equalities, so this yields a sequence of rewrite steps transforming the original two diagrams into two new diagrams that are identical to each other.
 Then some of the rewrite steps can be inverted to find a series of rewrites transforming $D$ into $D'$ (or conversely), thus proving that the two diagrams are equal according to the rules of the graphical calculus.

 Otherwise, the two diagrams must represent different operators, as multiplying the same operator by two different scalars yields two different operators.
\end{proof}

\subsection{Completeness for stabilizer zero diagrams}
\label{s:zero_nf}

We have shown that the stabilizer \ZX-calculus is complete for scaled diagrams as long as they are non-zero.
By Corollary \ref{cor:recognize_zero}, any stabilizer zero diagram can be rewritten to explicitly contain one of the following scalar diagrams as a subdiagram:
\begin{equation}\label{eq:zero_scalars2}
 \scalar{gn, label={[gphase]right:$\pi$}}, \quad\quad \scalar{rn, label={[rphase]right:$\pi$}}, \quad\quad \innerprod{gn, label={[gphase]right:$\pi/2$}}{rn, label={[rphase]right:$-\pi/2$}}, \quad \text{ or } \quad \innerprod{gn, label={[gphase]right:$-\pi/2$}}{rn, label={[rphase]right:$\pi/2$}}.
\end{equation}
Of course the calculus does not actually contain four distinct representations of 0, as shown in the following lemma.

\begin{lem}\label{lem:unique_zero}
 Any diagram that contains one of the subdiagrams in \eqref{eq:zero_scalars2} can be rewritten to contain \scalar{gn, label={[gphase]right:$\pi$}}.
\end{lem}
\begin{proof}
 By the colour change law, $\scalar{rn, label={[rphase]right:$\pi$}} \; = \; \scalar{gn, label={[gphase]right:$\pi$}}$.
 Using Lemma \ref{lem:y-states}, we find that:
 \begin{equation}
  \innerprod{gn, label={[gphase]right:$\pi/2$}}{rn, label={[rphase]right:$-\pi/2$}} \; = \; \halfscalar \; \innerprod{gn, label={[gphase]right:$-\pi/2$}}{rn, label={[rphase]right:$-\pi/2$}} \; \scalar{gn, label={[gphase]right:$\pi$}}.
 \end{equation}
 This result also applies to \innerprod{gn, label={[gphase]right:$-\pi/2$}}{rn, label={[rphase]right:$\pi/2$}} by Lemma \ref{lem:innerprod_wlog}.
\end{proof}

This lemma allows the zero rule and the zero scalar rule to be applied to any zero diagram.

\begin{thm}\label{thm:zero_completeness}
 The stabilizer \ZX-calculus is complete for zero diagrams.
\end{thm}
\begin{proof}
 We show that any zero diagram with $n$ inputs and $m$ outputs can be rewritten into the following normal form:
 \begin{equation}\label{eq:zero_normal_form}
  \input{tikz_files/zero_normal_form.tikz}.
 \end{equation}
 First, note that the normal form is clearly unique as a zero matrix, the interpretation of any zero diagram, is fully determined by its dimensions.
 
 Now, to rewrite a zero diagram into normal form, first apply the Euler decomposition rule to remove all Hadamard nodes.
 Apply the spider rule until all remaining edges are between one red and one green node.
 Then apply the zero rule (or its upside-down equivalent, depending on how the colours match up) to each edge, transforming the diagram into a completely disconnected graph where the only remaining edges are inputs and outputs of the diagram.
 Further applications of the zero rule (upside-down or not) can be used to change the colour of the nodes connected to inputs or outputs.
 Other than one copy of \scalar{gn, label={[gphase]right:$\pi$}}, any disconnected red and green nodes can be removed using the zero scalar rule.
 Finally, any copies of \halfscalar{} can be removed by using the zero scalar rule to create a new copy of \scalar{gn}, and then applying the star rule.
 This leaves the diagram in the normal form \eqref{eq:zero_normal_form}.
 
 As all rewrite rules are invertible, this implies that any two zero diagrams with the same numbers of inputs and outputs can be rewritten into each other.
 Therefore the \ZX-calculus is complete for stabilizer zero diagrams.
\end{proof}

The results derived in this section are not actually specific to stabilizer diagrams.
In fact, given the zero and zero scalar rules, any diagram that explicitly contains \scalar{gn, label={[gphase]right:$\pi$}} can be brought into the normal form given in \eqref{eq:zero_normal_form}.
Still, Theorem \ref{thm:zero_completeness} only holds for the stabilizer \ZX-calculus, as for a larger fragment of the calculus it may not always be possible to rewrite zero diagrams so that they explicitly contain \scalar{gn, label={[gphase]right:$\pi$}} as a subdiagram.

\subsection{The full stabilizer completeness result}
\label{s:full_stabilizer_completeness}

The completeness result for non-zero stabilizer diagrams and the one for stabilizer zero diagrams straightforwardly combine to a completeness result for arbitrary scaled stabilizer diagrams.

\begin{thm}
 The stabilizer \ZX-calculus is complete, i.e.\ for two \ZX-calculus diagrams $D_1$ and $D_2$ in which all phase angles are integer multiples of $\pi/2$:
 \begin{equation}
  \left\llbracket D_1 \right\rrbracket = \left\llbracket D_2 \right\rrbracket \quad \implies \quad D_1 = D_2,
 \end{equation}
 where the second equality is according to the rewrite rules given in Section \ref{s:rewrite_rules}.
\end{thm}
\begin{proof}
 Bring the diagram into GS-LC form (up to Choi-Jamio{\l}kowski isomorphism) and decompose the scalars.
 If the diagram is zero, proceed as in Theorem \ref{thm:zero_completeness}.
 Otherwise, proceed as in Theorem \ref{thm:non-zero_completeness}.
\end{proof}

Thus any equality about pure state stabilizer quantum mechanics with post-selected measurements that can be derived using matrices can also be derived graphically, i.e.\ the \ZX-calculus has the same power as any other formalism for stabilizer quantum theory.

\subsection{Example: Quantum key distribution}
\label{s:example_QKD}

Consider the BB84 protocol for quantum key distribution using a two-qubit Bell state \cite{bennett_quantum_1984}: Alice and Bob each hold one half of the entangled state, and they each randomly decide to measure their qubit in either the computational basis $\{\ket{0}, \ket{1}\}$ or the Hadamard basis $\{\ket{+}, \ket{-}\}$.
They later compare their choice of basis over a classical communication channel, e.g.\ a telephone line.
Assuming the Bell state is $\frac{1}{\sqrt{2}}\left(\ket{00}+\ket{11}\right)$, if they have picked the same basis, their results always agree and they have a key; if they have picked different bases, their results are uncorrelated.
We are not interested in the security properties of this key distribution scheme here, instead we simply use it as motivation for some \ZX-calculus derivations since all the states and operations required for the protocol are in the stabilizer formalism.

The Bell state given above is represented in the \ZX-calculus as a ``cup'' with some normalising factors:
\begin{equation}
 \frac{1}{\sqrt{2}}\left(\ket{00}+\ket{11}\right) = \left\llbracket \halfscalar \; \innerprodgr \; \input{tikz_files/cup_diagram.tikz} \right\rrbracket.
\end{equation}
Measurements in the \ZX-calculus are post-selected, so we have to consider each pair of measurement outcomes in turn.
Graphically, the normalised outcomes of computational and Hadamard basis measurements on single qubits are:
\begin{equation}
 \bra{0}= \left\llbracket \halfscalar{} \; \innerprodgr{} \; \effect{rn} \right\rrbracket, \quad
 \bra{1} = \left\llbracket \halfscalar{} \; \innerprodgr{} \; \effect{rn, label={[rphase]right:$\pi$}} \right\rrbracket, \quad
 \bra{+}= \left\llbracket \halfscalar{} \; \innerprodgr{} \; \effect{gn} \right\rrbracket, \quad\text{and}\quad
 \bra{-} = \left\llbracket \halfscalar{} \; \innerprodgr{} \; \effect{gn, label={[gphase]right:$\pi$}} \right\rrbracket.
\end{equation}

First, assume Alice and Bob both measure in the computational basis. Can they both get outcome 0?
Constructing the \ZX-calculus diagram for the overlap between the Bell state and $\bra{00}$ and then simplifying it using \eqref{eq:halfscalar_innerprodgr2}, the topology rule, the spider rule, and the star rule, yields:
\begin{equation}
 \halfscalar \; \halfscalar \; \halfscalar \; \innerprodgr \; \innerprodgr \; \innerprodgr \; \input{tikz_files/bell1.tikz} \;
 = \; \halfscalar \; \halfscalar \; \innerprodgr \; \innerprod{rn}{rn} \;
 = \; \halfscalar \; \halfscalar \; \innerprodgr \; \scalar{rn} \;
 = \; \halfscalar \; \innerprodgr,
\end{equation}
which is non-zero, so this outcome is possible.
The probability of this outcome can be found by multiplying this amplitude with its dagger, graphically:
\begin{equation}
 \halfscalar \; \input{tikz_files/innerprodrg.tikz} \; \halfscalar \; \innerprodgr \;
 = \; \halfscalar \; \halfscalar \; \innerprodgr \; \innerprodgr \;
 = \; \halfscalar,
\end{equation}
i.e.\ the probability of Alice and Bob both measuring 0 is $\llbracket\halfscalar\rrbracket=1/2$.
Similarly, the overlap of the Bell state with the effect $\bra{11}$ is, graphically:
\begin{equation}
 \halfscalar \; \halfscalar \; \halfscalar \; \innerprodgr \; \innerprodgr \; \innerprodgr \; \input{tikz_files/bell5.tikz} \;
 = \; \halfscalar \; \halfscalar \; \innerprodgr \; \innerprod{rn, label={[rphase]right:$\pi$}}{rn, label={[rphase]right:$\pi$}} \;
 = \; \halfscalar \; \halfscalar \; \innerprodgr \; \scalar{rn} \;
 = \; \halfscalar \; \innerprodgr,
\end{equation}
so the probability of Alice and Bob both getting measurement outcome 1 is again $1/2$.

On the other hand, if we consider whether Alice can get 0 while Bob gets 1, we find the following:
\begin{equation}
 \halfscalar \; \halfscalar \; \halfscalar \; \innerprodgr \; \innerprodgr \; \innerprodgr \; \input{tikz_files/bell2.tikz} \;
 = \; \halfscalar \; \halfscalar \; \; \innerprodgr \; \innerprod{rn}{rn,label={[rphase]right:$\pi$}} \;
 = \; \halfscalar \; \halfscalar \; \innerprodgr \; \scalar{rn,label={[rphase]right:$\pi$}} \;
 = \; \halfscalar \; \halfscalar \; \scalar{gn} \; \scalar{rn} \; \scalar{rn,label={[rphase]right:$\pi$}} \;
 = \; \scalar{gn,label={[gphase]right:$\pi$}},
\end{equation}
where the penultimate equality is by the zero rule and the last one by the star and colour change rules. As $\llbracket\scalar{gn,label={[gphase]right:$\pi$}}\rrbracket=0$, that combination of outcomes is impossible.

If Alice measures in the computational basis and Bob in the Hadamard basis, the probability of Alice getting outcome $\theta$ and Bob getting $\phi$ for some fixed $\theta,\phi\in\{0,\pi\}$ is:
\begin{equation}
 \begin{array}{c} \halfscalar \; \halfscalar \; \halfscalar \; \input{tikz_files/innerprodrg.tikz} \; \input{tikz_files/innerprodrg.tikz} \; \input{tikz_files/innerprodrg.tikz} \\ \halfscalar \; \halfscalar \; \halfscalar \; \innerprodgr \; \innerprodgr \; \innerprodgr \end{array} \input{tikz_files/bell3.tikz} \;
 = \; \halfscalar \; \halfscalar \; \halfscalar \; \innerprod{gn,label={[gphase]right:$-\phi$}}{rn,label={[rphase]right:$-\theta$}} \; \innerprod{gn,label={[gphase]right:$\phi$}}{rn,label={[rphase]right:$\theta$}} \;
 = \; \halfscalar \; \halfscalar \; \halfscalar \; \innerprod{gn}{rn} \; \innerprod{gn}{rn} \;
 = \; \halfscalar \; \halfscalar,
\end{equation}
where the penultimate equality is by Lemmas \ref{lem:pi_multiplication} and \ref{lem:overlap_with_ket_zero}.
Thus if Alice and Bob choose different bases, their outcomes are random and uncorrelated: each of the four combinations of outcomes has probability $\llbracket\halfscalar\;\halfscalar\rrbracket = 1/4$.

\section{Completeness for the single-qubit Clifford+T group}
\label{s:Clifford+T}

We have shown in the previous chapter and sections that the \ZX-calculus is complete for pure state stabilizer quantum mechanics with post-selected measurements even though it is incomplete for universal pure state qubit quantum mechanics.
While not obvious, it is also not surprising that a finite set of rewrite rules suffices to derive all equalities between stabilizer operators since there are only finitely many stabilizer operations on any fixed finite number of qubits.

The addition of any non-stabilizer operation to a set of generating operations for the Clifford group yields an infinitely large set of operations, which is moreover dense in the space of all pure state qubit operations.
Nevertheless, the incompleteness proof does not hold for any such approximately universal set of operations (cf.\ Definition \ref{dfn:approximately_universal}) because the set of allowed basic $X$- and $Z$-rotations is still restricted and thus the Euler decomposition rule does not hold in general.

In this section, we make a step towards a completeness result for the Clifford+T group, which is approximately universal, by showing that the \ZX-calculus is complete for scalar-free single-qubit unitaries within this group, i.e.\ for line diagrams where all phases are integer multiples of $\pi/4$.
The result was originally published in \cite{backens_zx-calculus_2014}.

For simplicity, and because we are only considering unitary operators, we ignore scalars in this section.
As a sequential composite of unitary operators can never be zero, we do not need to consider zero diagrams.
When rewriting line diagrams built from phase shifts and Hadamard nodes, the only scalars that can arise are those that appear in rewrite rules (excluding the scalar and zero rules).
This means that ``ignoring scalars'' is equivalent to replacing the star rule from Section \ref{s:rewrite_rules} with the following \emph{scalar rule}:
\begin{equation}
 \genscalar{$s$} \; = \qquad,
\end{equation}
for any:
\begin{equation}
 \genscalar{$s$} \in \left\{ \innerprod{gn}{rn}, \quad \innerprod{gn,label={[gphase]right:$\alpha$}}{rn,label={[rphase]right:$\pi$}}, \quad \innerprod{gn,label={[gphase]right:$-\pi/2$}}{rn,label={[rphase]right:$-\pi/2$}} \right\},
\end{equation}
with $\alpha$ an integer multiple of $\pi/4$.

The scalar rule can be used left-to-right and right-to-left, therefore non-trivial scalars can be ``spawned'' where needed for other rewrite rules and scalars appearing after rewriting can be dropped.
This is done implicitly in the following subsections.

We first give some definitions and lemmas in Section \ref{s:Clifford+T_definitions}.
The completeness proof -- again via a unique normal form -- is contained in Section \ref{s:Clifford+T_completeness}.

\subsection{Preliminary definitions and lemmas}
\label{s:Clifford+T_definitions}

We denote the single-qubit Clifford group by $\mathcal{C}_1$.
From the results in Section \ref{s:full_stabilizer_completeness}, we know that the \ZX-calculus is complete for this group.
In fact, Lemma \ref{lem:single-qubit_Clifford_normal_form} gives a unique normal form for any single-qubit Clifford operator.
There are also other possible choices of normal forms for single-qubit Clifford operators, some of which are useful later.

\begin{lem}\label{lem:more_C1_normal_forms}
 The following two sets each contain a unique representation for each operator $C\in\mathcal{C}_1$:
 \begin{equation}\label{eq:local_Clifford_normal_form2}
  \left\{
  \input{tikz_files/C1-1.tikz}
  \right\} \quad\quad\text{and}\quad\quad
  \left\{
  \input{tikz_files/C1-2.tikz}
  \right\},
 \end{equation}
 where in both cases $\alpha,\beta,\gamma\in\{0,\pi/2,\pi,-\pi/2\}$.
\end{lem}
The proof is analogous to that of Lemma \ref{lem:single-qubit_Clifford_normal_form}.

It will be useful to define two further sets of \ZX-calculus operators:
 \begin{equation}
  \input{tikz_files/W.tikz}
  \qquad \text{and} \qquad
  \input{tikz_files/V.tikz}
 \end{equation}
In the remainder of this section we prove various lemmas about how operators in $\mathcal{C}_1$, $\mathcal{W}$ and $\mathcal{V}$ compose.

\begin{lem}\label{lem:U}
 The following set contains a unique representation of each operator of the form $TC$, where $C\in\mathcal{C}_1$:
 \begin{equation}
  \mathcal{U} = \left\{
  \input{tikz_files/U.tikz}
  \right\}
 \end{equation}
 if $\alpha,\beta,\gamma\in\{0,\pi/2,\pi,-\pi/2\}$.
\end{lem}
\begin{proof}
 This follows immediately from the second set of single-qubit Clifford normal forms given in Lemma \ref{lem:more_C1_normal_forms}.
\end{proof}

\begin{lem}\label{lem:CV}
 Let $C\in\mathcal{C}_1$, $U\in\mathcal{U}$ and $V\in\mathcal{V}$. Then:
 \begin{equation}
  \input{tikz_files/VC.tikz}
  \qquad \text{and} \qquad
  \input{tikz_files/UC.tikz}
 \end{equation}
 for some $W\in\mathcal{W}$, $U'\in\mathcal{U}$, $V'\in\mathcal{V}$ and $a,b\in\{0,1\}$. For the particular case of the first equality where $C$ consists solely of $\pi$ phase shifts, $W$ is the identity and we have:
 \begin{equation}
  \input{tikz_files/V-pi_a.tikz}, \qquad \input{tikz_files/V-pi_b.tikz}, \quad \text{and} \quad \input{tikz_files/V-pi_c.tikz},
 \end{equation}
 with $\bar{V}\in\mathcal{V}\setminus\{V\}$.
\end{lem}
\begin{proof}
 Substitute for $C$ using the first set of normal forms given in Lemma \ref{lem:more_C1_normal_forms} and for $V$ and $U$ using the definitions of $\mathcal{V}$ and $\mathcal{U}$; the results then follow from straightforward application of the rules of the \ZX-calculus.
\end{proof}

\begin{lem}\label{lem:pi-commutation}
 Suppose $V_1,\ldots,V_n\in\mathcal{V}$ for some positive integer $n$. Then if $a,b\in\{0,1\}$:
 \begin{equation}
  \input{tikz_files/VV-pi.tikz}
 \end{equation}
 for some $a',b'\in\{0,1\}$ and $V_1',\ldots,V_n'\in\mathcal{V}$.
\end{lem}
\begin{proof}
 By induction on $n$, using the second part of Lemma \ref{lem:CV}.
\end{proof}

Lemmas \ref{lem:CV} and \ref{lem:pi-commutation} show that single-qubit Clifford operators interact nicely with the diagrams in the sets $\mathcal{U}$, $\mathcal{V}$, and $\mathcal{W}$.
Red and green $\pi$-phase shifts in particular can be moved past operators from $\mathcal{V}$ in a generalisation of the $\pi$-commutation rule.

\subsection{The Clifford+T completeness proof}
\label{s:Clifford+T_completeness}

We now use the definitions in the previous section to define a normal form for single-qubit Clifford+T operators and show that it is unique.
This proof is inspired by an analogous result for quantum circuits in \cite{matsumoto_representation_2008}.
As before, the existence of a unique normal form implies completeness.

\begin{thm}\label{thm:Clifford+T_diagrams}
 Any single-qubit operator consisting of phase shifts that are multiples of $\pi/4$ and Hadamard operators is either a pure Clifford operator or it can be written in the normal form:
 \begin{equation}\label{eq:UVVW}
  \input{tikz_files/UVVW.tikz}
 \end{equation}
 for some integer $n\geq 0$, where $W\in\mathcal{W}$, $V_1,\ldots,V_n\in\mathcal{V}$ and $U\in\mathcal{U}$.
\end{thm}

\begin{proof}
 Any single-qubit Clifford+T operator can be written solely in terms of \phase{rn,label={[rphase]right:$\pi/2$}} and \phase{gn,label={[gphase]right:$\pi/4$}}.
 To prove the theorem, it thus suffices to show that adding \phase{rn,label={[rphase]right:$\pi/2$}} or \phase{gn,label={[gphase]right:$\pi/4$}} to any Clifford operator or any diagram in normal form yields a diagram that can be rewritten to a Clifford operator or normal form diagram.

 Consider first \phase{rn,label={[rphase]right:$\pi/2$}}.
 This is a Clifford operator, so adding it to a Clifford diagram yields another Clifford diagram.
 Furthermore:
 \begin{equation}
  \input{tikz_files/W-pi_2.tikz} = \; \genunitary{$C$}
 \end{equation}
 for some $C\in\mathcal{C}_1$, so if $n>0$:
 \begin{equation}
  \input{tikz_files/UVVW-pi_2.tikz}
  = \; \input{tikz_files/UVVC.tikz} \;
  = \; \input{tikz_files/U-pi-pi-VVW.tikz}
  = \; \input{tikz_files/UVVW-prime.tikz},
 \end{equation}
 by Lemmas \ref{lem:CV} and \ref{lem:pi-commutation}, where $a,b\in\{0,1\}$, $W'\in\mathcal{W}$, $U'\in\mathcal{U}$ and $V_1',\ldots,V_n'\in\mathcal{V}$.
 From Lemma \ref{lem:CV}, we also have that, if $n=0$, the diagram resulting from the application of \phase{rn,label={[rphase]right:$\pi/2$}} to a normal form diagram can be rewritten into normal form.
 This covers all the cases.

 Now consider \phase{gn,label={[gphase]right:$\pi/4$}} instead.
 Note that:
 \begin{equation}
  \input{tikz_files/C-pi_4.tikz} = \; \genunitary{$U'$} \qquad \text{and} \qquad
  \input{tikz_files/U-pi_4.tikz} = \; \genunitary{$C'$}
 \end{equation}
 for some $U'\in\mathcal{U}$ and $C'\in\mathcal{C}_1$.
 Furthermore, unless $W$ is the identity:
 \begin{equation}
  \input{tikz_files/W-pi_4.tikz}
 \end{equation}
 for some $V\in\mathcal{V}$.
 Thus adding \phase{gn,label={[gphase]right:$\pi/4$}} to a Clifford operator or a normal form diagram with non-trivial $W$ results in diagrams that can be rewritten to normal form.
 If $W$ is the identity and $n=0$, then the result of adding \phase{gn,label={[gphase]right:$\pi/4$}} is a Clifford diagram.

 It remains to check what happens when $W$ is the identity and $n>0$.
 For any $V_n\in\mathcal{V}$, we can find $W\in\mathcal{W}$ and $a\in\{0,1\}$ such that:
 \begin{equation}
  \input{tikz_files/V-pi_4.tikz}.
 \end{equation}
 Then by Lemmas \ref{lem:CV} and \ref{lem:pi-commutation}, the entire diagram can be brought into normal form.

 Thus, whenever \phase{rn,label={[rphase]right:$\pi/2$}} or \phase{gn,label={[gphase]right:$\pi/4$}} is added to a pure Clifford diagram or a normal form diagram, the resulting diagram can be rewritten into a pure Clifford diagram or a normal form diagram, completing the proof.
\end{proof}

Theorem \ref{thm:Clifford+T_diagrams} proves that any proper Clifford+T operator can be brought into normal form.
To get a completeness result, it remains to show that this normal form is unique.
We proceed in several steps, following \cite{matsumoto_representation_2008}: Firstly, we show that no non-trivial normal form diagram represents the identity map; this is Theorem \ref{thm:not_identity}.
Then, we prove that the Hermitian adjoint of any normal form diagram has a normal form with the same number of non-Clifford phase shifts, see Lemma \ref{lem:adjoint}.
Finally, we combine those two results in Theorem \ref{thm:Clifford+T_normal_form} to show that normal forms are unique.

\begin{thm}\label{thm:not_identity}
 No normal form diagram as given in \eqref{eq:UVVW} is equal to the identity.
\end{thm}
\begin{proof}
 We show that if $D$ is a normal form diagram, then there does not exist a complex number $c$ such that $\intf{D}=cI$, where $I$ is the single-qubit identity operator.
 As the \ZX-calculus is sound, this implies that no normal form diagram is equal to the identity within the \ZX-calculus.

 Following \cite{matsumoto_representation_2008}, we use an adaptation of the stabilizer formalism.
 Let:
 \begin{equation}
  M_{(x,y,z)}:=xX+yY+zZ,
 \end{equation}
 where $X,Y,Z$ are the Pauli matrices.
 We say that a single qubit state $\ket{\psi}$ is \emph{stabilized} by $(x,y,z)$ if $M_{(x,y,z)}\ket{\psi}=\ket{\psi}$.
 It is straightforward to show that if $(x,y,z)$ stabilizes $\ket{0}$, then $(x,y,z)=(0,0,1)$.

 Let $S$ be the phase gate, and let $R=\intf{\phase{rn,label={[rphase]right:$\pi/2$}}}$.
 Throughout this proof, we refer to diagrams and their interpretations interchangeably, e.g.\ we say that $\mathcal{V}=\{TR, TSR\}$.
 Now suppose $(x,y,z)$ stabilizes some state $\ket{\psi}$.
 Then for any $C\in\mathcal{C}_1$, $C\ket{\psi}$ is stabilized by some expression of the form $(a\sigma(x),b\sigma(y),c\sigma(z))$, where $\sigma$ is some permutation on the set $\{x,y,z\}$ and $a,b,c\in\{\pm1\}$.
 This is because $C\ket{\psi} = (CM_{(x,y,z)}C^{-1})C\ket{\psi}$ and conjugation by a Clifford operator maps the set of Pauli matrices to itself, up to factors of $\pm 1$.
 Furthermore:
 \begin{mitem}
  \item $T\ket{\psi}$ is stabilized by $\frac{1}{\sqrt{2}}(x-y,x+y,z\sqrt{2})$,
  \item $TR\ket{\psi}$ is stabilized by $\frac{1}{\sqrt{2}}(x+z,x-z,y\sqrt{2})$, and
  \item $TSR\ket{\psi}$ is stabilized by $\frac{1}{\sqrt{2}}(z-x,x+z,y\sqrt{2})$.
 \end{mitem}

 We shall consider the effect of applying a normal form diagram to $\ket{0}$.
 First, consider the case where $W$ is the identity and $n=0$, i.e.\ the diagram is simply of the form $TC$ for some Clifford operator $C$.
 Now $TC\ket{0}$ is stabilized by one of the following expressions:
 \begin{equation}\label{eq:TC_stabilizers}
  \frac{1}{\sqrt{2}}(\pm1,\pm1,0),\quad \frac{1}{\sqrt{2}}(\mp1,\pm1,0), \quad\text{and}\quad (0,0,\pm1).
 \end{equation}
 Even though one of the potential stabilizers is $(0,0,1)$, it is straightforward to check that $TC$ is not a scalar multiple of the identity for any $C$.

 Next consider the possible stabilizers for $V_1TC\ket{0}$, where $V_1\in\mathcal{V}$.
 These are:
 \begin{gather*}
  \frac{1}{2}(\pm1,\pm1,\pm\sqrt{2}),\quad \frac{1}{2}(\mp1,\pm1,\pm\sqrt{2}),\quad \frac{1}{2}(\mp1,\mp1,\pm\sqrt{2}),\quad \frac{1}{2}(\pm1,\mp1,\pm\sqrt{2}),\\
  \frac{1}{\sqrt{2}}(\pm1,\pm1,0),\quad\text{and}\quad \frac{1}{\sqrt{2}}(\mp1,\pm1,0).
 \end{gather*}

 Any stabilizer in the set above can be expressed as:
 \begin{equation}\label{eq:x1x2}
  \frac{1}{\sqrt{2^m}}(x_1+x_2\sqrt{2},y_1+y_2\sqrt{2},z_1+z_2\sqrt{2}),
 \end{equation}
 where $m,x_1,x_2,y_1,y_2,z_1,z_2\in\ZZ$ with $m\geq0$.
 Applying a transformation from $\mathcal{V}$ maps that stabilizer to:
 \begin{equation}
  \frac{1}{\sqrt{2^{m+1}}}\left((x_1+z_1)+(x_2+z_2)\sqrt{2},(x_1-z_1)+(x_2-z_2)\sqrt{2},2y_2+y_1\sqrt{2}\right),
 \end{equation}
 or to:
 \begin{equation}
  \frac{1}{\sqrt{2^{m+1}}}\left((z_1-x_1)+(z_2-x_2)\sqrt{2},(x_1+z_1)+(x_2+z_2)\sqrt{2},2y_2+y_1\sqrt{2}\right).
 \end{equation}
 Note that $\mathcal{W}\subset\mathcal{C}_1$, so the effect of $W\in\mathcal{W}$ is at most a permutation of the numbers $x,y,z$ and the introduction of minus signs.
 Thus the stabilizer of $U\ket{\psi}$ for any normal form operator $U$ can be written in the form \eqref{eq:x1x2}.

 Following \cite{matsumoto_representation_2008}, we consider the parity of $x_1,x_2,y_1,y_2,z_1$ and $z_2$ under the transformations given by repeated application of elements of $\mathcal{V}$.
 For the stabilizers given in \eqref{eq:TC_stabilizers}, we have either $x_1$ and $y_1$ odd and the others even, or $z_1$ odd and the others even.
 For a given $a,b$, the parity of $|a-b|$ is the same as that of $a+b$, so the two transformations in $\mathcal{V}$ have the same effects on the parity of $x_1,x_2,y_1,y_2,z_1$ and $z_2$.
 
 If $x_1$ and $y_1$ are odd and the others even, then after application of some $V\in\mathcal{V}$, $x_1,y_1$, and $z_2$ are odd.
 A second application of $V$ leads to a stabilizer where all factors are odd except for $z_1$.
 A third application of $V$ gives a stabilizer where once again $x_1,y_1$, and $z_2$ are odd.
 Thus the parity of these factors changes cyclically.

 If $z_1$ is odd in the beginning and the other factors are even, then after one application of $V$, $x_1,y_1$ and $z_2$ are odd, after which the same cyclical behaviour appears as above.

 Note that if $WV_n\ldots V_1TC$ is to be a scalar multiple of the identity, then $V_n\ldots V_1TC\ket{0}$ must have a stabilizer in the set $\{(0,0,c),(0,c,0)\}$ for some non-zero $c$, i.e.\ either $x_1=x_2=y_1=y_2=0$ or $x_1=x_2=z_1=z_2=0$.
 In particular, $WV_n\ldots V_1TC$ can only be the identity if $V_n\ldots V_1TC\ket{0}$ has a stabilizer in which either $x_1,x_2,y_1,$ and $y_2$ are all even, or $x_1,x_2,z_1,$ and $z_2$ are all even.
 Yet, as shown above, for any $V_n\ldots V_1TC\ket{0}$, the factor $x_1$ in the stabilizer is always odd.
 Thus $WV_n\ldots V_1TC$ is never the identity, completing the proof.
\end{proof}

\begin{lem}\label{lem:adjoint}
 Consider a normal form diagram $D=WV_n\ldots V_1U$.
 Then $D^\dagger$ is equal to some normal form diagram with the same number of copies of elements of $\mathcal{V}$, i.e.\ $D^\dagger=W'V_n'\ldots V_1'U'$ for some $W'\in\mathcal{W}, V_1',\ldots,V_n'\in\mathcal{V}$ and $U'\in\mathcal{U}$.
\end{lem}
\begin{proof}
 By the properties of the dagger functor, $D^\dagger = U^\dagger V_1^\dagger\ldots V_n^\dagger W^\dagger$.
 Now for any $U\in\mathcal{U}$, we can find $C\in\mathcal{C}_1$ such that:
 \begin{equation}
  \input{tikz_files/U_dagger.tikz},
 \end{equation}
 and for any $V\in\mathcal{V}$, we have:
 \begin{equation}
  \input{tikz_files/V_dagger.tikz}.
 \end{equation}
 Thus by Lemmas \ref{lem:CV} and \ref{lem:pi-commutation}:
 \begin{equation}\label{eq:adjoint}
  \input{tikz_files/WVVU_dagger.tikz}
  \;=\;
  \input{tikz_files/WVVU_dagger-1.tikz}
  =\;
  \input{tikz_files/WVVU_dagger-2.tikz}
  =\;
  \input{tikz_files/WVVU_dagger-3.tikz}
  =\;
  \input{tikz_files/WVVU_dagger-4.tikz}
  \;=\;
  \input{tikz_files/WVVU_dagger-5.tikz}
 \end{equation}
 for some $W'\in\mathcal{W}$, $V_0',\ldots,V_n',V_1'',\ldots,V_n''\in\mathcal{V}$,  $U'\in\mathcal{U}$ and $a,b\in\{0,1\}$.
 Note that $V_1'',\ldots,V_n''$ is just a relabelling of $V_{n-1}',\ldots,V_0'$.
\end{proof}

\begin{thm}\label{thm:Clifford+T_normal_form}
 The normal form for Clifford+$T$ diagrams given in \eqref{eq:UVVW} is unique.
\end{thm}
\begin{proof}
 Suppose there are two normal form diagrams which are equal but not identical.
 Pick a shortest pair of such diagrams, i.e.\ suppose the topmost nodes in the two diagrams have different colours or different phases (or both).
 If the topmost nodes are the same, remove them both and keep going like this until a stage is reached where the remaining topmost nodes are different.
 As the two diagrams are not identical, this must be possible.

 Call these two diagrams $D_1$ and $D_2$.
 As $D_1=D_2$ by assumption, and because any normal form diagram is unitary, it must be the case that $D_1^\dagger\circ D_2$ is equal to the identity.
 We show that under the given assumptions, $D_1^\dagger\circ D_2$ must be equal to some non-trivial normal form diagram.
 By Theorem \ref{thm:not_identity}, this normal form diagram cannot be equal to the identity, thus leading to a contradiction.
 From that we conclude that two normal form diagrams are equal if and only if they are identical.

 Suppose $D_1$ can be written in normal form as $WV_n\ldots V_1U$ and $D_2$ as $W'V_m'\ldots V_1'U'$. 
 The requirement that the topmost nodes of $D_1$ and $D_2$ be different can be satisfied in different ways.
 Where the conditions are not symmetric under interchange of $D_1$ and $D_2$, by Lemma \ref{lem:adjoint} it nevertheless suffices to consider just one of the two options.
 We hence distinguish the following cases:
 \begin{mitem}
  \item $W=W'=I$, $n=m=0$, and the topmost nodes of $U$ and $U'$ differ,
  \item $W=W'=I$, $n=0\neq m$, and the topmost nodes of $U$ and $V_m'$ differ,
  \item $W=W'=I$, $n,m\neq 0$, and $V_n\neq V_m'$,
  \item $W\neq W'$, $n=m=0$,
  \item $W\neq W'$, $n=0\neq m$, and
  \item $W\neq W'$, $n,m\neq 0$.
 \end{mitem}

 Firstly, if $W=W'=I$ and $n=m=0$, then $D_1=U$ and $D_2=U'$ with $U,U'\in\mathcal{U}$.
 Now any element of $\mathcal{U}$ can be expressed as $TC$, for some $C\in\mathcal{C}$.
 Thus $D_1=TC$ and $D_2=TC'$, and as $U\neq U'$ we must have $C\neq C'$.
 Therefore:
 \begin{equation}
  \input{tikz_files/U-U_dagger.tikz}.
 \end{equation}

 Secondly, if $W=W'=I$ and $n=0\neq m$, consider $U$ and $V_m'$.
 Note that $U=TC$ for some Clifford operator $C$, and $V_m'=TC'$ for some Clifford operator $C'$.
 Again, the requirement that the topmost nodes of $U$ and $V_m'$ be different means that $C\neq C'$.
 As in the first case, we thus find $U^\dagger V_m' = C''$ for some $C''$.
 Then by Lemmas \ref{lem:CV} and \ref{lem:pi-commutation}, $D_1^\dagger\circ D_2$ has a normal form $W''V_{m-1}''\ldots V_1''U''$.
 As $m>0$, this is non-trivial.

 The third case, $W=W'=I$, $n,m\neq 0$, and $V_n\neq V_m'$, can be reduced to a case where $W\neq W'$ by applying \phase{gn,label={[gphase]right:$-\pi/4$}} to both diagrams and using the spider rule.

 For $W\neq W'$, we have (after some rewriting):
 \begin{equation}\label{eq:WdaggerW}
  \input{tikz_files/W-W_dagger.tikz}
  \in\left\{
  \input{tikz_files/W-W_dagger-set.tikz}
  \right\}.
 \end{equation}
 Then if $n=m=0$:
 \begin{equation}
  \input{tikz_files/UW-WU_dagger.tikz},
 \end{equation}
 since:
 \begin{equation}
  \input{tikz_files/W-W_dagger-pi.tikz}
  \in\left\{
  \input{tikz_files/W-W_dagger-pi-set.tikz}
  \right\} = \left\{
  \input{tikz_files/W-W_dagger-pi-set-V.tikz}
  \right\}
 \end{equation}
 for some $\alpha,\beta,\gamma\in\{0,\pi/2,\pi,-\pi/2\}$, $a,c\in\{0,1\}$ and $V\in\mathcal{V}$.

 The argument for the case $W\neq W'$ and $n=0\neq m$ is very similar, noting that for any $c\in\{0,1\}$, $\gamma\in\{0,\pi/2,\pi,-\pi/2\}$, and $V\in\mathcal{V}$:
 \begin{equation}
  \input{tikz_files/V-pi-gamma.tikz}
 \end{equation}
 for some $V'\in\mathcal{V}$ and $a,b\in\{0,1\}$.
 Hence by Lemmas \ref{lem:CV} and \ref{lem:pi-commutation}, the diagram can be rewritten into normal form.

 Lastly, consider the case where $W\neq W'$ and $n,m\neq 0$.
 By Lemma \ref{lem:adjoint}, we can rewrite $D_1^\dagger$ to:
 \begin{equation}
  \input{tikz_files/W_dagger-pi-VVW.tikz}.
 \end{equation}
 Now:
 \begin{equation}
  \input{tikz_files/pi_2-pi-V.tikz}
 \end{equation}
 for some $\beta\in\{0,\pi/2,\pi,-\pi/2\}$ and $b\in\{0,1\}$.
 Thus the argument concludes in the same way as in the previous case.

 We have shown that for any pair of normal form diagrams $D_1$ and $D_2$, $D_1^\dagger\circ D_2$ has a non-trivial normal form unless the two diagrams are identical.
 Therefore, by Theorem \ref{thm:not_identity} and by unitarity of Clifford+T operators, two normal form diagrams are equal if and only if they are identical, i.e.\ the normal form is unique.
\end{proof}

The existence of a unique normal form means that any two diagrams representing the same operator can be rewritten into each other since all the rewrite rules are invertible.
Thus Theorem \ref{thm:Clifford+T_normal_form} immediately implies:

\begin{thm}
 The \ZX-calculus with the rewrite rules given in Section \ref{s:rewrite_rules} is complete for the scalar-free single-qubit Clifford+T group.
\end{thm}

Hence any equality between single-qubit Clifford+T diagrams that holds up to a non-zero scalar factor can be derived entirely graphically.

%% file: tikz_files/y_state_der_1.tikz
\begin{tikzpicture}
	\begin{pgfonlayer}{nodelayer}
		\node [style=Hadamard] (0) at (0, -0) {};
		\node [style=none] (1) at (0, 0.5) {};
		\node [label={[gphase]right:$-\pi/2$}, style=gn] (2) at (0, -0.75) {};
	\end{pgfonlayer}
	\begin{pgfonlayer}{edgelayer}
		\draw (1.center) to (2);
	\end{pgfonlayer}
\end{tikzpicture}

%% file: tikz_files/y_state_der_2.tikz
\begin{tikzpicture}
	\begin{pgfonlayer}{nodelayer}
		\node [label={[rphase]right:$\pi/2$}, style=rn] (0) at (0, -0.25) {};
		\node [label={[gphase]right:$\pi/2$}, style=gn] (1) at (0, 1) {};
		\node [style=gn] (2) at (0, -1.5) {};
		\node [style=none] (3) at (0, 1.5) {};
	\end{pgfonlayer}
	\begin{pgfonlayer}{edgelayer}
		\draw (3.center) to (2);
	\end{pgfonlayer}
\end{tikzpicture}

%% file: tikz_files/innerprodrg.tikz
\begin{tikzpicture}
	\begin{pgfonlayer}{nodelayer}
		\node [style=rn] (0) at (0, 0.4) {};
		\node [style=gn] (1) at (0, -0.15) {};
	\end{pgfonlayer}
	\begin{pgfonlayer}{edgelayer}
		\draw (0) to (1);
	\end{pgfonlayer}
\end{tikzpicture}

%% file: Spekkens.tex
\chapter{A complete graphical calculus for Spekkens' toy bit theory}
\label{s:spekkens}

In the previous chapters, we have shown that graphical languages can replace conventional formalisms for several fragments of quantum theory without loss of mathematical rigour.
We now show that similar graphical languages can also be used outside quantum theory.

Toy models are developed in quantum foundations to explore the differences between classical and quantum behaviour.
These are models whose description is entirely classical, but which nevertheless exhibit many properties and effects usually associated with quantum mechanics.

Spekkens' toy theory is one such toy model, which is described in terms of local hidden variables.
The toy bit theory -- the toy theory for the simplest non-trivial system -- is very similar to stabilizer quantum mechanics.
There are also some phenomena that appear in stabilizer quantum theory but are not replicated in the toy model, e.g.\ the violation of Bell inequalities.
Spekkens' toy theory is a $\psi$-epistemic theory by construction, i.e.\ a theory where the state that an observer assigns to a system, is not real: it is only an artefact of the restricted knowledge of the observer.
Quantum theory on the other hand is considered to be $\psi$-ontic, i.e.\ it is a theory where the states an observer assigns to a system are real \cite{pusey_reality_2012}.

Our work builds on the analysis of the toy theory using a stabilizer formalism \cite{pusey_stabilizer_2012}, as well as the categorical formulation of the toy theory \cite{coecke_spekkens_2011,coecke_phase_2011}.
Most of the original results in this chapter were originally published in \cite{backens_complete_2014}.

We first give an introduction to Spekkens' toy theory and its categorical formulation.
Next, we construct a graphical calculus for the toy theory and show that it is universal and sound for the maximal-knowledge fragment of the theory with post-selected measurements -- this corresponds to pure states and post-selected measurements in quantum theory.
Finally, we show that the graphical calculus for the toy theory is complete.

\section{Definition of Spekkens' toy bit theory}
\label{s:spekkens_definition}

Spekkens' toy theory was originally constructed using a \emph{knowledge balance principle} \cite{spekkens_evidence_2007} and has since been reformulated in terms of classical mechanics with restrictions on the knowledge an observer may have of the canonical variables \cite{spekkens_quasi-quantization_2014}.
We use the more recent definition, which has the added advantage of making the similarities between the toy bit theory and stabilizer quantum mechanics more obvious.

The basic ideas behind the toy theory are given in Section \ref{s:spek_basics}.
We then explain the valid states, transformations, and measurements of the theory in detail.
Section \ref{s:toy_theory_categorical} contains the category-theoretical formulation of the toy theory.

\subsection{Basic idea: the principle of classical complementarity}
\label{s:spek_basics}

A single toy bit is a system with four states, these are the \emph{ontic states} or states of reality.
An ontic state can be described by giving the values -- 0 or 1 -- for two variables $Q$ and $P$.

An observer or experimenter working with toy bits does not have direct access to the ontic states, instead they assign to a system an \emph{epistemic state}, a state of knowledge.
The observer can learn about the state of a system by measuring \emph{quadrature variables}, which are linear combinations of the variables $Q$ and $P$; for a single toy bit, these are $Q$, $P$, or $Q\oplus P$, where $\oplus$ denotes addition modulo 2.
As in quantum mechanics, the quadrature variables $Q$ and $P$ for the same toy bit are considered non-commuting: this is done by imposing a commutation relation $[-,-]$ satisfying:
\begin{equation}
 [Q,Q]=0=[P,P] \qquad \text{and} \qquad [Q,P] = 1 = [Q,P],
\end{equation}
which is furthermore linear, so that e.g.:
\begin{equation}
 [Q\oplus P,P] = 1.
\end{equation}
Multiple toy bits can be considered jointly, in which case the variables $Q$ and $P$ for separate subsystems are considered to commute, i.e.:
\begin{equation}
 [Q_i,P_j] = \delta_{ij},
\end{equation}
where the subscripts denote the subsystem to which the variable belongs and $\delta_{ij}$ is 1 if $i=j$ and 0 otherwise.

Now the knowledge an observer may have is determined by the \emph{principle of classical complementarity} \cite{spekkens_quasi-quantization_2014}:

\begin{quotation}
 \noindent The valid epistemic states are those where an agent knows the values of a set of commuting quadrature variables and is ignorant otherwise.
\end{quotation}

Knowing the values of some quadrature variables implies knowledge of the values of any other quadrature variables that arise as linear combinations of the original ones.
Furthermore, if the original variables commute pairwise, then the linear combinations also commute with the original variables and amongst themselves.
Therefore the set of quadrature variables whose values are known to an agent forms a group.

\subsection{Valid states}
\label{s:spek_states}

An epistemic state is fully specified by giving a set of quadrature variables generating the full group of known quadrature variables, and their values.
Epistemic states consist of a probability distribution over multiple ontic states that are compatible with the known quadrature variables.
As an observer has no information other than the quadrature variables, the probability distribution is always the uniform probability distribution with support on all the ontic states compatible with the values of the known quadrature variables.

In the following, the word ``state'' without any qualifiers will be taken to refer to epistemic states.

\emph{States of maximal knowledge} are those epistemic states where the observer knows the values of a maximal set of commuting quadrature variables, i.e.\ a set with which no other quadrature variable commutes.
These states correspond to pure states in quantum theory, and we only consider states of maximal knowledge in this thesis.

Toy theory states for small systems can be visualised as follows: draw the phase space of a single toy bit as four cells arranged in a square, see Figure \ref{fig:toy_bit_state_space} a.
An epistemic state can then be represented by colouring in the boxes corresponding to the allowed ontic states.

\begin{figure}
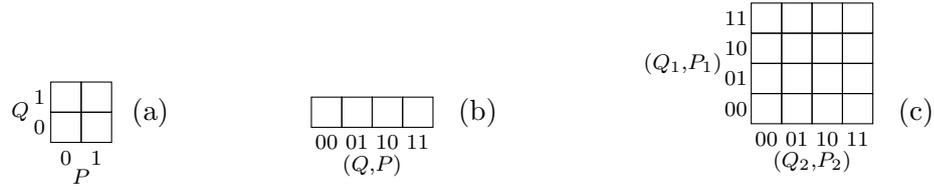

 \centering
 \input{tikz_files/spekkens_state_square.tikz} (a) $\qquad\qquad$
 \input{tikz_files/spekkens_state_line.tikz} (b) $\qquad\qquad$
 \input{tikz_files/spekkens_2state.tikz} (c)
 \caption{(a) \& (b) Visualisations of the state space of a single toy bit. (c) Visualisation of the joint state space of two toy bits. Specific states can be represented by colouring in cells in the diagram.}
 \label{fig:toy_bit_state_space}
\end{figure}

\begin{ex}
 The valid epistemic states of a single toy bit are as follows:
 \begin{equation}\label{eq:single_toy_bit_states}
  \SpekOneSquare{\fill}{\fill}{\draw}{\draw} \, , \quad
  \SpekOneSquare{\draw}{\draw}{\fill}{\fill} \, , \quad
  \SpekOneSquare{\draw}{\fill}{\fill}{\draw} \, , \quad
  \SpekOneSquare{\fill}{\draw}{\draw}{\fill} \, , \quad
  \SpekOneSquare{\fill}{\draw}{\fill}{\draw} \, , \quad
  \SpekOneSquare{\draw}{\fill}{\draw}{\fill} \, , \quad \text{and} \quad
  \SpekOneSquare{\fill}{\fill}{\fill}{\fill} \, .
 \end{equation}
 The first six are the states of maximal knowledge, the last one is a state of less-than-maximal knowledge.
\end{ex}

The four cells denoting the phase space of a single toy bit can also be arranged in a line as in Figure \ref{fig:toy_bit_state_space} b.
This allows the joint state of two toy bits to be visualised on a 4 by 4 grid, cf.\ Figure \ref{fig:toy_bit_state_space} c. 

Joint states of multiple toy bits are \emph{product states} if there exists a generating set for the group of known quadrature observables such that each generator only acts on one subsystem.
States for which there is no such generating set are \emph{correlated}; these correspond to entangled states in quantum theory.

\begin{ex}\label{ex:toy_product_state}
 The epistemic state of two toy bits denoted by:
 \begin{equation}
  \SpekTwoSquare{\Spekfill{-2}{-2} \Spekfill{-2}{0} \Spekfill{-1}{-2} \Spekfill{-1}{0}}
 \end{equation}
 is a product state as it can be expressed in terms of separate conditions on the two toy bits: for example, this state is fully determined by the conditions $P_1=0$ and $Q_2=0$.
 The full group of known quadrature variables for this state is $\{I, P_1, Q_2, P_1\oplus Q_2\}$.
\end{ex}

\begin{ex}\label{ex:toy_correlated_state}
 Consider the epistemic state of two toy bits defined by $Q_1\oplus Q_2=0$ and $P_1\oplus P_2=0$.
 The group of known quadrature variables for this state is:
 \begin{equation}
  \{ I, Q_1\oplus Q_2, P_1\oplus P_2, Q_1\oplus Q_2\oplus P_1\oplus P_2 \}.
 \end{equation}
 This cannot be expressed in terms of separate conditions for the two subsystems, so the state is correlated.
 The correlation is also obvious in the visualisation:
 \begin{equation}
  \SpekTwoSquare{\Spekfill{-2}{-2} \Spekfill{-1}{-1} \Spekfill{0}{0} \Spekfill{1}{1}}
 \end{equation}
\end{ex}

A maximal set of commuting quadrature variables on $n$ toy bits has size $2^n$ and can be specified by giving $n$ independent generators for the corresponding group \cite{spekkens_quasi-quantization_2014}.
Thus a state of maximal knowledge on $n$ toy bits can be specified by giving $n$ commuting quadrature variables, together with their values.

\subsection{Reversible transformations}
\label{s:spek_transformations}

The reversible transformations in the toy theory have to satisfy two conditions:
\begin{itemize}
 \item They arise from reversible transformations of the ontic states, and
 \item they map valid epistemic states to valid epistemic states.
\end{itemize}
Reversible transformations can be represented in the visualisation introduced in the previous section in different ways.
One option is to use arrows to show the transformation of the ontic states.

\begin{ex}
 The transformation:
 \begin{center}
  \input{tikz_files/spekkens_transf_23.tikz}
 \end{center}
 swaps the ontic states $(0,1)$ and $(1,0)$ and keeps the other two ontic states invariant.
 The effect on the epistemic states follows from the effect on the underlying ontic states.
 E.g., the following state is mapped to itself:
 \begin{equation}
  \SpekOneSquare{\draw}{\fill}{\fill}{\draw} \; \mapsto \; \SpekOneSquare{\draw}{\fill}{\fill}{\draw} \; ,
 \end{equation}
 whereas this epistemic state changes:
 \begin{equation}
  \SpekOneSquare{\fill}{\fill}{\draw}{\draw} \; \mapsto \; \SpekOneSquare{\fill}{\draw}{\fill}{\draw} \; .
 \end{equation}
\end{ex}

The set of valid reversible transformations of a single toy bit consists of all 24 permutations of the four ontic states.
For more complicated systems, not all permutations of the ontic states yield valid transformations of the epistemic states.

A reversible transformation is called \emph{local} if it is a product of permutations of the states of each individual system.

\subsection{Valid measurements}
\label{s:spek_measurements}

Any quadrature variable makes a valid measurement observable in the toy theory.
A commuting set of quadrature variables can be measured at the same time.
By the principle of classical complementarity, after the measurement of any set of quadrature variables, the values of any quadrature variables for which there exists a measured variable with which they do not commute are completely unknown.
Thus measurements in the toy theory potentially change the ontic state of the system.

\begin{ex}
 Assume an agent, who knows that a toy bit is in the state:
 \begin{center}
  \SpekOneSquare{\fill}{\draw}{\draw}{\fill}
 \end{center}
 measures the observable $P$ and gets outcome 1.
 This means the system is left in the state:
 \begin{center}
  \SpekOneSquare{\draw}{\fill}{\draw}{\fill} .
 \end{center}
 It is possible to retrodict that the system must have been in the ontic state:
 \begin{center}
  \SpekOneSquare{\draw}{\draw}{\draw}{\fill}
 \end{center}
 originally for this measurement outcome to occur, but the measurement scrambles the ontic states so that the system may no longer be in that state afterwards.
\end{ex}

Any decomposition of the phase space into a set of valid epistemic states corresponds to a valid measurement.
Correspondingly, any valid epistemic state is a possible outcome of some valid measurement.
Measurements on multiple systems can be correlated or not, like states.

\subsection{The categorical formulation of the toy theory}
\label{s:toy_theory_categorical}

Number the ontic states of the toy theory 1 through 4 in the order in which they appear in Figure \ref{fig:toy_bit_state_space} b.
Epistemic states can then be denoted by sets.

\begin{ex}
 The first state in \eqref{eq:single_toy_bit_states}, $Q=0$, corresponds to the set $\{1,2\}$.
 The product state of two toy bits from Example \ref{ex:toy_product_state} is represented by the set:
 \begin{equation}
  \{(1,1),(1,2),(3,1),(3,2)\}.
 \end{equation}
 Alternatively, the same state can be written as:
 \begin{equation}
  \{1,3\}\times\{1,2\},
 \end{equation}
 emphasising the fact that it is a product state.
\end{ex}

Rather than considering states of $n$ toy bits to be sets, they can also be seen as relations from the one-element set $I=\{\bullet\}$ into $IV^n$, the $n$-fold Cartesian product of the set $IV=\{1,2,3,4\}$.

\begin{ex}
 The state $Q=0$ corresponds to the relation:
 \begin{equation}
  \bullet\sim\{1,2\}.
 \end{equation}
\end{ex}

Post-selected measurements on $n$ toy bits can be seen as relations from $IV^n$ to $I$, e.g.\ the single-toy bit measurement of the $P$ variable with outcome 1 corresponds to:
\begin{equation}
 \{2,4\}\sim\bullet.
\end{equation}
Reversible transformations can also be considered as relations.
This perspective puts state preparation and post-selected measurements on equal footing with reversible transformations and allows any process on toy bits to be considered as a relation.

The allowed processes in the maximal-knowledge fragment of the toy theory with post-selected measurements form a dagger compact closed category called $\Spek$, which is a subcategory of $\FRel$ \cite{coecke_spekkens_2011} (cf.\ also Section \ref{s:category_theory}).

\begin{dfn}\label{dfn:Spek}
 The category $\Spek$ is composed of the following: The objects of $\Spek$ are the one-element set $I=\{\bullet\}$, the four-element set $IV=\{1,2,3,4\}$, and its $n$-fold Cartesian products with itself, denoted $IV^n$.
 The arrows of $\Spek$ are generated by parallel composition, sequential composition, and dagger, from the following basic relations \cite{edwards_non-locality_2009}:
\begin{itemize}
 \item the 24 permutations $IV\to IV$,
 \item the map $\delta:IV\to IV^2$ defined as:
  \begin{equation}
   \delta = \begin{cases} 1\sim\{(1,1),(2,2)\} \\ 2\sim\{(1,2),(2,1)\} \\ 3\sim\{(3,3),(4,4)\} \\ 4\sim\{(3,4),(4,3)\}, \text{ and} \end{cases}
  \end{equation}
 \item the measurement effect $\epsilon:IV\to I$ defined as:
  \begin{equation}
   \epsilon = \{1,3\}\sim\bullet,
  \end{equation}
\end{itemize}
\end{dfn}

Parallel composition, sequential composition, and dagger are defined as in $\FRel$, see Section \ref{s:category_theory}.

The scalars in this formulation of the toy theory are the relations from $I\to I$.
There are just two of those: the identity scalar $\{(\bullet,\bullet)\}$, and the empty relation $\emptyset$.
This means that the categorical formulation of the toy theory is possibilistic, i.e.\ it does not allow the computation of probabilities but simply shows whether an outcome is possible (the identity scalar) or not (the empty relation).

\section{A graphical calculus for the toy theory}
\label{s:spekkens_graphical}

We have introduced the toy bit theory in the previous section and given the corresponding category.
Building up on that work, we now construct a graphical calculus for the toy theory, which is closely analogous to the scalar-free stabilizer \ZX-calculus.
In particular, we use the same notation in terms of red and green spiders for the toy theory.
It should always be clear from context whether a specific diagram is part of the \ZX-calculus or the toy theory graphical calculus.

We first define basic elements of the graphical notation and show how to combine them into more complicated diagrams.
Next, we give rewrite rules for those diagrams.
In Sections \ref{s:spekkens_universal} and \ref{s:spekkens_sound}, we argue that the calculus is universal and sound for Spekkens' toy bit theory.
Finally, we compare the graphical calculus for the toy theory to the \ZX-calculus.

\subsection{Components and their interpretations}
\label{s:spekkens_components}

Like the \ZX-calculus, the graphical calculus for the toy theory is read from bottom to top and most maps are denoted by circular nodes, which may have labels attached.
As before, we use $\intf{D}$ to denote the process corresponding to a diagram $D$.
 
Define \splitnode{} to be the following map from one toy bit to two toy bits:
\begin{equation}
 \intf{\splitnode{}} := \begin{cases}
    1 \sim \{(1,1),(2,2)\} \\
    2 \sim \{(1,2),(2,1)\} \\
    3 \sim \{(3,3),(4,4)\} \\
    4 \sim \{(3,4),(4,3)\} .
   \end{cases}
\end{equation}
This is a valid process in the toy theory; it can be considered to consist of the preparation of an ancilla in some fixed state followed by some joint reversible operation on the original toy bit and the ancilla.

Let \joinnode{} be the relational converse of \splitnode{}:
\begin{equation}
 \intf{\joinnode{}} := \intf{\splitnode{}}^\dagger.
\end{equation}
This is also a valid process in the toy theory, which can be thought of as a reversible operation on two toy bits, followed by a post-selected measurement of one of them.

More complicated diagrams in the toy theory graphical calculus can be built by putting smaller diagrams side-by-side, which corresponds to taking the Cartesian product of the corresponding relations; i.e.\ if:
\begin{center}
 \gendiagram{$D$} $\quad$ and $\quad$ \gendiagram{$D'$}
\end{center}
denote two arbitrary diagrams, then:
\begin{equation}
 \intf{ \gendiagram{$D$} \; \gendiagram{$D'$} } = \intf{ \gendiagram{$D$} } \times \intf{ \gendiagram{$D'$} }.
\end{equation}
Connecting the inputs of some diagram to the outputs of another corresponds to the operation of relational composition.
Graphically, assuming the number of outputs of $D$ is equal to the number of inputs of $D'$:
\begin{equation}
 \intf{ \input{tikz_files/composite_diagram.tikz} } = \intf{ \gendiagram{$D'$} } \circ \intf{ \gendiagram{$D$} }.
\end{equation}

Motivated by the \ZX-calculus, we introduce \emph{spiders} as a short-hand notation for specific diagrams built from \splitnode{} and \joinnode{}: a green node with $n$ inputs and $m$ outputs for positive integers $n,m$ is defined as follows:
\begin{equation}\label{eq:spider}
 \intf{ \input{tikz_files/spider_single.tikz} }  :=
 \intf{ \input{tikz_files/spider_no_phase.tikz} }.
\end{equation}

Represent the following four single-toy bit states by green nodes with phase labels:
\begin{align}
  \intf{\state{gn,label={[gphase]right:$00$}}} &:= \bullet\sim\{1,3\}, \label{eq:state00} \\
  \intf{\state{gn,label={[gphase]right:$01$}}} &:= \bullet\sim\{1,4\}, \\
  \intf{\state{gn,label={[gphase]right:$10$}}} &:= \bullet\sim\{2,3\}, \quad\text{and} \\
  \intf{\state{gn,label={[gphase]right:$11$}}} &:= \bullet\sim\{2,4\}, \label{eq:state11}
\end{align}
and let $\state{gn}$ be short-hand for \state{gn,label={[gphase]right:$00$}}.
These two alternative notations make later definitions consistent.
Spiders can now be given phase labels via the following definition:
\begin{equation}\label{eq:phased_spider}
 \intf{\input{tikz_files/spider_def_Spek_LHS.tikz}} \; := \; \intf{\input{tikz_files/spider_def_Spek_RHS.tikz}}
\end{equation}
where $x,y\in\{0,1\}$.
Furthermore, spiders without inputs can be defined by composing \state{gn} and a spider with one input.
Let \effect{gn} be the converse of \state{gn}, seen as a relation:
\begin{equation}
 \intf{\effect{gn}} := \begin{cases} 1\sim\bullet \\ 3\sim\bullet. \end{cases}
\end{equation}
Then arbitrary spiders with no outputs can be defined as composites of a one-output spider and \effect{gn}.
In this way, definitions \eqref{eq:spider} and \eqref{eq:phased_spider} can be extended to arbitrary non-negative numbers of inputs and outputs $n$ and $m$.

Let \HadSpek{} be the following reversible single-toy bit operation:
\begin{equation}
  \intf{\HadSpek} := \begin{cases}
    1 \sim 1 \\
    2 \sim 3 \\
    3 \sim 2 \\
    4 \sim 4.
   \end{cases}
 \end{equation}
As a final short-hand, define red spiders as green spiders with copies of \HadSpek{} on all inputs and outputs:
 \begin{equation}
  \left\llbracket \input{tikz_files/spek_red_spider.tikz} \right\rrbracket := \left\llbracket \input{tikz_files/spek_colour_rhs.tikz} \right\rrbracket .
 \end{equation}

A single straight wire corresponds to the identity relation and a wire crossing is the obvious \SWAP{} relation interchanging the states of the two subsystems.
A ``cup'' is interpreted as follows:
\begin{equation}
 \intf{\input{tikz_files/cup_diagram.tikz}} := \bullet\sim\{(1,1),(2,2),(3,3),(4,4)\},
\end{equation}
and the cap is its converse.

As \splitnode{}, \joinnode{}, \state{gn, label={[gphase]right:$xy$}}, and \effect{gn} are all special cases of green spiders, the graphical calculus can be considered to consist of green phased spiders with $n$ inputs and $m$ outputs, where $n$ and $m$ are now non-negative integers; red phased spiders with arbitrary numbers of inputs and outputs; and \HadSpek{}.

In the following, we re-use most of the \ZX-calculus terminology introduced in Section \ref{s:ZX_terminology}, with the term ``zero diagram'' now referring to diagrams representing the empty relation.

\subsection{Rewrite rules}
\label{s:spekkens_rules}

We postulate the following rewrite rules for the toy theory graphical calculus.
Any rule given here can also be used with the colours red and green swapped.
Rules can furthermore be used upside-down.
In the following, $n,m,k,l$ are non-negative integers, $a,b,c,d\in\{0,1\}$, and addition is modulo 2:
 \begin{itemize}
  \item the spider rule:
   \begin{equation}
    \input{tikz_files/spek_spider.tikz},
   \end{equation}
  \item the loop rule:
   \begin{equation}
    \input{tikz_files/spek_loop.tikz},
   \end{equation}
  \item the cup rule:
   \begin{equation}
    \input{tikz_files/cup.tikz},
   \end{equation}
  \item the bialgebra rule:
   \begin{equation}
    \input{tikz_files/bialgebra_no_scalars.tikz},
   \end{equation}
  \item the copy rule:
   \begin{equation}
    \input{tikz_files/copy_no_scalars.tikz}\;,
   \end{equation}
  \item the 11-copy rule:
   \begin{equation}
    \input{tikz_files/spek_11-copy.tikz},
   \end{equation}
  \item the 11-commutation rule:
   \begin{equation}
    \input{tikz_files/spek_11-commutation.tikz},
   \end{equation}
  \item the colour change rule:
   \begin{equation}
    \input{tikz_files/spek_red_spider.tikz} \; = \; \input{tikz_files/spek_colour_rhs.tikz},
   \end{equation}
  \item the ``Euler decomposition'' rule:
   \begin{equation}
    \HadSpek \; = \; \input{tikz_files/spek_Euler_dec.tikz},
   \end{equation}
  \item the scalar rule:
   \begin{equation}
    \innerprod{gn,label={[gphase]right:$ab$}}{rn,label={[rphase]right:$cd$}} = \begin{cases}  \scalar{gn,label={[gphase]right:$11$}} &\text{if } a=d\neq b=c, \text{ and} \\ \quad &\text{otherwise}, \end{cases}
   \end{equation}
  \item and the zero rule:
   \begin{equation}
    \input{tikz_files/spek_zero.tikz}\;.
   \end{equation}
 \end{itemize}

As in the \ZX-calculus, whenever a rule holds for any number of edges, that number may be zero.
Furthermore, there is a meta rule: ``only the topology matters'', i.e.\ two diagrams represent the same process whenever they contain the same set of nodes connected up in the same ways, no matter how those nodes are arranged on the plane.

\subsection{Universality}
\label{s:spekkens_universal}

The graphical calculus for the toy theory as defined in the previous two subsections is universal for the maximal knowledge fragment of Spekkens' toy bit theory with post-selected measurements.
This follows from the category-theoretical formulation of the toy theory by Coecke et al.\ \cite{coecke_spekkens_2011} (cf.\ Section \ref{s:toy_theory_categorical}), where it is shown that all processes in the toy theory arise -- via parallel and sequential composition, and taking the relational converse -- from the 24 reversible transformations of a single toy bit together with a map $\delta$ from one toy bit to two toy bits, and a post-selected measurement outcome $\epsilon$.
It is straightforward to see that \HadSpek{} and the phase shifts suffice to construct all 24 reversible single-toy bit transformations, which correspond to the 24 permutations of the ontic states.
The maps $\delta$ and $\epsilon$ from Definition \ref{dfn:Spek} are exactly the maps denoted by \splitnode{} and \effect{gn}.
The graphical calculus allows parallel and sequential composition, as well as the taking of relational converses, which corresponds to flipping diagrams upside-down.
Therefore any process in the maximal knowledge fragment of Spekkens' toy bit theory with post-selected measurements can be represented graphically.

\subsection{Soundness}
\label{s:spekkens_sound}

Most of the rewrite rules of the toy theory graphical calculus can straightforwardly be checked to be sound by translating the diagrams on both sides of the equality into the corresponding maps and possibly using induction over the number of inputs and/or outputs, cf.\ the \ZX-calculus soundness argument in Section \ref{s:ZX_sound}.

For soundness of the spider rule, we again rely on results from the categorical formulation of Spekkens' toy bit theory.
As shown in \cite{coecke_phase_2011}, the maps \splitnode{} and \effect{gn} form a category-theoretical \emph{observable}.
This means that any connected diagram constructed from these maps, their converses, wire crossings, and curved wires is determined solely by its number of inputs and outputs.
Graphically, this corresponds exactly to the spider law without phase labels \cite{coecke_povms_2008}.
The states \state{gn, label={[gphase]right:$xy$}} form a \emph{phase group} for this observable \cite{coecke_phase_2011}.
In particular, they form a group under the operation given by composition with \joinnode{}:
\begin{equation}
 \input{tikz_files/phase_group_operation.tikz},
\end{equation}
with group identity \state{gn} and all group elements being self-inverse.
From this, it follows that the spider law with phase labels is also sound.
Equivalently, the phase group can be considered to consist of green phase shifts under sequential composition, with \phase{gn} as the group identity (cf.\ Section \ref{s:stabilizer_ZX}).
The phase group is isomorphic to $\ZZ_2\times\ZZ_2$ \cite{coecke_phase_2011}.

Soundness of the topology meta-rule also follows from the category-theoretical formulation of the toy theory: The toy theory is modelled as a dagger compact closed category, therefore the results given in Section \ref{s:graphical_languages_algebraic} apply.

\subsection{The toy theory graphical calculus and the \ZX-calculus}
\label{s:spekkens_ZX}

Category-theoretically, the only difference between the toy theory and scalar-free stabilizer quantum theory is the phase group of the respective observables: for the toy theory the phase group is isomorphic to the Klein Four group $\ZZ_2\times\ZZ_2$, whereas for stabilizer quantum theory the phase group is isomorphic to the cyclic group of order four, $\ZZ_4$ \cite{coecke_phase_2011}.

Correspondingly, the rewrite rules of the toy theory graphical calculus that do not involve specific phases are exactly the same as those of the \ZX-calculus if the phase groups are swapped out.
The phase shift \phase{gn,label={[gphase]right:$11$}} takes the place of the $\pi$-phase shift in the \ZX-calculus in that it is copied by spiders of the other colour, interacts interestingly with phase shifts of the other colour, and yields the zero scalar when sandwiched between \state{gn} and \effect{gn}.

The \ZX-calculus analogue to the $11$-commutation rule is the scalar-free $\pi$-commutation rule:
\begin{equation}
 \input{tikz_files/pi-comm2.0.tikz},
\end{equation}
where $\alpha\in\{-\pi/2,0,\pi/2,\pi\}$.
At first glance this looks different to the $11$-commutation rule: the $\pi$-commutation rule sends any phase shift to its inverse whereas the $11$-commutation rule swaps the two bits denoting the phase.
In fact, both commutation rules can be expressed in the same way nevertheless.
Let $\varphi$ denote $11$ or $\pi$ and let $\theta$ be an arbitrary phase label for the respective theory.
We can write the two commutation rules in general form as:
\begin{equation}
 \input{tikz_files/spek_commutation_general.tikz},
\end{equation}
where $f$ is some map from the phase group to itself.
Then in both the \ZX-calculus for stabilizer quantum mechanics and in the graphical calculus for Spekkens' toy bit theory, the map $f$ can be characterised as follows: $f$ maps both $\varphi$ and the identity of the phase group back to themselves, but it swaps the remaining two elements of the phase group.

\section{Completeness of the toy theory graphical calculus}
\label{s:spekkens_completeness}

We now show that the toy theory graphical calculus is complete by adapting the completeness proof for the scalar-free stabilizer \ZX-calculus given in Chapter \ref{ch:completeness}.
There are several parts to the argument: First, we show that the results characterising all equalities between stabilizer states from \cite{van_den_nest_graphical_2004}, which are central to the \ZX-calculus completeness proof, also hold in Spekkens' toy theory.
We then argue that it is sufficient to consider equalities between toy states rather than more general processes in the toy theory because the toy theory has map-state duality.
Next, we prove that diagrams in the toy theory graphical calculus can be brought into a normal form called GS-LO form.
Finally, we show that the rewriting strategies used in the \ZX-calculus completeness proof also work in the toy theory graphical calculus.

Where the steps in the completeness argument for the toy theory differ only marginally from the corresponding steps in the stabilizer \ZX-calculus completeness proof, the proofs are left out or given in sketch form.

\subsection{A binary formalism and graph state theorems for the toy theory}
\label{s:binary}

Completeness of a graphical language means that any equality that can be derived in the standard formalism for the same theory can also be derived graphically.
Thus it is useful to have some simple way of characterising the equalities that can be derived in the underlying theory.

The completeness proof for the stabilizer \ZX-calculus makes use of two theorems about relationships between stabilizer states under local Clifford unitaries, i.e.\ unitary stabilizer operations that are tensor products of single-qubit unitaries.
Those theorems, proved by Van den Nest et al.\ \cite{van_den_nest_graphical_2004} are given here as Theorems \ref{thm:stabilizer_graph_state} and \ref{thm:graph_states_LC}.
We now show that these results translate to the toy theory.

As described in Section \ref{s:spekkens_definition}, a state of maximal knowledge on $n$ toy bits is given by a set of $n$ commuting quadrature variables, together with the values for each of the variables.
These quadrature variables can be represented as binary vectors, similar to the representation of Pauli products, where the $m$-th and $(m+n)$-th component together encode the quadrature variable acting on the $m$-th toy bit according to the following encoding:
\begin{align}
 Q &\mapsto 01, \\
 P &\mapsto 10, \text{ and} \\
 Q\oplus P &\mapsto 11,
\end{align}
with 00 indicating that no quadrature variable is acting on the given toy bit.
Thus, ignoring the values of the quadrature variables, any state of maximal knowledge can be described by a binary $2n$ by $n$ matrix in the same way as a pure quantum state, cf. Section \ref{s:binary_stabilizer_formalism}.

\begin{ex}
 The toy state from Example \ref{ex:toy_product_state}, $P_1=0\wedge Q_2=0$, corresponds to the check matrix:
 \begin{equation}
  \left( \begin{array}{cc} 1 & 0 \\ 0 & 0 \\ \hline 0 & 0 \\ 0 & 1 \end{array} \right),
 \end{equation}
 where the first column represents $P_1$ and the second column $Q_2$.
 Similarly, the state from Example \ref{ex:toy_correlated_state}, $Q_1\oplus Q_2=0\wedge P_1\oplus P_2=0$, can be represented by the check matrix:
 \begin{equation}
  \left( \begin{array}{cc} 0 & 1 \\ 0 & 1 \\ \hline 1 & 0 \\ 1 & 0 \end{array} \right).
 \end{equation}
 The values of the quadrature variables are not represented in the check matrix.
\end{ex}

In fact, the binary $2n$ by $n$ matrices representing valid epistemic states of the toy theory are exactly the same as the ones representing valid stabilizer states.

\begin{lem}
 A binary $2n$ by $n$ matrix $S$ represents a valid state in the toy theory if and only if $S^T J S=0$, where:
 \begin{equation}
  J = \begin{pmatrix}0&I\\I&0\end{pmatrix}
 \end{equation}
 with $I$ the $n$ by $n$ identity matrix.
 The valid reversible transformations of the toy theory are represented in the binary picture by $2n$ by $2n$ binary matrices $Q$ satisfying $Q^T J Q = J$.
\end{lem}

This follows from the principle of classical complementarity, as shown in \cite{spekkens_quasi-quantization_2014}.

The conditions for $2n$ by $n$ binary matrices to represent valid states and the condition for $2n$ by $2n$ binary matrices to represent valid transformations are exactly the same as in the binary formalism for stabilizer quantum mechanics, cf. Lemmas \ref{lem:check_matrix} and \ref{lem:binary_operation}.
Therefore the binary matrix formalism for Spekkens' toy bit theory is exactly the same as the check matrix formalism for stabilizer quantum mechanics, if one ignores the values of the quadrature variables in the former and the eigenvalues in the latter.
An equivalent result was shown in \cite{pusey_stabilizer_2012}, albeit not using check matrices.

This equivalence can be used to define graph states for the toy theory.

\begin{dfn}\label{dfn:toy_graph_check_matrix}
 A $n$-toy bit \emph{graph state in Spekkens' toy bit theory} is a state having the same check matrix representation as some $n$-qubit graph state in stabilizer quantum mechanics; i.e.\ there exists a $n$ by $n$ symmetric binary matrix $\theta$ with zeroes along the diagonal such that:
 \begin{equation}
  \left( \begin{array}{cc} \theta \\ \hline I \end{array} \right)
 \end{equation}
 is a check matrix for the state.
\end{dfn}

The values of the quadrature variables can be ignored when describing toy states by check matrices because each value can be changed by a local reversible toy theory operation that leaves all other values invariant.
This is analogous to the case of eigenvalues in quantum theory, cf.\ Section \ref{s:binary_stabilizer_formalism}.

Theorems \ref{thm:stabilizer_graph_state} and \ref{thm:graph_states_LC} are proved entirely within the binary formalism.
We have shown that the binary formalism for the toy theory is exactly the same as that for stabilizer quantum theory, and we have defined graph states for the toy theory which are analogous to those in quantum theory.
Therefore these theorems carry over to the toy theory, i.e.\ we have:

\begin{thm}\label{thm:vdN1_toy}
 Any toy stabilizer state is equivalent to some toy graph state under local toy transformations $\sigma\in (S_4)^n$.
\end{thm}

\begin{thm}\label{thm:vdN2_toy}
 Two toy graph states on the same number of toy bits are equivalent under local toy transformations if and only if there is a sequence of local complementations (cf.\ Definition \ref{dfn:local_complementation}) that transform one graph into other.
\end{thm}

These two \emph{graph state theorems} allow all possible equalities between toy theory states to be characterised in a straightforward way.

\subsection{Map-state duality for the toy theory}
\label{s:toy_Choi}

The graph state theorems, as implied by the name, apply only to toy states, not to more general processes in the toy theory.
Yet it suffices to consider only equalities between states in order to get a completeness result for the entire theory.
This is because, like quantum theory, the toy theory exhibits map-state duality, also called the Choi-Jamio{\l}kowski isomorphism.

\begin{thm}
 For any pair of positive integers $n$ and $m$, there exists a bijection between the toy theory operators from $n$ to $m$ toy bits and the states on $n+m$ toy bits.
\end{thm}

Diagrammatically, this duality is represented in \eqref{eq:Choi-Jamiolkowski}; as in the \ZX-calculus, the toy diagram equality follows from the topology rule.

The Choi-Jamio{\l}kowski isomorphism allows toy theory operators to be turned into states.
Any equalities derived between these states then apply also to the original operators.
Thus, a completeness result for the entire toy theory can be derived by considering only toy states.

\subsection{Graph states and related diagrams in the toy theory graphical calculus}
\label{s:toy_diagrams}

In the first part of this section, we defined graph states for the toy theory via their check matrices.
We now show that, like their quantum equivalents in the \ZX-calculus, they also have an elegant graphical representation.

\begin{dfn}\label{dfn:toy_graph_state}
 Let $G$ be a finite simple undirected graph, i.e.\ a graph with finitely many vertices, at most one edge between any pair of vertices, and no self-loops.
 Let the set of vertices be $V$ and the set of edges $E$.
 The associated graph state in the toy theory graphical calculus comprises the following:
 \begin{mitem}
  \item for each vertex in $V$, a green node with one output, and
  \item for each edge in $E$, a copy of \HadSpek{} connected to the green nodes representing the vertices at either end of the edge.
 \end{mitem}
\end{dfn}

To show that this definition is equivalent to Definition \ref{dfn:toy_graph_check_matrix}, we consider the operators stabilizing the graph state.

\begin{lem}\label{lem:fixpoint}
 Let \gengraph{$G$} denote the toy bit theory state associated with a graph $G=(V,E)$. Then for any vertex $v\in V$:
 \begin{equation}
  \gengraph{$G$} \; = \; \input{tikz_files/spek_fp.tikz}
 \end{equation}
 where $a_k=1$ if $\{v,k\}\in E$ and $a_k=0$ otherwise.
 This means that \gengraph{$G$} is an eigenstate of any operator that applies \phase{rn,label={[rphase]right:$11$}} to one of the vertices and \phase{gn,label={[gphase]right:$11$}} to all neighbours of that vertex.
\end{lem}

This lemma follows from the rules of the red-green calculus for the toy theory; the proof is entirely analogous to that for the \ZX-calculus \cite{duncan_graph_2009}.

\begin{cor}
 Definition \ref{dfn:toy_graph_state} is equivalent to Definition \ref{dfn:toy_graph_check_matrix}.
\end{cor}
\begin{proof}
 Lemma \ref{lem:fixpoint} gives $n$ quadrature variables of which the diagram is an eigenstate: a phase shift \phase{rn, label={[rphase]right:$11$}} on the $k$-th output corresponds to a term $Q_k$ in the quadrature variable, a phase shift \phase{gn, label={[gphase]right:$11$}} on the $l$-th output corresponds to a term $P_l$.
 Translate each of these variables into a binary vector as described in Section \ref{s:binary}, and assemble the resulting vectors as the columns of a matrix with the vector for the variable involving the term $Q_m$ as the $m$-th column.
 The resulting check matrix then has the form required by Definition \ref{dfn:toy_graph_check_matrix}.

 Conversely, take a check matrix of the form given in Definition \ref{dfn:toy_graph_check_matrix}.
 Let $\theta$ be the upper square of that check matrix and define $G$ to be the graph with adjacency matrix $\theta$.
 Then it is straightforward to see that the diagram constructed for this $G$ according to Definition \ref{dfn:toy_graph_state} represents the state determined by the check matrix.
\end{proof}

The local complementation operations from Theorem \ref{thm:vdN2_toy} can be derived from the rules of the toy theory graphical calculus.
From now on, we use the term ``local complementation'' to refer to an operation on graph states together with the application of a local operation to all the toy bits that keeps the overall toy state invariant. 

\begin{lem}
\label{lem:local_complementation}
 The following local complementation rewrite rule holds in the red-green calculus for the toy theory:
 \begin{equation}
  \gengraph{$G$} \; = \; \input{tikz_files/spek_lc.tikz}
 \end{equation}
 where $a_k=1$ if $\{v,k\}\in E$ and $a_k=0$ otherwise, and $G\star v$ denotes the graph-theoretical local complementation as defined in \eqref{eq:graph_local_complementation}.
\end{lem}
\begin{proof}[Proof (sketch)]
 The proof is analogous to the \ZX-calculus case as given by Duncan and Perdrix \cite{duncan_graph_2009}.
 We show here as an example the case of the complete graph on three vertices (rearranged with two inputs at the bottom for ease of reading):
 \begin{equation}
  \input{tikz_files/spek_local_complementation_derivation1.tikz}\!\!\!\input{tikz_files/spek_local_complementation_derivation2.tikz}
 \end{equation}
 The first equality uses the decomposition of \HadSpek{} in terms of red and green phase shifts.
 In the second step, the spider rule is used to ``push'' the green phase shifts through their green neighbours.
 At the same time, the colour change law and the fact that \HadSpek{} is self-inverse are used to change the green node at the top into a red one.
 The next step is an application of the bialgebra law.
 The penultimate step uses the fact that \state{rn,label={[rphase]right:$01$}} = \state{gn,label={[gphase]right:$01$}}, which is the case $a=0$ of \eqref{eq:spek_red_green_states} below, followed by the spider law.
 Lastly, the colour change rule is applied again.

 The full proof then proceeds by induction over the number of vertices in the graph state.
\end{proof}

\begin{remark}
 We can now define a toy-theory version of the \emph{local complementation along an edge} by applying three local complementations to a pair of toy bits $v, w \in V$ where $\{v,w\}\in E$, yielding:
 \begin{equation}
  \input{tikz_files/spek_pivoting.tikz}
 \end{equation}
 Here:
 \begin{equation}
  \sigma_j' = \begin{cases} \sigma_j \circ (23) & \text{if }j\in\{v,w\} \\ \sigma_j & \text{otherwise,} \end{cases}
 \end{equation}
 where $(23)$ denotes the transposition of 2 and 3.
 The graph $G' = (V,E')$ satisfies the same properties as in the stabilizer \ZX-calculus equivalent, see Section \ref{s:equivalence_GS-LC}.
\end{remark}

It will be useful to have a normal form for reversible single-toy bit operators.

\begin{lem}\label{lem:single_toy_bit_operator}
 Any single-toy bit operator, i.e.\ any diagram or subdiagram consisting solely of phase shifts and \HadSpek{} can be written uniquely in one of the following forms:
 \begin{equation}
  \input{tikz_files/spek_single_toybit.tikz},
 \end{equation}
 where $a,b,c,d,e,f,g\in\{0,1\}$ and $\bar{e}=e\oplus 1$.
\end{lem}

This is straightforward to check, analogously to the corresponding result in the \ZX-calculus.
In the following, whenever we talk about reversible single-toy bit operators we assume that they are normalised as in the above lemma.

\begin{dfn}
 A diagram in the red-green calculus for Spekkens' toy theory is called a \emph{GS-LO diagram} (graph state with local operators) if it consists of a graph state as in Definition \ref{dfn:toy_graph_state} with single-toy bit operators on each output.
\end{dfn}

This is analogous to the definition of GS-LC diagrams in the \ZX-calculus.
GS-LO diagrams play a central role in the graphical calculus for the toy theory, as shown by the following theorem.

\begin{thm}\label{thm:GS_LO}
 Any state diagram in the red-green calculus for Spekkens' toy theory is equal to some GS-LO diagram, possibly composed with \scalar{gn, label={[gphase]right:$11$}}, according to the rewrite rules.
\end{thm}
\begin{proof}[Proof (sketch)]
 Consider the scalar part and the non-scalar part of the diagram separately.

 The proof that the non-scalar part of a state diagram in the toy theory graphical calculus can be brought into GS-LO form is analogous to the proof of Theorem \ref{thm:ZX_GS-LC} for the \ZX-calculus and its constituent lemmas, noting the following facts:
 \begin{itemize}
  \item Let $a\in\{0,1\}$ and $\bar{a}=a\oplus 1$, then:
   \begin{equation}\label{eq:spek_red_green_states}
    \input{tikz_files/spek_single_toy_state.tikz}
   \end{equation}
   Here, the first step uses the fact that \HadSpek{} is self-inverse and the second step uses the decomposition of \HadSpek{} into red and green phase shifts.
   The third step is an application of the spider law to merge the bottom two nodes, which is again used in the fourth step to pull apart the green node.
   In the fifth step, the bottom red node is copied: this works for both values of $a$.
   The penultimate step, involves dropping the scalar diagram on the left and merging the two red nodes in the non-scalar part by the spider law.
   The last equality is by the colour change law.
  \item Any scalar subdiagram appearing during the rewrite process consists of at most two nodes.
  Subdiagrams consisting of exactly two nodes of different colours can be removed using the scalar rule.
  Single-node scalars can be rewritten into two-node scalars as follows: let $a,b\in\{0,1\}$, and let $\bar{b}=b\oplus 1$, then:
   \begin{equation}\label{eq:spek_two-node_scalar}
    \scalar{gn, label={[gphase]right:$ab$}} \; = \; \innerprod{gn, label={[gphase]right:$a\bar{b}$}}{gn, label={[gphase]right:$01$}} \; = \; \innerprod{gn, label={[gphase]right:$a\bar{b}$}}{rn, label={[rphase]right:$01$}}
   \end{equation}
  by the spider law and \eqref{eq:spek_red_green_states}.
  Then the scalar rule can be applied.
  \item Any single-toy bit operator can be written as:
   \begin{equation}
    \input{tikz_files/spek_green-red-green.tikz}
   \end{equation}
   for some $a,b,c,d,e,f\in\{0,1\}$.
  \item A loop with a \HadSpek{} node in it disappears:
   \begin{equation}
    \input{tikz_files/spek_HadSpek_loop.tikz}
   \end{equation}
 \end{itemize}

 Scalar parts of a diagram can be decomposed into disconnected segments of at most two nodes each as in Corollary \ref{cor:decompose_scalars}.
 Any non-zero such segment can then be dropped by the scalar rule.

 Multiple copies of the zero scalar can be rewritten into just one copy:
 \begin{equation}
  \scalar{gn, label={[gphase]right:$11$}} \; \scalar{gn, label={[gphase]right:$11$}} \;
  = \; \innerprod{gn, label={[gphase]right:$11$}}{gn} \; \innerprod{gn, label={[gphase]right:$11$}}{gn} \;
  = \; \input{tikz_files/multiplication_11_der1.tikz} \;
  = \; \input{tikz_files/multiplication_11_der2.tikz} \;
  = \; \scalar{gn} \; \scalar{gn, label={[gphase]right:$11$}} \;
  = \; \scalar{gn, label={[gphase]right:$11$}},
 \end{equation}
 where the first step is by the spider rule, the second by the copy rule, the third by the $11$-copy rule, then the copy rule again, and the final step follows from \eqref{eq:spek_two-node_scalar} and the scalar rule. 

 Thus any state diagram can be brought into the desired form.
\end{proof}

Note that, as mentioned above, Corollary \ref{cor:decompose_scalars} translates to the toy theory graphical calculus; this is why the scalar rule is sufficient to ensure that all scalar diagrams can be rewritten to the empty diagram or to \scalar{gn, label={[gphase]right:$11$}}.
There is also a result analogous to Corollary \ref{cor:recognize_zero}:

\begin{cor}\label{cor:spek_recognize_zero}
 A diagram in the toy theory graphical calculus is zero if and only if it can be rewritten to explicitly contain \scalar{gn, label={[gphase]right:$11$}} as a subdiagram.
 Furthermore, it is straightforward to decide whether a diagram is zero by bringing the diagram into GS-LO form and simplifying all the scalars.
\end{cor}

This corollary allows us to focus on non-zero diagrams only in the next parts of the proof.

The GS-LO form is not unique, i.e.\ there may be different GS-LO diagrams representing the same state.
It is not clear how to define a unique normal form, but it is possible to reduce the number of diagrams needing to be considered further.

\begin{dfn}\label{dfn:rGS-LO}
 A diagram in Spekkens' toy theory is said to be in \emph{reduced GS-LO (or rGS-LO) form} if it is non-zero, in GS-LO form, and satisfies the following additional conditions:
 \begin{itemize}
  \item All vertex operators belong to the set:
   \begin{equation}\label{eq:reduced_vertex_operators_spek}
    R = \input{tikz_files/spek_set_R.tikz}
   \end{equation}
  \item Two adjacent vertices must not both have vertex operators that include red nodes.
 \end{itemize}
\end{dfn}

\begin{thm}\label{thm:rGS-LO}
 Any non-zero toy stabilizer state diagram is equal to some rGS-LO diagram within the graphical calculus.
\end{thm}

The proof is analogous to that of Theorem \ref{thm:ZX_rGS-LC}, using Lemma \ref{lem:single_toy_bit_operator}.

The following two propositions show that, as in the case of stabilizer QM, rGS-LO forms are not unique.

\begin{prop}\label{prop:rGS-LO_transformation1}
 Suppose a rGS-LO diagram contains a pair of neighbouring toy bits $p$ and $q$ in the following configuration, where $a,b\in\{0,1\}$:
 \begin{equation}
  \input{tikz_files/spek_rGS-LO_transformation1.tikz}
 \end{equation}
 Then a local complementation about $q$, followed by a local complementation about $p$, yields a diagram which can be brought into rGS-LO form by at most two applications of the fixpoint rule.
\end{prop}
\begin{proof}[Proof (sketch)]
 The effect of the local complementations on the vertex operators of $p$ and $q$ is the following:
 \begin{equation}
  \input{tikz_files/spek_rGS-LO_transformation1_proof.tikz}.
 \end{equation}
 If $a=1$, we apply a fixpoint operation to $p$ and if $b=1$, we apply a fixpoint operation to $q$; then the vertex operators of $p$ and $q$ are in $R$.
 The fixpoint operations add \phase{gn,label={[gphase]right:$11$}} to neighbouring toy bits, which maps the set $R$ to itself.
 As fixpoint operations do not change any edges, we do not have to worry about them when considering whether the rest of the diagram satisfies Definition \ref{dfn:rGS-LO}.

 The rest of the proof is analogous to the stabilizer QM case.
\end{proof}

\begin{prop}\label{prop:rGS-LO_transformation2}
 Suppose a rGS-LO diagram contains a pair of neighbouring toy bits $p$ and $q$ in the following configuration, where $a,b\in\{0,1\}$:
 \begin{equation}
  \input{tikz_files/spek_rGS-LO_transformation2.tikz}
 \end{equation}
 Then a local complementation along the edge $\{p,q\}$ yields a diagram which can be brought into rGS-LO form by at most two applications of the fixpoint rule.
\end{prop}
\begin{proof}
 After the local complementation along the edge, the vertex operator of $p$ is given by:
 \begin{equation}\label{eq:pivoting}
  \input{tikz_files/spek_vertex_operator_u.tikz}.
 \end{equation}
 For the vertex operator of $q$, we have:
 \begin{equation}
  \input{tikz_files/spek_rGS-LO_transformation2_proof.tikz}.
 \end{equation}
 Thus if $a$ or $b$ is 1, we apply a fixpoint operator to the appropriate vertex.
 From the properties of local complementations along edges it follows that the overall transformation preserves the two properties of rGS-LO states.
\end{proof}

With the definitions and results in this section, state diagrams in the toy theory graphical calculus can be simplified significantly.
By map-state duality, the results can be applied to arbitrary diagrams.

For completeness it remains to be shown that whenever two rGS-LO diagrams represent the same toy state, they can be rewritten into each other using the rewrite rules for the toy theory graphical calculus.

\subsection{Equalities between rGS-LO diagrams}
\label{s:spek_completeness_graphical}

Throughout this section, we consider non-zero diagrams only, which is possible by Corollary \ref{cor:spek_recognize_zero}.
The graphical calculus is complete for toy theory states if, given any two rGS-LO diagrams representing the same state, we can show that they are equal using the rules of the graphical calculus.
In this section, we exhibit an algorithm for rewriting two diagrams representing the same toy state to be identical.
As rewrite rules are invertible, this is equivalent to being able to rewrite one diagram into the other.
Again, the algorithm is similar to that for the stabilizer \ZX-calculus, cf.\ Section \ref{s:equality_testing}.

Given two toy state diagrams on the same number of toy bits, we start by pairing up red nodes between the two diagrams.

\begin{dfn}
 A pair of rGS-LO diagrams on the same number of toy bits is called \emph{simplified} if there are no pairs of toy bits $p,q$ such that $p$ has a red node in its vertex operator in the first diagram but not in the second, $q$ has a red node in the second diagram but not in the first, and $p$ and $q$ are adjacent in at least one of the diagrams.
\end{dfn}

\begin{prop}
 Any pair of rGS-LO diagrams on $n$ toy bits is equal to a simplified pair.
\end{prop}

The proof of the above proposition is analogous to the stabilizer QM case, Proposition \ref{prop:simplified}.

As in the \ZX-calculus, if there exist red nodes that cannot be paired up between the two diagrams, then the diagrams cannot represent the same state.

\begin{lem}\label{lem:unpaired_red}
 Consider a simplified pair of rGS-LO diagrams and suppose there exists an unpaired red node, i.e.\ there is a toy bit $p$ which has a red node in its vertex operator in one of the diagrams, but not in the other. Then the two diagrams are not equal.
\end{lem}
\begin{proof}\footnote{This proof closely follows that of Lemma \ref{lem:unpaired_red_node}. Nevertheless, as there are some differences and the details are complicated, we give it here in full.}
 Let $D_1$ be the diagram in which $p$ has the red node, $D_2$ the other diagram. There are multiple cases:

 \emph{In either diagram, $p$ has no neighbours}: In this case, the overall state factorises and the two diagrams are equal only if the two states of $p$ are the same. But:
 \begin{equation}
  \input{tikz_files/spek_distinct_single_states.tikz}
 \end{equation}
 for $a,b,c\in\{0,1\}$, so the diagrams must be unequal.

 \emph{$p$ is isolated in one of the diagrams but not in the other}: We argue in Section \ref{s:binary} that, as in stabilizer QM, two toy graph states with local operators are equal only if one can be transformed into the other via a sequence of local complementations with corresponding changes to the local operators.
 As a local complementation never turns a vertex with neighbours into a vertex without neighbours, or conversely, the two diagrams cannot be equal.

 \emph{$p$ has neighbours in both diagrams}: Without loss of generality, assume that $p$ is the first toy bit.
 Let $N_1$ be the set of all toy bits that are adjacent to $p$ in $D_1$, and define $N_2$ similarly.
 The vertex operators of any toy bit in $N_1$ must be green phases in both diagrams.
 In $D_1$, this is because of the definition of rGS-LO diagrams, in $D_2$ it is because the pair of diagrams is simplified.
 Suppose the original diagrams involve $n$ toy bits each.
 Let $G$ be the graph on $n$ vertices (named according to the same convention as in $D_1$ and $D_2$) whose edges are $\{\{p,v\} | v\in N_1\}$.
 Now consider the following diagram:
 \begin{equation}\label{eq:isolate_p}
  \input{tikz_files/spek_isolate_p.tikz}
 \end{equation}
 where the ellipse labelled $G$ denotes the toy graph state corresponding to $G$, except that each vertex in the graph has not only an output but also an input.
 Call this diagram $U$.
 It is straightforward to see that $U$ is invertible: composing it with itself upside-down yields the identity.
 Therefore composing this diagram with $D_1$ and $D_2$ yields two new diagrams which are equal if and only if $D_1=D_2$.
 We denote the new diagrams by $U\circ D_1$ and $U\circ D_2$ and show that, no matter what the properties of $D_1$ and $D_2$ are (beyond the existence of an unpaired red node on $p$):
 \begin{mitem}
  \item in $U\circ D_1$, the toy bit $p$ is in state \state{rn} or \state{rn,label={[rphase]right:$11$}};
  \item in $U\circ D_2$, $p$ is either entangled with other toy bits, or in one of the states \state{gn,label={[gphase]right:$ab$}}, where $a,b\in\{0,1\}$.
 \end{mitem}
 By the arguments used in the first two cases, this implies that $U\circ D_1\neq U\circ D_2$ and therefore $D_1\neq D_2$.

 Let $n=\abs{N_1}$, $m=\abs{N_1\cap N_2}$, and suppose the toy bits are arranged in such a way that the first $m$ elements of $N_1$ are those which are also elements of $N_2$, if there are any.
 Consider first the effect on diagram $D_1$.
 The local operator on $p$ combines with the single-toy bit operators from $U$ to:
 \begin{equation}
  \input{tikz_files/spek_local_operator_on_p.tikz},
 \end{equation}
 where $a\in\{0,1\}$.
 As green phase shifts can be pushed through other green nodes, the subdiagram involving $p$ and the elements of $N_1$ in $U\circ D_1$ is equal to:
 \begin{equation}
  \input{tikz_files/spek_U_after_D1_a.tikz} \; = \;\; \input{tikz_files/spek_U_after_D1_b.tikz}
 \end{equation}
 Here, $b_1,\ldots,b_n,c_1,\ldots,c_n\in\{0,1\}$.
 Note that at the end $p$ is isolated and in the state \state{rn,label={[rphase]right:$aa$}}.
 The fact that we have ignored all toy bits not originally adjacent to $p$ in $D_1$ does not change that.

 Next consider $U\circ D_2$.
 As $N_1$ is not in general equal to $N_2$, the subdiagram consisting of $p$ and vertices in $N_1$ looks as follows:
 \begin{equation}
  \input{tikz_files/spek_U_after_D2.tikz}
 \end{equation}
 where $l=m+1$ and $d,e,f_1,\ldots,f_n,g_1,\ldots,g_n\in\{0,1\}$.
 Note that we neglect edges that do not involve $p$ and also edges between $p$ and vertices not in $N_1$.
 We now distinguish different cases, depending on the values of $d$ and $e$.

 If $d=0$ and $e=1$, apply a local complementation about $p$.
 This does not change the edges incident on $p$:
 \[
  \input{tikz_files/spek_01_on_p_1.tikz} = \!\! \input{tikz_files/spek_01_on_p_2.tikz}
 \]
 \[
  = \!\!\!\!\!\!\! \input{tikz_files/spek_01_on_p_3.tikz} = \!\! \input{tikz_files/spek_01_on_p_4.tikz}
 \]
 \begin{equation}
  = \input{tikz_files/spek_01_on_p_5.tikz}
 \end{equation}
 Now if $N_1=N_2$, $p$ has no more neighbours and is in the state \state{rn,label={[rphase]right:$01$}}.
 This is not the same as the state $p$ has in diagram 1, so the diagrams are not equal.
 Else, after the application of $U$, $p$ still has some neighbours in diagram 2.
 Local complementations do not change this fact.
 Thus the two diagrams cannot be equal.
 The case $d=1,e=0$ is entirely analogous, except that there is a fixpoint operation in addition to the local complementation at the beginning.

 If $d=e=0$, there are two subcases.
 First, suppose there exists $v\in N_2$ such that $v\notin N_1$.
 Apply a local complementation about this $v$.
 This operation changes the vertex operator on $p$ to \phase{gn,label={[gphase]right:$01$}}.
 It also changes the edges incident on $p$, but the important thing is that $p$ still has at least one neighbour.
 Thus we can proceed as in the case $d=0,e=1$.

 Secondly, suppose there is no $v\in N_2$ which is not in $N_1$.
 Since $N_2\neq\emptyset$ ($N_2=\emptyset$ corresponds to the case ``$p$ has no neighbours in $D_2$'', which was considered above), we must then be able to find $v\in N_1\cap N_2$.
 The diagram looks as follows, where now $m>0$ (again, we are ignoring edges that do not involve $p$):
 \begin{equation}
  \input{tikz_files/spek_00_on_p_1.tikz} = \input{tikz_files/spek_00_on_p_2.tikz}
 \end{equation}
 To show that the two diagrams are unequal it suffices to show that in diagram 2 the state of $p$ either factors out, but is not \state{rn} or \state{rn,label={[rphase]right:$11$}}, or that it remains entangled with other toy bits.
 We are thus justified in ignoring large portions of the above diagram to focus only on $p$, $v$ and the edge between the two.
 In particular, we ignore for the moment the edges between $p$ and toy bits other than $v$, as well as the last \HadSpek{} on $p$.
 Then:
 \begin{center}
  \input{tikz_files/spek_p20_1a.tikz}
 \end{center}
 \begin{equation}
  \input{tikz_files/spek_p20_1b.tikz}
 \end{equation}
 where for the second equality we have applied a local complementation to $v$ and used the Euler decomposition, the third equality follows by a local complementation on $p$, and the last one comes from the merging of $p$ with the green node in the bottom left.
 Note that, in the end, $p$ and $v$ are still connected by an edge.
 None of the operations we ignored in picking out this part of the diagram can change that.
 Thus, as before, the state of $p$ cannot be the same as in diagram 1.
 The two diagrams are unequal.

 The case $d=e=1$ is analogous to $d=e=0$, except in either subcase we start with a fixpoint operation on the chosen $v$.

 We have thus shown that a simplified pair of rGS-LO diagrams are not equal if there are any unpaired red nodes.
\end{proof}

The existence of unpaired red nodes is not the only sign that a simplified pair of diagrams cannot be equal.
In fact, as in the \ZX-calculus, a simplified pair of diagrams are either identical or they do not represent the same state.

\begin{thm}\label{thm:state_completeness}
 The two diagrams making up a simplified pair of rGS-LO diagram are equal, i.e.\ they correspond to the same toy theory state, if and only if they are identical.
\end{thm}

The proof of this theorem is analogous to that of Theorem \ref{thm:rGS-LC_equality} for the stabilizer \ZX-calculus.

\subsection{A normal form for zero diagrams}

As in the \ZX-calculus, it is possible to define a unique normal form for zero diagrams in the toy theory graphical calculus.

\begin{thm}\label{thm:spek_zero_completeness}
 The toy theory graphical calculus is complete for zero diagrams.
\end{thm}
\begin{proof}[Proof (sketch).]
 By Corollary \ref{cor:spek_recognize_zero}, zero diagrams in the toy theory graphical calculus can be recognised.
 Now, analogously to the proof of Theorem \ref{thm:zero_completeness} in the \ZX-calculus, any zero diagram with $n$ inputs and $m$ outputs in the toy theory graphical calculus can be rewritten into the form:
 \begin{equation}\label{eq:spek_zero_normal_form}
  \input{tikz_files/spek_zero_normal_form.tikz}.
 \end{equation}
 Furthermore, this normal form is unique.
\end{proof}

The existence of a unique normal form immediately implies completeness for zero diagrams.

\subsection{Completeness}

By map-state duality for the toy theory, as given in Section \ref{s:toy_Choi}, and invertibility of the rewrite rules, Theorem \ref{thm:state_completeness} directly implies that the toy theory graphical calculus is complete for non-zero diagrams. Combining this with the normal form for zero diagrams in Theorem \ref{thm:spek_zero_completeness} yields: 

\begin{thm}
 The red-green calculus is complete for Spekkens' toy bit theory.
\end{thm}

Equalities between two diagrams in the toy theory graphical calculus can be derived as follows: if the diagrams are not states, bend all inputs around to become outputs.
Bring the two diagrams into GS-LO form.
If both diagrams are zero, bring them into the zero normal form.
Otherwise, bring the diagrams into rGS-LO form.
Simplify the pair of diagrams.
Then either the two diagrams are identical, in which case some of the rewrite steps can be inverted to get a sequence of rewrites transforming one diagram into the other, or they are not identical, in which case the two diagrams do not represent the same operator, so there is no equality to derive.

If the diagrams were not states to begin with, the appropriate outputs can be bent back into inputs in all diagram.
This yields a sequence of valid rewrites transforming one of the original diagrams into the other.

Thus, the maximal-knowledge fragment of Spekkens' toy bit theory can be analysed fully using this graphical calculus.

%% file: tikz_files/spekkens_transf_23.tikz
\begin{tikzpicture}[baseline=0.7,fill=teal,every path/.style={draw}]
  \draw (0,0) rectangle +(.8,.8);
  \draw (.8,0) rectangle +(.8,.8);
  \draw (0,.8) rectangle +(.8,.8);
  \draw (.8,.8) rectangle +(.8,.8);
  \draw[thick,<->] (1.2,.4) -- (.4,1.2);
\end{tikzpicture}

%% file: conclusions.tex
\chapter{Conclusions and further work}
\label{ch:conclusions}

In this thesis, I have shown that for many applications in quantum computing and quantum foundations, intutive and powerful graphical languages can be used without loss of mathematical rigour.
Here, the notion of being ``powerful and rigorous'' is captured by the property of completeness, meaning that any equality that can be derived using conventional formalisms can also be derived graphically.

I have proved that the \ZX-calculus, a graphical language for pure state qubit QM, is complete for stabilizer quantum theory.
This means that, within stabilizer QM, any true equality can be derived using the rewrite rules of the \ZX-calculus.
Furthermore, measurement amplitudes and probabilities can be calculated entirely graphically.

I have also shown that the \ZX-calculus is complete for the single-qubit Clifford+T group.
This group of operations is approximately universal, i.e.\ any single-qubit unitary can be approximated to arbitrary accuracy using only Clifford unitaries and the T gate.
In most physical implementations of quantum computers, general unitary operators cannot be directly applied but instead need to be approximated using some finite set of operators like Clifford+T.
Thus this completeness result implies that the \ZX-calculus can be used to analyse a wide range of realistic problems in quantum computation.

Finally, I have shown that similar graphical languages can replace conventional formalisms even outside quantum theory.
I have defined a graphical calculus for the maximal knowledge fragment of Spekkens' toy bit theory, a local hidden variable model that behaves very similar to pure state stabilizer QM, and shown that this graphical calculus is universal, sound, and complete.
This means that the graphical calculus has the full power of any formalism for analysing the toy theory.
The toy theory graphical calculus is modelled after the \ZX-calculus, therefore similarities and differences between stabilizer QM and the toy bit theory can be fully explored using analogous graphical methods.

A number of further research directions arise from this work.

\section{Further work: automated graphical reasoning}

Graphical languages are amenable to automated reasoning.
The software system Quantomatic \cite{kissinger_quantomatic_2015,quantomatic} discussed in Section \ref{s:automated_graphical_reasoning} enables automated and semi-automated manipulation of diagrams in the \ZX-calculus and similar graphical languages.
It would be interesting to implement the normalisation and equality testing algorithms from this thesis in that system, thus allowing automated simplification and comparison of diagrams for stabilizer QM and the single-qubit Clifford+T group, as well as Spekkens' toy bit theory.
Then, many questions in areas such as error-correcting codes or measurement-based quantum computation could be analysed automatically.
Quantomatic could also be used to compile single-qubit unitaries into the Clifford+T gate set, or simplify such approximations.
Furthermore, Quantomatic could automatically explore similarities and differences between stabilizer QM and Spekkens' toy bit theory.

This immediately offers a new question, namely that of the computational complexity of the equality decision problems.

\section{Further work on \ZX-calculus completeness}

The \ZX-calculus as defined in Section \ref{s:ZX-calculus} is incomplete for general pure state qubit quantum mechanics, but it is complete for pure state qubit stabilizer quantum mechanics as well as for the single-qubit Clifford+T group.

An obvious next step is to attempt to combine the two existing completeness results into a completeness proof for multi-qubit Clifford+T operators.
Such a result is not precluded by the incompleteness proof for the general \ZX-calculus -- cf.\ Section \ref{s:possible_completeness} -- but neither does it follow straightforwardly from the existing completeness proofs.

For example, as Perdrix and Wang recently showed, making the \ZX-calculus complete for multi-qubit Clifford+T group requires the addition of at least one new rule \cite{perdrix_zx_2015}, the \emph{supplementarity rule} first introduced in \cite{coecke_three_2011} and given here in correctly scaled form:
\begin{equation}\label{eq:supplementarity}
 \innerprod{gn}{rn} \; \innerprod{gn}{rn} \; \input{tikz_files/supplementary.tikz}
\end{equation}
for any $\alpha\in(-\pi,\pi]$.
The special cases of the supplementarity rule where $\alpha$ is an integer multiple of $\pi/2$ can be derived from the rewrite rules given in Section \ref{s:rewrite_rules}, thus the supplementarity rule is not required for stabilizer completeness.
Furthermore, as the LHS of \eqref{eq:supplementarity} is not a line graph, the rule does not apply in the context of the single-qubit Clifford+T diagrams considered in Section \ref{s:Clifford+T_completeness}.

It is still unclear whether the addition of the supplementarity rule is sufficient to make the \ZX-calculus complete for multi-qubit Clifford+T operators \cite{perdrix_zx_2015}.

There exists a presentation of the two-qubit Clifford+T group in terms of generators and relations \cite{greylyn_generators_2014}, which -- if translated into the \ZX-calculus -- would give a completeness result for two-qubit diagrams.
Yet the distinction between two-qubit diagrams and multi-qubit diagrams is not very natural in the \ZX-calculus, and several of the generators used in the derivation of that result only have complicated representations in the \ZX-calculus, making the translation of that proof difficult.

\section{Further work on the graphical calculus for Spekkens' toy theory}

In this thesis, only pure state qubit stabilizer quantum mechanics and the maximal knowledge fragment of Spekkens' toy bit theory have been considered.
An obvious next step would be to extend the graphical calculi to mixed states in the quantum case and states of less-than-maximal knowledge in the toy theory.
The category-theoretical formulations underlying the graphical calculi can easily be extended in this way using the CPM-construction and these extensions carry over to categorical graphical calculi \cite{selinger_dagger_2007}.

Furthermore, it would be interesting to extend this argument to stabilizer quantum mechanics for higher dimensional systems and the higher-dimensional toy theory.
Some steps in this direction have been made by generalising the \ZX-calculus to qudits and to Spekkens' toy theory for systems of dimension greater than two, though it is still unclear whether these graphical languages are complete \cite{ranchin_depicting_2014}.

Rigorous graphical languages have many applications in the analysis of quantum physics and related theories.